\documentclass[11pt]{article}

\usepackage[active]{srcltx}
\usepackage{hyperref}
\usepackage{amsmath,amsthm,amscd,amsfonts,amssymb}
\usepackage{color}
\usepackage[all]{xy}

 1 
\newtheorem{theorem}{Theorem}
\newtheorem{proposition}{Proposition}
\newtheorem{corollary}{Corollary}
\newtheorem{lemma}{Lemma}
\newtheorem{definition}{Definition}

\def\beq{\begin{equation}}
\def\eeq{\end{equation}}
\def\bea{\begin{eqnarray}}
\def\eea{\end{eqnarray}}
\def\beann{\begin{eqnarray*}}
\def\eeann{\end{eqnarray*}}
\def\beasn{\begin{sneqnarray}}
\def\eeasn{\end{sneqnarray}}
\def\ben{\begin{enumerate}}
\def\een{\end{enumerate}}
\def\bit{\begin{itemize}}
\def\eit{\end{itemize}}

\def\derpar#1#2{\displaystyle\frac{\partial{#1}}{\partial{#2}}}
\def\derpars#1#2#3{\displaystyle\frac{\partial^2{#1}}{\partial{#2}\partial{#3}}}
\def\restric#1#2{\left.#1\right|_{#2}}

\def\W{{\cal W}}
\def\C{{\cal C}}
\def\P{{\cal P}}
\def\X{{\cal X}}

\def\vf{\mathfrak{X}}

\def\Lag{{\cal L}}
\def\Leg{{\cal FL}}

\def\Lie{\mathop\textnormal{L}\nolimits}

\def\tabaddress#1{{\small\it\begin{tabular}[t]{c}#1
\\[1.2ex]\end{tabular}}}

\title{A NEW MULTISYMPLECTIC UNIFIED FORMALISM FOR SECOND ORDER
CLASSICAL FIELD THEORIES}

\author{
{\sc  Pedro Daniel Prieto-Mart\'\i nez\thanks{{\bf e}-{\it mail}:
   peredaniel@ma4.upc.edu} }\\
   {\sc Narciso Rom\'an-Roy\thanks{{\bf e}-{\it mail}:
   nrr@ma4.upc.edu}}  \\
   \tabaddress{Departamento de Matem\'atica Aplicada IV.
   Edificio C-3, Campus Norte UPC\\
   C/ Jordi Girona 1. 08034 Barcelona. Spain}}
   \date{\today \\
   }

\begin{document}

\maketitle

\pagestyle{myheadings}

\thispagestyle{empty}

\begin{abstract}
We present a new multisymplectic framework for second-order classical field theories which is based
on an extension of the unified Lagrangian-Hamiltonian formalism to these kinds of systems. This
model provides a straightforward and simple way to define the Poincar\'e-Cartan form and clarifies
the construction of the Legendre map (univocally obtained as a consequence of the constraint
algorithm). Likewise, it removes the undesirable arbitrariness in the solutions to the field
equations, which are analyzed in-depth, and written in terms of holonomic sections and multivector
fields. Our treatment therefore completes previous attempt to achieve this aim. The formulation is
applied to describing some physical examples; in particular, to giving another alternative
multisymplectic description of the Korteweg-de Vries equation.
\end{abstract}


\vbox{\raggedleft AMS s.\,c.\,(2010): \null  70H50, 70S05, 35G99, 53D42, 55R10}\null
\markright{\textnormal{P.D. Prieto-Mart\'\i nez, N. Rom\'an-Roy:  
 \sl Unified formalism for second order field theories.}}

\clearpage

\tableofcontents

\section{Introduction}
\label{sec:Introduction}

Higher-order field theories are relevant in physics and  applied mathematics because they appear
in many of important situations; for instance, the standard gravitational theories, in particular
Hilbert's Lagrangian for gravitation, are of this kind; as well as string theories, Podolsky's
generalization  of electromagnetism,  the different forms of the Korteweg-de Vries equation in
fluid theory, and other interesting models in physics. As a consequence,  many works are devoted
to the development of a formalism for these kinds of theories and their application to many models
in mechanics and field theory (a long but non-exhaustive list of references can be found in
\cite{art:Batlle_Gomis_Pons_Roman88,art:Prieto_Roman11,art:Prieto_Roman12_1}).

In higher-order mechanical systems and field theories, the formalism shows explicit dependence
on accelerations or higher-order derivatives of the generalized coordinates of position, or in
the higher-order derivatives of the fields. Thus, for Lagrangian systems, if the Lagrangian
function depends on derivatives of order $k$, the corresponding Euler-Lagrange equations are of
order $2k$. These kinds of systems are therefore modeled geometrically using higher-order tangent
and jet bundles as the main tool  (see, for instance,
\cite{proc:Cantrijn_Crampin_Sarlet86,unpub:Crampin_Saunders11,
book:DeLeon_Rodrigues85,proc:DeLeon_Rodrigues87,art:Gracia_Pons_Roman91,
art:Gracia_Pons_Roman92,art:Krupkova00,unpub:Mukherjee_Paul11,
book:Saunders89,art:Saunders_Crampin90}).

In particular, as regards higher-order field theories, great efforts have been made to extend
the classical multisymplectic framework developed for describing first-order field theories to this
realm. The usual way to do this consists in generalizing the construction of the
\emph{Poincar\'e-Cartan form} for a higher-order Lagrangian density and then stating the Lagrangian
formalism \cite{art:Aldaya_Azcarraga78_2,art:Aldaya_Azcarraga78_1,proc:Ferraris_Francaviglia86,
proc:Garcia_Munoz83,proc:Garcia_Munoz91,procs:Grigore2009,art:Horak-Kolar83,art:Saunders87}.
Nevertheless, this procedure involves some ambiguity, since the definition of the
Poincar\'e-Cartan form in a higher-order jet bundle is not unique, and despite that for the second-order
case it is proved that all these forms are equivalent \cite{book:Saunders89,art:Saunders_Crampin90},
this is not true for the general higher-order cases. These and other kinds of problems involving
the non-uniqueness of the geometrical constructions also appear in the definition of the Legendre
transformation associated with a higher-order Lagrangian and as well as a suitable choice of the
multimomentum phase space for the Hamiltonian formalism of the theory
\cite{art:Aldaya_Azcarraga80,art:Francaviglia_Krupka82,art:Kolar84,proc:Krupka84}.

A way to overcome these difficulties and simplify the formalism was recently achieved in
\cite{art:Campos_DeLeon_Martin_Vankerschaver09} using the so-called \emph{Skinner-Rusk} or
\emph{Lagrangian-Hamiltonian unified formalism} for field theories. The origin of this formalism
is the seminal paper \cite{art:Skinner_Rusk83}, where R. Skinner and R. Rusk present a new framework
for first-order autonomous mechanical systems  that compresses the Lagrangian and Hamiltonian formalisms
into a single one. This was subsequently generalized to first-order non-autonomous dynamical systems
\cite{art:Barbero_Echeverria_Martin_Munoz_Roman08,art:Cortes_Martinez_Cantrijn02}, control systems
\cite{art:Barbero_Echeverria_Martin_Munoz_Roman07}, higher-order autonomous and non-autonomous
mechanical systems \cite{proc:Cantrijn_Crampin_Sarlet86,art:Colombo_Martin_Zuccalli10,
book:DeLeon_Rodrigues85,art:Gracia_Pons_Roman91,book:Miron10,art:Prieto_Roman11,
art:Prieto_Roman12_1,art:Prieto_Roman14}, and first-order classical field theories
\cite{art:DeLeon_Marrero_Martin03,art:Echeverria_Lopez_Marin_Munoz_Roman04,RRS-2005,RRSV-2011}.
Then, in \cite{art:Campos_DeLeon_Martin_Vankerschaver09} the authors present an extension of this
formulation to higher-order field theories in order to develop an unambiguous framework for
higher-order classical field theories. While this model allows us to simplify previous formulations,
some arbitrary parameters appearing in the solutions of the higher-order field equations
and in the definition of
the Legendre transformation must be fixed ``ad-hoc''. Another interesting approach to the
higher-order unified formalism for field theory, but using infinite-order jet bundles, is given in
\cite{art:Vitagliano10}.

In this paper, we present a modification of the model given in
\cite{art:Campos_DeLeon_Martin_Vankerschaver09} by using finite higher-order bundles to overcome
some of the ambiguities in the solutions of the equations given by the model, thus clarifying the
construction of the Legendre map and the choice of the jet and the multimomentum bundles for the
Lagrangian and the Hamiltonian formalisms, as well as the field equations in both formalisms.
Our model is therefore a completion of the approaches given in
\cite{art:Campos_DeLeon_Martin_Vankerschaver09,art:Vitagliano10}.
Our treatment works for second-order field theories because we want it to be applied here and in
future papers to describe the well known theories previously cited: gravitation, Korteweg-de Vries
equation and other models in physics, all of which are of second-order. Another advantage of
working at this order is that we can use the diffeomorphism among several geometric structures
in order to avoid part of the ambiguity inherent to the theory. In any case, further work to
generalize our results to higher-order cases is in progress.

The organization of the paper is as follows. First, in Section \ref{sec:MathBackground}, we review
the geometric structures of higher-order jet bundles, introduce the concepts of holonomic sections
and multivector fields in order to state the field equations on these bundles, and define the space
of symmetric multimomenta suitable for the Hamiltonian formalism. Section \ref{sec:LagHamFormalism}
is devoted to developing our proposal of the Lagrangian-Hamiltonian unified formalism for
second-order field theories. After introducing the unified jet-multimomentum bundles and their
relevant submanifolds where the formalism takes place, we state the field equations in the unified
formalism using sections and multivector fields. Thanks to this unified framework, we establish the
Lagrangian and Hamiltonian formalisms for second-order field theories (in Sections
\ref{sec:LagrangianForm} and \ref{sec:HamiltonianForm}) for both the regular and singular
(almost-regular) cases. Finally, in Section \ref{sec:Examples} we apply our formulation to describe
an academic model: a first-order Lagrangian as a second-order one, and two physical systems:
the bending or deflection of a plate with clamped edges and the classical
Korteweg-de Vries equation. A comparison of our results with those of previous papers is given
in the last Section \ref{sec:Conclusions}, where we also summarize our results and outlook.

All the manifolds are real, second countable and smooth ($C^\infty$). The maps and the structures are assumed
to be $C^\infty$.  Sum over repeated indices is understood. The usual multi-index notation introduced
in \cite{book:Saunders89} is used: a multi-index $I$ is an element of $\mathbb{Z}^m$ such that every
component is positive, the $i$th position of the multi-index is denoted $I(i)$, and
$|I| = \sum_{i=1}^{m} I(i)$ is the length of the multi-index, while $I! = \prod_{i=1}^{m} I(i)!$.
Finally, an expression of the type $|I| = k$ means that the expression (or the sum) is taken for
every multi-index of length $k$. The same applies for inequalities. (See \cite{book:Saunders89},
\S 6.1 for details).

\section{Geometric structures of higher-order jet bundles}
\label{sec:MathBackground}

\subsection{Higher-order jet bundles. Coordinate total derivatives}
\label{sec:HOJetBundles}

(See \cite{book:Saunders89} for details).

Let $M$ be an orientable $m$-dimensional smooth manifold, and let $\eta \in \Omega^m(M)$ be a
volume form for $M$. Let $E \stackrel{\pi}{\longrightarrow} M$ be a bundle with $\dim E = m + n$.
If $k \in \mathbb{N}$, the \emph{$k$th-order jet bundle} of the projection $\pi$, $J^k\pi$, is
the manifold of the $k$-jets of local sections $\phi \in \Gamma(\pi)$; that is, equivalence classes
of local sections of $\pi$ by the relation of equality on every partial derivative up to order $k$.
A point in $J^k\pi$ is denoted by $j^k_x\phi$, where $x \in M$ and $\phi \in \Gamma(\pi)$ is a
representative of the equivalence class. We have the following natural projections: if $r \leqslant k$,
$$
\begin{array}{rcl}
\pi^k_r \colon J^k\pi & \longrightarrow & J^r\pi \\
j^k_x\phi & \longmapsto & j^r_x\phi
\end{array}
\quad ; \quad
\begin{array}{rcl}
\pi^k \colon J^k\pi & \longrightarrow & E \\
j^k_x\phi & \longmapsto & \phi(x)
\end{array}
\quad ; \quad
\begin{array}{rcl}
\bar{\pi}^k \colon J^k\pi & \longrightarrow & M \\
j^k_x\phi & \longmapsto & x
\end{array}
$$
Observe that $\pi^s_r\circ\pi^k_s = \pi^k_r$, $\pi^k_0 = \pi^k$ (where $J^0\pi$ is canonically
identified with $E$), $\pi^k_k = \textnormal{Id}_{J^k\pi}$, and $\bar{\pi}^k = \pi \circ \pi^k$.

Local coordinates in $J^k\pi$ are introduced as follows: let $(x^i)$, ($1 \leqslant i \leqslant m$)
be local coordinates in $M$, and $(x^i,u^\alpha)$, ($1 \leqslant \alpha \leqslant n$), local
coordinates in $E$ adapted to the bundle structure. Let $\phi \in \Gamma(\pi)$ be a section with
coordinate expression $\phi(x^i) = (x^i,\phi^\alpha(x^i))$. Then, local coordinates in $J^k\pi$ are
$(x^i,u^\alpha,u_I^\alpha)$, where
$$
u^\alpha = \phi^\alpha \quad ; \quad u_I^\alpha =
\frac{\partial^{|I|}\phi^\alpha}{\partial x^I} \quad (1 \leqslant |I| \leqslant k) \, .
$$

Using these coordinates, the local expressions of the natural projections are
$$
\pi^k_r(x^i,u^\alpha,u_I^\alpha) = (x^i,u^\alpha,u_J^\alpha) \ ; \
\pi^k(x^i,u^\alpha,u_I^\alpha) = (x^i,u^\alpha) \ ; \
\bar{\pi}^k(x^i,u^\alpha,u_I^\alpha) = (x^i) \, .
$$

If $\phi \in \Gamma(\pi)$, we denote the \emph{$k$th prolongation} of $\phi$ to
$J^k\pi$ by $j^k\phi \in \Gamma(\bar{\pi}^k)$. In natural coordinates of $J^k\pi$, if
$\phi(x^i) = (x^i,\phi^\alpha(x^i))$, its $k$th prolongation is given by
$$
j^k\phi(x^i) = \left( x^i,\phi^\alpha,\frac{\partial^{|I|}\phi^\alpha}{\partial x^{I}} \right) \, ,
\quad 1 \leqslant |I| \leqslant k \, .
$$

\begin{definition}
Let $E \stackrel{\pi}{\longrightarrow} M$ be a bundle, $x \in M$, $\phi \in \Gamma(\pi)$
a section in $x$, and $v \in T_xM$. The \emph{$k$th holonomic lift} of $v$ by $\phi$
is defined as
$$
((j^k\phi)_*(v),j^{k+1}_x\phi) \in (\pi^{k+1}_{k})^*T J^k\pi \, .
$$
\end{definition}

In coordinates, if $v \in T_xM$ is given by $v = v^i\restric{\frac{\partial}{\partial x^i}}{x}$,
its $k$th holonomic lift is
\begin{equation}\label{eqn:LocalCoordHolonomicLiftingTangentVectors}
(j^k\phi)_*(v) = v^i\left( \restric{\derpar{}{x^i}}{j_x^{k}\phi}
+ \sum_{|I|=0}^{k}\restric{u_{I+1_i}^\alpha(j_x^{k+1}\phi)\derpar{}{u_{I}^\alpha}}{j_x^k\phi} \right) \, .
\end{equation}

The vector space $(\pi^{k+1}_{k})^*(T J^k\pi)_{j^{k+1}_x\pi}$ has a canonical splitting as a
direct sum of two subspaces:
$$
(\pi^{k+1}_{k})^*(T J^k\pi)_{j^{k+1}_x\phi} =
(\pi^{k+1}_{k})^*(V(\bar{\pi}^k))_{j^{k+1}_x\phi}
\oplus (j^k\phi)_*(T_xM) \, ,
$$
where $(j^k\phi)_*T_xM$ denotes the set of $k$th holonomic lifts of tangent vectors in $T_xM$
by $\phi$. As a consequence, the vector bundle
$(\pi^{k+1}_{k})^*\tau_{J^k\pi} \colon (\pi^{k+1}_{k})^*T J^k\pi \to J^k\pi$
has a canonical splitting as a direct sum of two subbundles
$$
\xymatrix{
(\pi^{k+1}_{k})^*T J^k\pi = (\pi^{k+1}_{k})^*V(\bar{\pi}^k) \oplus
H(\pi^{k+1}_{k}) \ar[rr]^-{(\pi^{k+1}_{k})^*\tau_{J^k\pi}}
& \ & J^k\pi
} \, ,
$$
where $H(\pi^{k+1}_k)$ is the union of the fibres $(j^k\phi)_*(T_xM)$, for $x \in M$.

Now, if $\vf(\pi^{k+1}_{k})$ denotes the module of vector fields along the projection
$\pi^{k+1}_{k}$, the submodule corresponding to sections of
$\restric{(\pi^{k+1}_{k})^*\tau_{J^{k}\pi}}{(\pi^{k+1}_{k})^*V(\bar{\pi}^k)}$
is denoted by $\vf^v(\pi^{k+1}_{k})$, and the submodule corresponding to sections of
$\restric{(\pi^{k+1}_{k})^*\tau_{J^{k}\pi}}{H(\pi^{k+1}_{k})}$ is denoted by $\vf^h(\pi^{k+1}_{k})$.
The splitting for the bundles given above induces the following canonical splitting for the module
$\vf(\pi^{k+1}_{k})$:
$$
\vf(\pi^{k+1}_{k}) = \vf^v(\pi^{k+1}_{k}) \oplus \vf^h(\pi^{k+1}_{k}) \, .
$$
An element of the submodule $\vf^h(\pi^{k+1}_{k})$ is called a \emph{total derivative}.

\begin{definition}
Given a vector field $X \in \vf(M)$, a section $\phi \in \Gamma(\pi)$ and a point $x \in M$,
the \emph{$k$th holonomic lift} of $X$ by $\phi$, $j^kX \in \vf^{h}(\pi^{k+1}_{k})$,
is defined as
$$
(j^kX)_{j^{k+1}_x\phi} = (j^k\phi)_*(X_{x}) \, .
$$
\end{definition}

In local coordinates, if $X \in \vf(M)$ is given by $\displaystyle X = X^i\derpar{}{x^i}$,
then, bearing in mind the local expression \eqref{eqn:LocalCoordHolonomicLiftingTangentVectors}
of the $k$th holonomic lift for tangent vectors, the $k$th holonomic lift of $X$ is
$$
j^kX = X^i\left( \derpar{}{x^i} + \sum_{|I|=0}^{k}u_{I+1_i}^\alpha \derpar{}{u_I^\alpha} \right) \, .
$$

Finally, the \emph{coordinate total derivatives} are the holonomic lifts of the local vector
fields $\partial / \partial x^i \in \vf(M)$, which are denoted by $d / dx^i \in \vf(\pi^{k+1}_{k})$,
and whose coordinate expressions are
$$
\frac{d}{dx^i} = \derpar{}{x^i} + \sum_{|I|=0}^{k} u_{I+1_i}^\alpha \derpar{}{u_I^\alpha} \, ,
\quad 1 \leqslant i \leqslant m \, .
$$

\subsection{Holonomic sections and multivector fields}
\label{sec:Holonomy&Semisprays}

(See appendix \ref{sec:MultiVF} for the terminology and notation on multivector fields in a manifold).

\begin{definition}
A section $\psi \in \Gamma(\bar{\pi}^k)$ is \emph{holonomic of type $r$}
($1 \leqslant r \leqslant k$) if $j^{k-r+1}\phi = \pi^{k}_{k-r+1} \circ \psi$, where
$\phi = \pi^k \circ \psi \in \Gamma(\pi)$; that is, the section $\pi^{k}_{k-r+1} \circ \psi$
is the prolongation to the jet bundle $J^{k-r+1}\pi$ of a section $\phi \in \Gamma(\pi)$.
In particular, a section $\psi$ is \emph{holonomic of type $1$} (or simply
\emph{holonomic}) if $j^k(\pi^{k} \circ \psi) = \psi$; that is, $\psi$ is the $k$th
prolongation of a section $\phi = \pi^{k} \circ \psi \in \Gamma(\pi)$.
\end{definition}

The commutative diagram that illustrates the previous definition is the following
$$
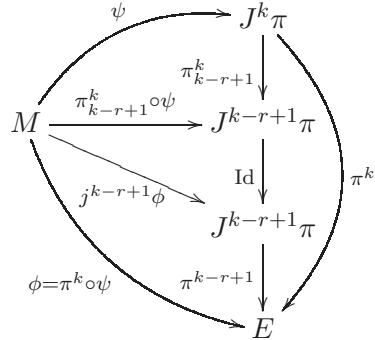
\xymatrix{
\ & \ & J^k\pi \ar[d]_{\pi^k_{k-r+1}} \ar@/^2.5pc/[ddd]^{\pi^k} \\
M \ar@/^1.5pc/[urr]^{\psi} \ar@/_1.5pc/[ddrr]_{\phi = \pi^k\circ\psi}
\ar[rr]^-{\pi^k_{k-r+1}\circ\psi} \ar[drr]_{j^{k-r+1}\phi}
& \ & J^{k-r+1}\pi \ar[d]_{\textnormal{Id}} \\
\ & \ & J^{k-r+1}\pi \ar[d]_{\pi^{k-r+1}} \\
\ & \ & E
}
$$

In the natural coordinates of $J^k\pi$, if $\psi \in \Gamma(\bar{\pi}^k)$ is given by
$\psi(x^i) = (x^i,\psi^\alpha,\psi_I^\alpha)$ ($1 \leqslant |I| \leqslant k$),
then the condition for $\psi$ to be holonomic of type $r$ gives the system of partial
differential equations
\begin{equation}\label{eqn:HolonomyConditionSect1}
\psi_{I}^\alpha = \frac{\partial^{|I|} \psi^\alpha}{\partial x^{I}} \, ,
\qquad 1 \leqslant |I| \leqslant k-r+1 \, , \ 1 \leqslant \alpha \leqslant n \, ,
\end{equation}
or, equivalently,
\begin{equation}\label{eqn:HolonomyConditionSect2}
\psi_{I+1_i}^\alpha = \derpar{\psi_I^\alpha}{x^i} \, ,
\qquad 1 \leqslant |I| \leqslant k-r \, , \ 1 \leqslant i \leqslant m \, ,
\ 1 \leqslant \alpha \leqslant n \, .
\end{equation}

\begin{definition}
A multivector field $\X \in \vf^m(J^{k}\pi)$ is \emph{holonomic of type $r$},
with $1 \leqslant r \leqslant k$, if the following conditions are satisfied:
\begin{enumerate}
\item $\X$ is integrable.
\item $\X$ is $\bar{\pi}^k$-transverse.
\item The integral sections $\psi \in \Gamma(\bar{\pi}^k)$ of $\X$ are holonomic of type $r$.
\end{enumerate}
In particular, a multivector field $\X \in \vf^{m}(J^{k}\pi)$ is
\emph{holonomic of type $1$} (or simply \emph{holonomic}) if it is integrable,
$\bar{\pi}^{k}$-transverse and its integral sections $\psi \in \Gamma(\bar{\pi}^k)$ are the
$k$th prolongations of sections $\phi \in \Gamma(\pi)$.
\end{definition}

In natural coordinates, if $\X \in \vf^{m}(J^k\pi)$ is a locally decomposable
and $\bar{\pi}^k$-transverse multivector field locally given by
$$
\X = \bigwedge_{i=1}^{m} f_i
\left(  \derpar{}{x^i} + F_i^\alpha\derpar{}{u^\alpha} + F_{I,i}^\alpha\derpar{}{u_I^\alpha} \right) \, ,
\quad (1 \leqslant |I| \leqslant k) \, ,
$$
with $f_i$ non-vanishing local functions. Then, the condition for $\X$ to be holonomic of type $r$
gives the following equations:
\begin{equation}\label{eqn:MultiVFHolonomyLocal}
F_i^\alpha = u_i^\alpha \quad ; \quad
F_{I,i}^\alpha = u_{I+1_i}^\alpha \, , \quad 1 \leqslant |I| \leqslant k-r \, , \
1 \leqslant i \leqslant m \, , \ 1 \leqslant \alpha \leqslant n \, .
\end{equation}
Hence, the local expression of a locally decomposable holonomic multivector field of type $r$ is
$$
\X = \bigwedge_{i=1}^{m} f_i
\left(  \derpar{}{x^i} + u_i^\alpha\derpar{}{u^\alpha} + \sum_{|I|=1}^{k-r} u_{I+1_i}^\alpha\derpar{}{u_I^\alpha}
+ \sum_{|I|=k-r+1}^{k} F_{I,i}^\alpha\derpar{}{u_I^\alpha}\right) \, ,
$$
In the particular case $r=1$, the local expression is
$$
\X = \bigwedge_{i=1}^{m} f_i
\left(  \derpar{}{x^i} + u_i^\alpha\derpar{}{u^\alpha} + \sum_{|I|=1}^{k-1} u_{I+1_i}^\alpha\derpar{}{u_I^\alpha}
+ F_{K,i}^\alpha\derpar{}{u_K^\alpha}\right) \, , \quad |K| = k \, .
$$

\noindent\textbf{Remark:}
It is important to point out that a locally decomposable and $\bar{\pi}^k$-transverse multivector
field $\X$ satisfying the local equations \eqref{eqn:MultiVFHolonomyLocal} may not be
holonomic of type $r$, since these local equations are not a sufficient or necessary condition for
the multivector field to be integrable. However, we can assure that if such a multivector field
admits integral sections, then its integral sections are holonomic of type $r$. In first-order
theories, these equations are equivalent to the so-called \emph{semi-holonomy (or SOPDE) condition}
\cite{art:Echeverria_Munoz_Roman98}.

\subsection{The space of 2-symmetric multimomenta}
\label{sec:SymmetricMultimomenta}

For the sake of simplicity, in the following we restrict ourselves to the case $k=2$, that is, the
second-order case, which is our main goal in this paper. However, all the results that follow in
this Section can be stated for an arbitrary value of $k$ (see \cite{phd:Campos} for details).

Following \cite{proc:Carinena_Crampin_Ibort91,art:Echeverria_DeLeon_Munoz_Roman07,
art:Echeverria_Munoz_Roman00_JMP}, let us consider $\Lambda_2^m(J^{1}\pi)$ as the phase space for
the Hamiltonian formalism of a second-order field theory; that is, the bundle of $m$-forms over
$J^{1}\pi$ vanishing by the action of two $\bar{\pi}^{1}$-vertical vector fields. We have the
following canonical projections:
$$
\pi_{J^{1}\pi} \colon \Lambda_2^m(J^{1}\pi) \to J^{1}\pi \quad ; \quad
\bar{\pi}_{J^{1}\pi} = \bar{\pi}^{1} \circ \pi_{J^{1}\pi} \colon \Lambda_2^m(J^{1}\pi) \to M \, .
$$

This bundle is endowed with some canonical structures. First, we define the
\emph{tautological (or Liouville) $m$-form} on $\Lambda_2^m(J^{1}\pi)$ by
$$
\Theta_{1}(\omega)(X_1,\ldots,X_m) =
\omega(T\pi_{J^{1}\pi}(X_1),\ldots,T\pi_{J^{1}\pi}(X_m)) \, ,
$$
where $\omega \in \Lambda_2^m(J^{1}\pi)$, and $X_1,\ldots,X_m \in T_{\omega}\Lambda_2^m(J^{1}\pi)$.
Then, we can define a multisymplectic $(m+1)$-form $\Omega_{1} \in \Omega^{m+1}(\Lambda_2^m(J^{1}\pi))$
as $\Omega_{1} = -d\Theta_{1}$, which is called the \emph{canonical (or Liouville) multisymplectic
$(m+1)$-form} on $\Lambda_2^m(J^{1}\pi)$. Recall that a multisymplectic $k$-form in a $n$-dimensional
manifold $N$ is a closed $k$-form $\Omega$ (with $1 \leqslant k \leqslant n$) which is $1$-nondegenerate;
that is, for $p\in N$,we have that $i(X_p)\Omega_p = 0$ if, and only if, $X_p= 0$.

In addition, the bundle $\Lambda_2^m(J^{1}\pi)$ is diffeomorphic to the union of the affine maps
from $J^1_u\bar{\pi}^{1}$ to $(\Lambda^mM)_{\bar{\pi}^{1}(u)}$, where $u \in J^{1}\pi$ is an
arbitrary point; that is,
$$
\Lambda_2^m(J^{1}\pi) \cong \bigcup_{u \in J^{1}\pi}
 \textnormal{Aff}(J^1_u\bar{\pi}^{1},(\Lambda^mM)_{\bar{\pi}^{1}(u)}) \, .
$$
Using this identification and the fact that $J^2\pi$ is embedded into $J^1\bar{\pi}^{1}$, we can
define a canonical pairing between the elements of $J^2\pi$ and the elements of
$\Lambda_2^m(J^{1}\pi)$ as a fibered map over $J^{1}\pi$, defined as follows
$$
\begin{array}{rcl}
\C \colon J^{2}\pi \times_{J^{1}\pi} \Lambda_2^m(J^{1}\pi) & \longrightarrow & \Lambda_1^m(J^{1}\pi) \\
(j^{2}_x\phi,\omega) & \longmapsto & (j^{1}\phi)^*_{j^{1}_x\phi}\omega
\end{array}
$$
As $\C$ takes values in $\Lambda_1^m(J^{1}\pi)$, there exists a \emph{pairing function}
associated to $\C$ and the volume form $\eta \in \Omega^{m}(M)$,
denoted by  $C \colon J^{2}\pi \times_{J^{1}\pi} \Lambda_2^m(J^{1}\pi) \to \mathbb{R}$,
and such that
$C(j^{2}_x\phi,\omega)\cdot(\bar{\pi}_{J^{1}\pi})^*\eta = (j^{1}\phi)^*_{j^{1}_x\phi}\omega$.

Let $(U;x^i,u^\alpha)$, $1 \leqslant i \leqslant m$, $1 \leqslant \alpha \leqslant n$, be a local
chart in $E$ adapted to the bundle structure and such that
$\eta = d x^1 \wedge \ldots \wedge d x^m \equiv d^mx$. Then, the induced natural coordinates
in $J^{1}\pi$ are $((\pi^{1})^{-1}(U);x^i,u^\alpha,u_i^\alpha)$. Therefore, the induced local
coordinates in $\Lambda_2^m(J^{1}\pi)$ are
$((\pi^{1} \circ \pi_{J^{1}\pi})^{-1}(U);x^i,u^\alpha,u_i^\alpha,p,p_\alpha^i,p_\alpha^{ij})$,
$1 \leqslant i,j \leqslant m$, $1 \leqslant \alpha \leqslant n$. Observe that
$\dim\Lambda_2^m(J^1\pi) = m + n + 2nm + nm^2 + 1$.
In these coordinates, the Liouville $m$ and $(m+1)$-forms have the following local expressions
\begin{equation}\label{eqn:LiouvilleFormsLocal}
\begin{array}{l}
\Theta_{1} = pd^mx + p^i_\alpha d u^\alpha \wedge d^{m-1}x_i +
p^{ij}_\alpha d u_i^\alpha \wedge d^{m-1}x_j \, , \\[10pt]
\Omega_{1} = -d p \wedge d^mx - d p^i_\alpha \wedge d u^\alpha \wedge d^{m-1}x_i -
d p^{ij}_\alpha \wedge d u_j^\alpha \wedge d^{m-1}x_j \, .
\end{array}
\end{equation}
Finally, the pairing function $C$ associated to $\C$ and $\eta$ has the following
coordinate expression
\begin{equation}\label{eqn:CanonicalPairingLocal}
C(x^i,u^\alpha,u_i^\alpha,p,p_\alpha^i,p_\alpha^{ij}) = p + p_\alpha^iu_i^\alpha +
p_\alpha^{ij}u_{1_i+1_j}^\alpha \, .
\end{equation}

According to the results in \cite{art:Saunders_Crampin90}, let us consider the submanifold
$J^2\pi^\dagger \hookrightarrow \Lambda_2^m(J^1\pi)$ defined locally by
$$
J^2\pi^\dagger = \left\{ \omega \in \Lambda_2^m(J^1\pi) \, \colon \, p_\alpha^{ij} = p_\alpha^{ji} \
\mbox{ for every } 1 \leqslant i,j\leqslant m \, , \, 1 \leqslant \alpha \leqslant n \right\} \, .
$$
This submanifold is $\pi_{J^{1}\pi}$-transverse, and therefore fibers over $J^{1}\pi$, $E$ and $M$.
Let $\pi_{J^1\pi}^\dagger \colon J^2\pi^\dagger \to J^1\pi$ and
$\bar{\pi}_{J^1\pi}^\dagger = \bar{\pi}^{1} \circ \pi_{J^1\pi}^\dagger \colon J^2\pi^\dagger \to M$
be the canonical projections. Natural coordinates in $J^2\pi^\dagger$ adapted to the bundle structure
are $(x^i,u^\alpha,u^\alpha_i,p,p_\alpha^{i},p_\alpha^{I})$, where $|I| = 2$. Using these coordinates,
the natural embedding  $j_s \colon J^2\pi^\dagger \hookrightarrow \Lambda_2^m(J^1\pi)$ is given by
\begin{equation}\label{eqn:EmbeddingSymmetricMultimomentaLocal}
\begin{array}{c}
\displaystyle j_s^*x^i = x^i \quad ; \quad j_s^*u^\alpha = u^\alpha \quad ; \quad j_s^*u_i^\alpha = u_i^\alpha
\quad ; \quad j_s^*p_\alpha^i = p_\alpha^i \, , \\[10pt]
\displaystyle j_s^*p_\alpha^{ij} = \frac{1}{n(ij)} \, p_\alpha^{1_i+1_j} \, , \quad \mbox{where }
n(ij) = \begin{cases} 1 \, , & \mbox{ if } i=j \\ 2 \, , & \mbox{ if } i \neq j \end{cases}
\end{array}
\end{equation}

The submanifold $J^2\pi^\dagger \hookrightarrow \Lambda_2^m(J^1\pi)$ is called the \emph{extended
$2$-symmetric multimomentum bundle}. Although this submanifold is defined using coordinates,
this construction is canonical \cite{art:Saunders_Crampin90,phd:Campos}.

\noindent\textbf{Remark:}
Observe that $J^2\pi^\dagger$ is defined by $nm(m-1)/2$ local constraints, and therefore we have
$$
\dim J^{2}\pi^\dagger = \dim\Lambda_2^m(J^1\pi) - \frac{nm(m-1)}{2} = m+n+2mn+\frac{nm(m+1)}{2} + 1 \, .
$$

All the geometric structures defined above for $\Lambda_2^m(J^1\pi)$ can be restricted to
$J^2\pi^\dagger$. In particular, let us denote $\Theta_1^s = j_s^*\Theta_1 \in \Omega^{m}(J^2\pi^\dagger)$
and $\Omega_1^s = j_s^*\Omega_1 = -d\Theta_1^s \in \Omega^{m+1}(J^2\pi^\dagger)$ the pull-back of the
Liouville $m$ and $(m+1)$-forms to $J^2\pi^\dagger$, which we call the
\emph{symmetrized Liouville $m$ and $(m+1)$-forms}. Bearing in mind the local expressions
\eqref{eqn:LiouvilleFormsLocal} of the Liouville $m$ and $(m+1)$-forms, and
\eqref{eqn:EmbeddingSymmetricMultimomentaLocal} of the canonical embedding
$j_s \colon J^2\pi^\dagger \hookrightarrow \Lambda_2^m(J^1\pi)$,
the coordinate expressions of $\Theta_1^s$ and $\Omega_1^s$ are
\begin{equation}\label{eqn:LiouvilleSymmetricFormsLocal}
\begin{array}{l}
\displaystyle \Theta_{1}^s = p d^mx + p^i_\alpha d u^\alpha \wedge d^{m-1}x_i +
 \frac{1}{n(ij)} \, p^{1_i+1_j}_\alpha d u_i^\alpha \wedge d^{m-1}x_j \, , \\[10pt]
\displaystyle \Omega_{1}^s = -d p \wedge d^mx -
d p^i_\alpha \wedge d u^\alpha \wedge d^{m-1}x_i -
 \frac{1}{n(ij)} \, d p^{1_i+1_j}_\alpha \wedge d u_i^\alpha \wedge d^{m-1}x_j \, .
\end{array}
\end{equation}

An important fact concerning the pull-back of the multisymplectic $(m+1)$-form $\Omega_1$ to
$J^2\pi^\dagger$ is that it is multisymplectic in $J^2\pi^\dagger$. Since $\Omega_1^s = -d\Theta_1^s$
is obviously closed, it suffices to show that it is $1$-nondegenerate, that is,
$i(X)\Omega_1^s = 0$ if, and only if, $X = 0$. In coordinates:
let $X \in \vf(J^2\pi^\dagger)$ be a generic vector field locally given by
$$
X = f^i\derpar{}{x^i} + F^\alpha\derpar{}{u^\alpha} + F_i^\alpha\derpar{}{u_i^\alpha}
+ g \derpar{}{p} + G_\alpha^i\derpar{}{p_\alpha^i} + G_\alpha^I\derpar{}{p_\alpha^I} \, .
$$
Then, taking into account the coordinate expression \eqref{eqn:LiouvilleSymmetricFormsLocal}
of the $(m+1)$-form $\Omega_1^s$, the $m$-form $i(X)\Omega_1^s$ is locally given by
\begin{align*}
i(X)\Omega_1^s &=
f^k\left( d p \wedge d^{m-1}x_k - d p_\alpha^i \wedge d u^\alpha \wedge
d^{m-2}x_{ik} - \frac{d p^{1_i+1_j}_\alpha \wedge d u_i^\alpha \wedge d^{m-2}x_{jk}}{n(ij)} \right) \\
&\qquad {} +
F^\alpha d p_\alpha^i \wedge d^{m-1}x_i + F_i^\alpha \frac{1}{n(ij)} \, d p^{1_i+1_j} \wedge d^{m-1}x_j
- gd^mx \\
&\qquad {} - G_\alpha^id u^\alpha \wedge d^{m-1}x_i
- G_\alpha^I\sum_{1_i+1_j=I}\frac{1}{n(ij)} d u_i^\alpha \wedge d^{m-1}x_j \, ,
\end{align*}
where $d^{m-2}x_{jk} = i(\partial / \partial x^k) d^{m-1}x_j$.
From this coordinate expression it is clear that $i(X)\Omega_1^s = 0$ if, and only if, $X = 0$.
Hence $\Omega_1^s$ is multisymplectic.

Furthermore, from the canonical pairing
$\C \colon J^2\pi \times_{J^1\pi} \Lambda_2^m(J^1\pi) \to \Lambda_1^m(J^1\pi)$,
we can define a pairing $\C^s \colon J^2\pi \times_{J^1\pi} J^2\pi^\dagger \to \Lambda_1^m(J^1\pi)$ as
$$
\C^s(j^{2}_x\phi,\omega) = \C(j^{2}_x\phi,j_s(\omega)) = (j^1\phi)_{j^1_x\phi}^* \ j_s(\omega) \, .
$$
Again, since $\C^s$ takes values in $\Lambda_1^m(J^1\pi)$, there exists
$C^s \in C^\infty(J^2\pi \times_{J^1\pi} J^2\pi^\dagger)$ such that
$C^s(j^2_x\phi,\omega) \cdot (\bar{\pi}_{J^1\pi}^\dagger)^*\eta = (j^1\phi)_{j^1_x\phi}^* \ j_s(\omega)$.
In the natural coordinates of $J^2\pi^\dagger$, bearing in mind the local expressions
\eqref{eqn:CanonicalPairingLocal} of the pairing function $C$ and
\eqref{eqn:EmbeddingSymmetricMultimomentaLocal} of the canonical embedding, the coordinate
expression of $C^s$ is
\begin{equation}\label{eqn:CanonicalPairingSymmetricLocal}
C^s(x^i,u^\alpha,u^\alpha_i,u^\alpha_I,p,p_\alpha^i,p_\alpha^I) =
p + p_\alpha^iu_i^\alpha + p_\alpha^{I}u_{I}^\alpha \, .
\end{equation}

Finally, let us consider the quotient bundle $J^2\pi^\ddagger = J^2\pi^\dagger / \Lambda^m_1(J^1\pi)$,
which is called the \emph{restricted $2$-symmetric multimomentum bundle}. This bundle is endowed
with some natural projections, namely the quotient map
$\mu \colon J^2\pi^\dagger \to J^2\pi^\ddagger$, and the projections
$\pi_{J^1\pi}^\ddagger \colon J^2\pi^\ddagger \to J^1\pi$ and
$\bar{\pi}_{J^1\pi}^\ddagger \colon J^2\pi^\ddagger \to M$.

Observe that $J^2\pi^\ddagger$ can also be defined as the submanifold of
$\Lambda^m_2(J^1\pi) / \Lambda^m_1(J^1\pi)$ defined by the $nm(m-1)/2$ local constraints
$p_\alpha^{ij} - p_\alpha^{ji} = 0$. Hence, natural coordinates
$(x^i,u^\alpha,u^\alpha_i,p,p_\alpha^i,p_\alpha^{ij})$ in $\Lambda_2^m(J^1\pi)$ induce local
coordinates $(x^i,u^\alpha,u^\alpha_i,p_\alpha^i,p_\alpha^{ij})$ in the quotient. Therefore, natural
coordinates in $J^2\pi^\ddagger$ are $(x^i,u^\alpha,u^\alpha_i,p_\alpha^i,p_\alpha^I)$.
Observe that
$$\dim J^{2}\pi^\ddagger = \dim J^{2}\pi^\dagger - 1 = m+n+2mn+\frac{nm(m+1)}{2} \, .
$$

\section{Lagrangian-Hamiltonian unified formalism}
\label{sec:LagHamFormalism}

\subsection{Geometrical setting}
\label{sec:GeomSetting}

Let $E \stackrel{\pi}{\longrightarrow} M$ be the configuration bundle describing a classical field
theory, where $M$ is a $m$-dimensional orientable manifold with fixed volume form
$\eta \in \Omega^{m}(M)$ and $E$ is a $(m+n)$-dimensional manifold. Let
$\Lag \in \Omega^{m}(J^2\pi)$ be a \emph{second-order Lagrangian density} for this theory,
that is, a $\bar{\pi}^2$-semibasic $m$-form on $J^{2}\pi$. Since $\Lag$ is a
$\bar{\pi}^2$-semibasic $m$-form,
we can write $\Lag = L\cdot(\bar{\pi}^2)^*\eta$, where $L \in C^\infty(J^2\pi)$ is the
\emph{second-order Lagrangian function} associated to $\Lag$ and $\eta$.

According to \cite{art:Barbero_Echeverria_Martin_Munoz_Roman08,
art:Echeverria_Lopez_Marin_Munoz_Roman04, art:Prieto_Roman12_1}, let us consider the fiber bundles
$$
\W = J^3\pi \times_{J^1\pi} J^{2}\pi^\dagger \quad ; \quad
\W_r = J^3\pi \times_{J^1\pi} J^{2}\pi^\ddagger \, .
$$
The bundles $\W$ and $\W_r$ are called the \emph{extended $2$-symmetric jet-multimomentum bundle}
and the \emph{restricted $2$-symmetric jet-multimomentum bundle}, respectively.

These bundles are endowed with the canonical projections
\begin{gather*}
\rho_1 \colon \W \to J^{3}\pi \quad ; \quad
\rho_2 \colon \W \to J^{2}\pi^\dagger \quad ; \quad
\rho_{J^{1}\pi} \colon \W \to J^{1}\pi \quad ; \quad
\rho_M \colon \W \to M \, , \\
\rho_1^r \colon \W_r \to J^{3}\pi \quad ; \quad
\rho^r_2 \colon \W_r \to J^{2}\pi^\ddagger \quad ; \quad
\rho_{J^{1}\pi}^r \colon \W_r \to J^{1}\pi \quad ; \quad
\rho_M^r \colon \W_r \to M \, .
\end{gather*}

In addition, the natural quotient map $\mu \colon J^{2}\pi^\dagger \to J^{2}\pi^\ddagger$ induces
a natural projection (that is, a surjective submersion) $\mu_\W \colon \W \to \W_r$. Thus, we have
the following diagram
$$
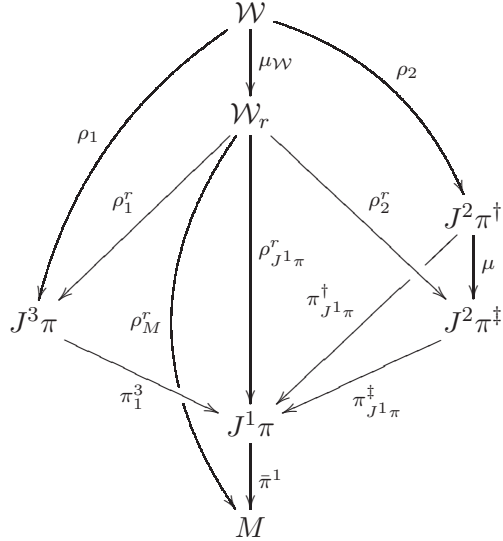
\xymatrix{
\ & \ & \W \ar@/_1.3pc/[llddd]_{\rho_1} \ar[d]^-{\mu_\W} \ar@/^1.3pc/[rrdd]^{\rho_2} & \ & \ \\
\ & \ & \W_r \ar[lldd]_{\rho_1^r} \ar[rrdd]^{\rho_2^r} \ar[ddd]^<(0.4){\rho_{J^{1}\pi}^r} \ar@/_2.5pc/[dddd]_-{\rho_M^r}|(.675){\hole} & \ & \ \\
\ & \ & \ & \ & J^{2}\pi^\dagger \ar[d]^-{\mu} \ar[lldd]_{\pi_{J^{1}\pi}^\dagger}|(.25){\hole} \\
J^{3}\pi \ar[rrd]_{\pi^{3}_{1}} & \ & \ & \ & J^{2}\pi^\ddagger \ar[dll]^{\pi_{J^{1}\pi}^\ddagger} \\
\ & \ & J^{1}\pi \ar[d]^{\bar{\pi}^{1}} & \ & \ \\
\ & \ & M & \ & \
}
$$

Let $(U;x^i,u^\alpha)$ be a local chart of coordinates in $E$ adapted to the bundle structure
and such that $\eta = d x^1 \wedge \ldots \wedge d x^m \equiv d^mx$. Then, we denote by
$((\pi^{3})^{-1}(U);x^i,u^\alpha,u^\alpha_{i},u^\alpha_{I},u^\alpha_{J})$ and
$((\pi^{1} \circ \pi_{J^{1}\pi}^\dagger)^{-1}(U);x^i,u^\alpha,u^\alpha_i,p,p_\alpha^i,p_\alpha^{I})$
the induced local charts in $J^3\pi$ and $J^{2}\pi^\dagger$, respectively, with $|I| = 2$
and $|J|=3$. Thus, $(x^i,u^\alpha,u^\alpha_i,p_\alpha^i,p_\alpha^{I})$ are the natural
coordinates in $J^{2}\pi^\ddagger$, and the coordinates in $\W$ and $\W_r$ are
$(x^i,u^\alpha,u^\alpha_i,u^\alpha_{I},u^\alpha_{J},p,p_\alpha^{i},p_\alpha^{I})$ and
$(x^i,u^\alpha,u^\alpha_i,u^\alpha_{I},u^\alpha_{J},p_\alpha^{i},p_\alpha^{I})$, respectively.
Observe that
$$
\dim\W = m + n + 2nm + nm(m+1) + \frac{nm(m+1)(m+2)}{6} + 1 \, ,
$$
and $\dim\W_r = \dim\W - 1$.

The bundle $\W$ is endowed with some canonical structures.

\begin{definition}
Let $\Theta_{1}^s \in \Omega^{m}(J^{2}\pi^\dagger)$ and
$\Omega_{1}^s \in \Omega^{m+1}(J^{2}\pi^\dagger)$
be the symmetrized Liouville forms. Then we define the following forms in $\W$
\begin{equation}\label{eqn:UnifiedCanonicalFormDef}
\Theta = \rho_2^*\Theta_1^s \in \Omega^{m}(\W) \quad ; \quad
\Omega = \rho_2^*\Omega_1^s \in \Omega^{m+1}(\W) \, ,
\end{equation}
which are called the \emph{second-order unified canonical forms}.
\end{definition}

Bearing in mind the local expressions \eqref{eqn:LiouvilleSymmetricFormsLocal} of the forms
$\Theta_1^s$ and $\Omega_1^s$, and taking into account that the projection $\rho_2$ is locally
given by
$$
\rho_2(x^i,u^\alpha,u^\alpha_i,u^\alpha_{I},u^\alpha_{J},p,p_\alpha^{i},p_\alpha^{I}) =
(x^i,u^\alpha,u^\alpha_i,p,p_\alpha^{i},p_\alpha^{I}) \, ,
$$
we obtain the coordinate expression of the unified canonical forms, which are
\begin{equation}\label{eqn:UnifiedCanonicalFormsLocal}
\begin{array}{l}
\displaystyle
\Theta = pd^mx + p_\alpha^i d u^\alpha \wedge d^{m-1}x_i + \frac{1}{n(ij)} \, p_\alpha^{1_i+1_j} d u_i^\alpha \wedge d^{m-1}x_j \, , \\[10pt]
\displaystyle
\Omega = - d p \wedge d^mx - d p_\alpha^i \wedge d u^\alpha \wedge d^{m-1}x_i - \frac{1}{n(ij)} \, d p_\alpha^{1_i+1_j} \wedge d u_i^\alpha \wedge d^{m-1}x_j \, .
\end{array}
\end{equation}
Observe that, although $\Omega_1^s$ is multisymplectic, the $(m+1)$-form $\Omega$ is
premultisymplectic, since it is closed and $1$-degenerate. Indeed, for every
$X \in \vf^{V(\rho_2)}(\W)$ we have $i(X)\Omega = 0$. This is easy to check in coordinates:
the $C^\infty(\W)$-module $\vf^{V(\rho_2)}(\W)$ is locally given by
\begin{equation}\label{eqn:PremultisymplecticKernelLocal}
\vf^{V(\rho_2)}(\W) = \left\langle \derpar{}{u^\alpha_{I}},\derpar{}{u^\alpha_{J}} \right\rangle \, ,
\end{equation}
with $|I| = 2$ and $|J| = 3$. Bearing in mind the local expression
\eqref{eqn:UnifiedCanonicalFormsLocal} for $\Omega$, we have
$$
i\left(\derpar{}{u^\alpha_{I}}\right)\Omega = i\left(\derpar{}{u^\alpha_{J}}\right) \Omega = 0 \, .
$$
Hence, $(\W,\Omega)$ is a premultisymplectic manifold of degree $m+1$, and we have
$\ker\Omega = \vf^{V({\rho}_2)}(\W)$.

The second canonical structure in $\W$ is the following:

\begin{definition}
The \emph{second-order coupling $m$-form} in $\W$ is the $\rho_M$-semibasic $m$-form
$\hat{\C} \in \Omega^{m}(\W)$ defined as follows: for every
$(j^3_x\phi,\omega) \in \W$ we have
\begin{equation}\label{eqn:CouplingFormDef}
\hat{C}(j^3_x\phi,\omega) = \C^s(\pi^3_2(j^3_x\phi),\omega)  \, .
\end{equation}
\end{definition}

As before, since $\hat{\C}$ is a $\rho_M$-semibasic $m$-form, there exists a function
$\hat{C} \in C^\infty(\W)$ such that $\hat{\C} = \hat{C}\cdot\rho_M^*\eta$. Bearing
in mind the local expression \eqref{eqn:CanonicalPairingSymmetricLocal} of $C^s$, the coordinate
expression of the second-order coupling form is
\begin{equation}\label{eqn:CouplingFormLocal}
\hat{\C} = \left( p + p_\alpha^iu_i^\alpha + p_\alpha^{I}u_{I}^\alpha \right)d^mx \, .
\end{equation}

We denote $\hat{\Lag} = (\pi^3_2 \circ \rho_1)^*\Lag \in \Omega^{m}(\W)$.
Since the $\Lag$ is a $\bar{\pi}^2$-semibasic form, we have that $\hat{\Lag}$ is a
$\rho_M$-semibasic $m$-form, and thus we can write $\hat{\Lag} = \hat{L} \cdot \rho_M^*\eta$,
where $\hat{L} = (\pi^3_2 \circ \rho_1)^*L \in C^\infty(\W)$ is the pull-back of the Lagrangian
function associated with $\Lag$ and $\eta$. Then, we define a \emph{Hamiltonian submanifold}
$$
\W_o = \left\{ w \in \W \colon \hat{\Lag}(w) = \hat{\C}(w) \right\} \stackrel{j_o}{\hookrightarrow} \W \, .
$$
Since both $\hat{\Lag}$ and $\hat{\C}$ are $\rho_M$-semibasic $m$-forms, the submanifold $\W_o$
is defined by the constraint $\hat{C} - \hat{L} = 0$. In local coordinates, bearing in mind the local
expression \eqref{eqn:CouplingFormLocal} of $\hat{\C}$, the constraint function is
$$
p + p_\alpha^iu_i^\alpha + p_\alpha^{I}u_{I}^\alpha - \hat{L} = 0 \, , \quad |I|=2 \, .
$$

\begin{proposition}\label{prop:WoDiffWr}
The submanifold $\W_o \hookrightarrow \W$ is $1$-codimensional, $\mu_\W$-transverse, and the map
$\Phi = \mu_\W \circ j_o \colon \W_o \to \W_r$ is a diffeomorphism.
\end{proposition}
\begin{proof}
First of all, observe that $\W_o$ is obviously $1$-codimensional, since it is defined by a single
constraint function.

To prove that $\Phi = \mu_\W \circ j_o \colon \W_o \to \W$ is a diffeomorphism,
we show that it is one-to-one. First, observe that for every $(j^3_x\phi,\omega) \in \W_o$, we have
$$
L(\pi^{3}_{2}(j^3_x\phi)) = \hat{L}(j^3_x\phi,\omega) = \hat{C}(j^3_x\phi,\omega) \, ,
$$
and, on the other hand,
$$
(\mu_\W \circ j_o)(j^3_x\phi,\omega) = \mu_\W(j^3_x\phi,\omega)
= (j^3_x\phi,\mu(\omega))
= (j^3_x\phi,[\omega]) \, .
$$

First, let us prove that $\mu_\W \circ j_o$ is injective. In fact, let
$(j^3_x\phi_1,\omega_1), (j^3_x\phi_2,\omega_2) \in \W_o$, then we wish to prove that
\begin{align*}
(\mu_\W \circ j_o)(j^3_x\phi_1,\omega_1) = (\mu_\W \circ j_o)(j^3_x\phi_2,\omega_2) &\Longleftrightarrow (j^3_x\phi_1,\omega_1) = (j^3_x\phi_2,\omega_2) \\
&\Longleftrightarrow j^3_x\phi_1 = j^3_x\phi_2 \mbox{ and } \omega_1 = \omega_2 \, .
\end{align*}
Now, using the previous expression for $(\mu_\W \circ j_o)(j^3_x\phi,\omega)$, we have
\begin{align*}
(\mu_\W \circ j_o)(j^3_x\phi_1,\omega_1) = (\mu_\W \circ j_o)(j^3_x\phi_2,\omega_2) &\Longleftrightarrow
(j^3_x\phi_1,[\omega_1]) = (j^3_x\phi_2,[\omega_2]) \\ &\Longleftrightarrow
j^3_x\phi_1 = j^3_x\phi_2 \mbox{ and } [\omega_1] = [\omega_2] \, ,
\end{align*}
From where we deduce $j^3_x\phi_1 = j^3_x\phi_2 \equiv j^3_x\phi$. Now, to prove $\omega_1 = \omega_2$,
observe that by definition of $\W_o$, we have
$$
L(\pi^{3}_{2}(j^3_x\phi)) = L(\pi^{3}_{2}(j^3_x\phi))
= \hat{C}(j^3_x\phi,\omega_1) = \hat{C}(j^3_x\phi,\omega_2) \, .
$$
Locally, from the third equality we obtain
$$
p(\omega_1) + p_\alpha^i(\omega_1)u_{i}^\alpha(j^3_x\phi) + p_\alpha^I(\omega_1)u_I^\alpha(j^3_x\phi)
= p(\omega_2) + p_\alpha^i(\omega_2)u_{i}^\alpha(j^3_x\phi) + p_\alpha^I(\omega_2)u_I^\alpha(j^3_x\phi) \, ,
$$
but $[\omega_1] = [\omega_2]$ implies
\begin{gather*}
p^i_\alpha(\omega_1) = p_\alpha^i([\omega_1]) = p_\alpha^i([\omega_2]) = p^i_\alpha(\omega_2) \, , \\
p^I_\alpha(\omega_1) = p_\alpha^I([\omega_1]) = p_\alpha^I([\omega_2]) = p^I_\alpha(\omega_2) \, .
\end{gather*}
Then $p(\omega_1) = p(\omega_2)$, and hence $\omega_1 = \omega_2$. Now, let us prove that
$\mu_\W \circ j_o$ is surjective. In fact, given $(j^3_x\phi,[\omega]) \in \W_r$, we
wish to find $(j^3_x\phi,\zeta) \in j_o(\W_o)$ such that $[\zeta] = [\omega]$. It suffices to take
$[\zeta]$ such that, in local coordinates of $\W$, it satisfies
\begin{gather*}
p_\alpha^i(\zeta) = p_\alpha^i([\zeta]) \quad , \quad
p_\alpha^I(\zeta) = p_\alpha^I([\zeta]) \\
p(\zeta) = L(\pi^{3}_{2}(j^3_x\phi)) - p_\alpha^i([\omega])u_{i}^\alpha(j^3_x\phi) - p_\alpha^I([\omega])u_I^\alpha(j^3_x\phi) \, .
\end{gather*}
This $\zeta$ exists as a consequence of the definition of $\W_o$. Now, since $\mu_\W \circ j_o$
is a one-to-one submersion, then, by equality on the dimensions of $\W_o$ and $\W_r$, it is a
one-to-one local diffeomorphism, and thus a global diffeomorphism.

Finally, in order to prove that $\W_o$ is $\mu_\W$-transversal, it is necessary to check if
$\Lie(X)(\xi) \equiv X(\xi) \neq 0$, for every $X \in \ker{\mu_\W}_*$ and every constraint function
$\xi$ defining $\W_o$. Since $\W_o$ is defined by the constraint $\hat{C} - \hat{L} = 0$ and
$\ker{\mu_\W}_* = \langle \partial/\partial p \rangle$, computing we have
$$
\derpar{}{p}(\hat{C} - \hat{L}) = \derpar{}{p}(p + p_\alpha^iu_{i}^\alpha + p_\alpha^Iu_I^\alpha - \hat{L}) = 1 \neq 0 \, ,
$$
then $\W_o$ is $\mu_\W$-transverse.
\end{proof}

As a consequence of Proposition \ref{prop:WoDiffWr}, the submanifold $\W_o$ induces a section
$\hat{h} \in \Gamma(\mu_\W)$ defined as
$\hat{h} = j_o \circ \Phi^{-1} \colon \W_r \to \W$,
which is called a \emph{Hamiltonian section of $\mu_\W$} or a
\emph{Hamiltonian $\mu_\W$-section}.
This section is specified by giving the local \emph{Hamiltonian function}
\begin{equation}\label{eqn:HamiltonianFunctionLocal}
\hat{H} = p_\alpha^iu_i^\alpha + p_\alpha^{I}u_{I}^\alpha - \hat{L} \, ,
\end{equation}
that is, $\hat{h}(x^i,u^\alpha,u^\alpha_i,u^\alpha_{I},u^\alpha_{J},p_\alpha^i,p_\alpha^{I}) =
(x^i,u^\alpha,u^\alpha_i,u^\alpha_{I},u^\alpha_{J},-\hat{H},p_\alpha^i,p_\alpha^{I})$.
Observe that $\hat{h}$ satisfies $\rho_1^r = \rho_1 \circ \hat{h}$ and
$\rho_2^r = \mu \circ \rho_2 \circ \hat{h}$. Hence, we have the following commutative diagram:
$$
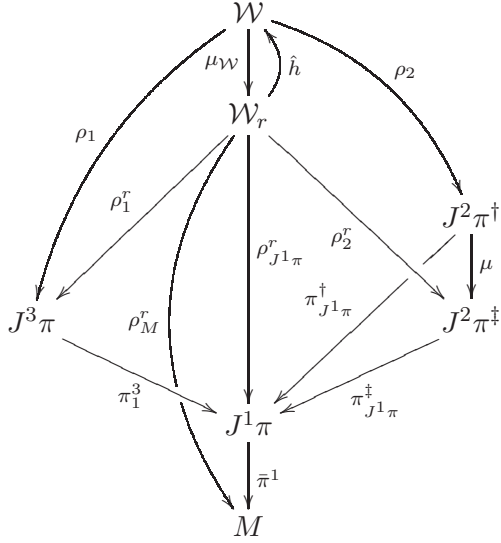
\xymatrix{
\ & \ & \W \ar@/_1.3pc/[llddd]_{\rho_1} \ar[d]_-{\mu_\W} \ar@/^1.3pc/[rrdd]^{\rho_2} & \ & \ \\
\ & \ & \W_r \ar@/_1pc/[u]_{\hat{h}} \ar[lldd]_{\rho_1^r} \ar[rrdd]_{\rho_2^r} \ar[ddd]^<(0.4){\rho_{J^{1}\pi}^r} \ar@/_2.5pc/[dddd]_-{\rho_M^r}|(.675){\hole} & \ & \ \\
\ & \ & \ & \ & J^{2}\pi^\dagger \ar[d]^-{\mu} \ar[lldd]_{\pi_{J^{1}\pi}^\dagger}|(.25){\hole} \\
J^{3}\pi \ar[rrd]_{\pi^{3}_{1}} & \ & \ & \ & J^{2}\pi^\ddagger \ar[dll]^{\pi_{J^{1}\pi}^\ddagger} \\
\ & \ & J^{1}\pi \ar[d]^{\bar{\pi}^{1}} & \ & \ \\
\ & \ & M & \ & \
}
$$

Next, we define the forms
$$
\Theta_r = \hat{h}^*\Theta \in \Omega^{m}(\W_r) \quad ; \quad
\Omega_r = \hat{h}^*\Omega \in \Omega^{m+1}(\W_r) \, ,
$$
with local expressions
\begin{equation}\label{eqn:HamiltonCartanFormsLocal}
\begin{array}{l}
\displaystyle
\Theta_r = - \hat{H} d^mx + p_\alpha^id u^\alpha \wedge d^{m-1}x_i + \frac{1}{n(ij)} \, p_\alpha^{1_i+1_j}d u_i^\alpha \wedge d^{m-1}x_j \, , \\[10pt]
\displaystyle
\Omega_r = d \hat{H} \wedge d^mx - d p_\alpha^i \wedge d u^\alpha \wedge d^{m-1}x_i - \frac{1}{n(ij)} \, d p_\alpha^{1_i+1_j} \wedge d u_i^\alpha \wedge d^{m-1}x_j \, .
\end{array}
\end{equation}

Finally, we generalize the definition of holonomic sections and multivector fields to the unified setting.

\begin{definition}
A section $\psi \in \Gamma(\rho_M^r)$ is \emph{holonomic of type $s$ in $\W_r$},
$1 \leqslant s \leqslant 3$, if the section $\rho_1^r \circ \psi \in \Gamma(\bar{\pi}^3)$
is holonomic of type $s$ in $J^3\pi$.
\end{definition}

\begin{definition}
A multivector field $\X \in \vf^{m}(\W_r)$ is
\emph{holonomic of type $s$ in $\W_r$},
$1 \leqslant s \leqslant 3$, if
\begin{enumerate}
\item $\X$ is integrable.
\item $\X$ is $\rho_M^r$-transverse.
\item The integral sections $\psi \in \Gamma(\rho_M^r)$ of $\X$ are holonomic of type $s$
in $\W_r$.
\end{enumerate}
\end{definition}

\subsection{Field equations for sections}
\label{sec:UnifFieldEquationsSect}

The \emph{Lagrangian-Hamiltonian problem for sections} associated with the system $(\W_r,\Omega_r)$
consists in finding holonomic sections $\psi \in \Gamma(\rho_M^r)$ satisfying the following condition
\begin{equation}\label{eqn:UnifDynEqSect}
\psi^*i(X)\Omega_r = 0 \, , \quad \mbox{for every } X \in \vf(\W_r) \, .
\end{equation}
In the induced natural coordinates of $\W_r$, let $\psi \in \Gamma(\rho_M^r)$ be a section
locally given by $\psi(x^i) = (x^i,u^\alpha,u^\alpha_i,u^\alpha_{I},u^\alpha_{J},p_\alpha^i,p_\alpha^{I})$.
Then, bearing in mind the coordinate expression \eqref{eqn:HamiltonCartanFormsLocal} of $\Omega_r$,
we obtain the following system of partial differential equations for the component functions of
the section $\psi$
\begin{align}
& \sum_{i=1}^{m}\derpar{p_\alpha^i}{x^i} - \derpar{\hat{L}}{u^\alpha} = 0 \, , \label{eqn:UnifDynEqSectLocal} \\
& \sum_{j=1}^{m} \frac{1}{n(ij)} \, \derpar{p_\alpha^{1_i+1_j}}{x^j} + p_\alpha^i - \derpar{\hat{L}}{u_i^\alpha} = 0 \, , \label{eqn:UnifDynEqSectRelationMomenta} \\
& p_\alpha^{I} - \derpar{\hat{L}}{u_{I}^\alpha} = 0 \, , \label{eqn:HOMomentaSect} \\
& u_i^\alpha - \derpar{u^\alpha}{x^i} = 0 \quad ; \quad u_{I}^\alpha - \sum_{1_i+1_j=I} \frac{1}{n(ij)} \, \derpar{u_i^\alpha}{x^j} = 0 \, . \label{eqn:UnifDynEqSectHolonomy}
\end{align}
Observe that equations \eqref{eqn:UnifDynEqSectHolonomy} give partially the holonomy condition for
the section $\psi$, but since we required this condition from the beginning, these equations are
automatically satisfied.

Notice also that equations \eqref{eqn:HOMomentaSect} do not involve any partial derivative of the
component functions of $\psi$: they are pointwise algebraic conditions that must be fullfilled for
every section $\psi \in \Gamma(\rho_M^r)$ solution to the field equation \eqref{eqn:UnifDynEqSect}.
These equations arise from the $\rho_2^r$-vertical part of the vector fields
$X \in \vf(\W_r)$, as shown in the following result.

\begin{lemma}\label{lemma:HOMomentaIntrinsic}
If $X \in \vf^{V(\rho_2^r)}(\W_r)$, then
$i(X)\Omega_r \in \Omega^{m}(\W_r)$ is $\rho_M^r$-semibasic.
\end{lemma}
\begin{proof}
This result is easy to prove in coordinates. In the natural coordinates of $\W_r$,
the $C^\infty(\W_r)$-module of $\rho_2^r$-vertical vector fields is given by
$$
\vf^{V(\rho_2^r)}(\W_r) = \left\langle \derpar{}{u_I^\alpha} \right\rangle \, ,
$$
with $2 \leqslant |I| \leqslant 3$. Then, bearing in mind the local expression
\eqref{eqn:HamiltonCartanFormsLocal} of $\Omega_r$, we have
$$
i\left( \derpar{}{u_I^\alpha} \right) \Omega_r =
\begin{cases}
\left( p_\alpha^{I} - \derpar{\hat{L}}{u_I^\alpha} \right) d^{m}x \, , & \mbox{for } |I| = 2 \, , \\[10pt]
0 = 0 \cdot d^{m}x \, , & \mbox{for } |I| > 2 \, .
\end{cases}
$$
Thus, in both cases we obtain a $\rho_M^r$-semibasic $m$-form.
\end{proof}

As a consequence of this result, we can define the submanifold
\begin{equation}\label{eqn:CompSubmanifoldSect}
\W_c = \left\{ w \in \W_r \colon (i(X)\Omega_r)(w) = 0\ \mbox{ for every } X \in \vf^{V(\rho_2^r)}(\W_r) \right\}
\stackrel{j_c}{\hookrightarrow} \W_r \, ,
\end{equation}
where every section $\psi \in \Gamma(\rho_M^r)$ solution to the equation \eqref{eqn:UnifDynEqSect}
must take values. This submanifold is called the \emph{first constraint submanifold}
of the premultisymplectic system $(\W_r,\Omega_r)$, and has codimension $nm(m+1)/2$.

As we have seen in the proof of Lemma \ref{lemma:HOMomentaIntrinsic}, the submanifold
$\W_c \hookrightarrow \W_r$ is locally defined by the constraints
\eqref{eqn:HOMomentaSect}. In combination with equations \eqref{eqn:UnifDynEqSectRelationMomenta},
we have the following result.

\begin{proposition}\label{prop:GraphLegMapSect}
A solution $\psi \in \Gamma(\rho_M^r)$ to equation \eqref{eqn:UnifDynEqSect} takes values in a
$nm$-codimensional submanifold $\W_\Lag \hookrightarrow \W_c$ which is identified with the graph of
a bundle map $\Leg \colon J^3\pi \to J^{2}\pi^\ddagger$ over $J^1\pi$ defined locally by
\begin{equation}\label{eqn:RestrictedLegendreMapLocal}
\Leg^*p^i_\alpha = \derpar{\hat{L}}{u_i^\alpha} - \sum_{j=1}^{m}\frac{1}{n(ij)} \frac{d}{dx^j}\left( \derpar{\hat{L}}{u_{1_i+1_j}^\alpha} \right) \quad ; \quad
\Leg^*p^I_\alpha = \derpar{\hat{L}}{u_I^\alpha} \, .
\end{equation}
\end{proposition}
\begin{proof}
Since $\W_c$ is defined locally by the constraints \eqref{eqn:HOMomentaSect}, it suffices
to prove that these contraints, in combination with the remaining local equations for the section
$\psi \in \Gamma(\rho_M^r)$ to be a solution to the equation \eqref{eqn:UnifDynEqSect}, give rise
to the local functions defining the bundle map given above, and thus to the submanifold
$\W_\Lag$.

Replacing $p^I_\alpha$ by $\partial \hat{L} / \partial u_I^\alpha$ in equations
\eqref{eqn:UnifDynEqSectRelationMomenta}, we obtain
$$
p_\alpha^i - \derpar{\hat{L}}{u_i^\alpha}
+ \sum_{j=1}^{m} \frac{1}{n(ij)} \, \frac{d}{dx^j} \, \derpar{\hat{L}}{u_{1_i+1_j}^\alpha} = 0 \, .
$$
Therefore, these constraints define a submanifold $\W_\Lag \hookrightarrow \W_c$, which can be
identified with the graph of a map $\Leg \colon J^3\pi \to J^2\pi^\ddagger$ given by
\begin{gather*}
\Leg^*x^i = x^i \quad ; \quad \Leg^*u^\alpha = u^\alpha \quad ; \quad \Leg^*u_i^\alpha \, , \\
\Leg^*p^i_\alpha = \derpar{\hat{L}}{u_i^\alpha} - \sum_{j=1}^{m}\frac{1}{n(ij)} \frac{d}{dx^j}\left( \derpar{\hat{L}}{u_{1_i+1_j}^\alpha} \right) \quad ; \quad
\Leg^*p^I_\alpha = \derpar{\hat{L}}{u_I^\alpha} \, . \qedhere
\end{gather*}
\end{proof}

The bundle map $\Leg \colon J^3\pi \to J^2\pi^\ddagger$ is called the \emph{restricted Legendre map}
associated with the Lagrangian density $\Lag$. Observe that
$$
\dim\W_\Lag = \dim J^3\pi = m + n + mn + \frac{nm(m+1)}{2} + \frac{nm(m+1)(m+2)}{6} \, .
$$

\noindent\textbf{Remark:}
The terminology ``Legendre map'' is justified, since $\Leg$ is a fiber bundle morphism from the
Lagrangian phase space to the Hamiltonian phase space that identifies the multimomenta coordinates
with functions on partial derivatives of the Lagrangian function, and thus generalizes the Legendre
map in first-order field theories (see \cite{art:Echeverria_DeLeon_Munoz_Roman07,art:Echeverria_Munoz_Roman00_JMP}), and first-order and
higher-order mechanics (see \cite{book:Abraham_Marsden78} for first-order mechanics and
\cite{book:DeLeon_Rodrigues85} for the higher-order setting).

According to \cite{art:Saunders_Crampin90}, we can give the following definition.

\begin{definition}
A second-order Lagrangian density $\Lag \in \Omega^{m}(J^2\pi)$ is \emph{regular} if for
every point $j^3_x\phi \in J^{3}\pi$ we have
$$
\operatorname{rank}(\Leg(j^3\phi)) = \dim J^{2}\pi + \dim J^{1}\pi - \dim E = \dim J^2\pi^\ddagger \, .
$$
Otherwise, the Lagrangian density is said to be \emph{singular}.
\label{reglag}
\end{definition}

Hence, a second-order Lagrangian density $\Lag \in \Omega^{m}(J^2\pi)$ is regular if, and only if, the
restricted Legendre map $\Leg \colon J^3\pi \to J^{2}\pi^\ddagger$ associated to $\Lag$ is a
submersion onto $J^{2}\pi^\ddagger$. This implies that there exist local sections of $\Leg$, that is,
maps $\sigma \colon U \to J^3\pi$, with $U \subset J^{2}\pi^\ddagger$ an open set, such that
$\Leg \circ \sigma = \textnormal{Id}_U$. If $\Leg$ admits a global section
$\Upsilon \colon J^2\pi^\ddagger \to J^3\pi$, then the Lagrangian density is said to be
\emph{hyperregular}.

Observe that
\begin{align*}
\dim J^3\pi &= m + n + nm + \frac{nm(m+1)}{2} + \frac{nm(m+1)(m+2)}{6} \\
&\geqslant m + n + nm + 2nm + \frac{nm(m+1)}{2} = \dim J^2\pi^\ddagger \, ,
\end{align*}
and the equality holds if, and only if, $m = 1$. Therefore, unlike in higher-order mechanics or
first-order field theories, the Legendre map cannot be a local diffeomorphism due to dimension
restrictions.

Computing the local expression of the tangent map to $\Leg$ in a natural chart of $J^{3}\pi$,
the regularity condition for the Lagrangian density $\Lag$ is equivalent to
$$
\det\left( \derpars{L}{u_I^\beta}{u_K^\alpha} \right)(j^3_x\phi) \neq 0
\, , \quad \mbox{for every } j^3_x\phi \in J^3\pi \, ,
$$
where $|I|=|K|=2$. That is, the Hessian of the Lagrangian function associated to $\Lag$ and $\eta$
with respect to the highest order velocities is a regular matrix at every point, which is the usual
definition for a regular Lagrangian density.

Note that since $\W_r$ is diffeomorphic to the submanifold
$\W_o \hookrightarrow \W$ (Proposition \ref{prop:WoDiffWr}), and $\W_o$
is defined locally by the constraint $p + p_\alpha^iu_i^\alpha + p_\alpha^Iu_I^\alpha - \hat{L} = 0$,
the restricted Legendre map $\Leg \colon J^3\pi \to J^2\pi^\ddagger$ can be extended in a
canonical way to a map $\widetilde{\Leg} \colon J^3\pi \to J^2\pi^\dagger$, defining
$\widetilde{\Leg}^*p$ as the pull-back of the local Hamiltonian function $-\hat{H}$. This
enables us to state the following result, which is a straightforward consequence of Proposition
\ref{prop:GraphLegMapSect}

\begin{corollary}
The submanifold $\W_\Lag \hookrightarrow \W$ is the graph of a bundle morphism
$\widetilde{\Leg} \colon J^3\pi \to J^2\pi^\dagger$ over $J^1\pi$ defined locally by
\begin{equation}\label{eqn:ExtendedLegendreMapLocal}
\begin{array}{l}
\widetilde{\Leg}^*p^i_\alpha = \derpar{\hat{L}}{u_i^\alpha} - \sum_{j=1}^{m}\frac{1}{n(ij)} \, \frac{d}{dx^j}\left( \derpar{\hat{L}}{u_{1_i+1_j}^\alpha} \right) \quad ; \quad
\widetilde{\Leg}^*p^I_\alpha = \derpar{\hat{L}}{u_I^\alpha} \, , \\[12pt]
\widetilde{\Leg}^*p = \hat{L} - u_i^\alpha\left(\derpar{\hat{L}}{u_i^\alpha} - \sum_{j=1}^{m}\frac{1}{n(ij)} \, \frac{d}{dx^j}\left( \derpar{\hat{L}}{u_{1_i+1_j}^\alpha} \right) \right) - u_{I}^\alpha\derpar{\hat{L}}{u_I^\alpha} \, ,
\end{array}
\end{equation}
and satisfying $\Leg = \mu \circ \widetilde{\Leg}$.
\end{corollary}

The bundle map $\widetilde{\Leg} \colon J^3\pi \to J^2\pi^\dagger$ is the
\emph{extended Legendre map} associated with the Lagrangian density $\Lag$. An important result
concerning both Legendre maps is the following.

\begin{proposition}\label{prop:RankLegendreMaps}
For every $j^3_x\phi \in J^3\pi$ we have
$\operatorname{rank}(\widetilde{\Leg}(j^3_x\phi)) = \operatorname{rank}(\Leg(j^3_x\phi))$.
\end{proposition}

Following the same patterns as in \cite{art:DeLeon_Marin_Marrero96} for first-order mechanical
systems, the proof of this result consists in computing in a natural chart of coordinates the local
expressions of the Jacobian matrices of both maps $\widetilde{\Leg}$ and $\Leg$. Then,
observe that the ranks of both maps depend on the rank of the Hessian matrix of the Lagrangian function
with respect to the highest order velocities, and that the additional row in the Jacobian matrix of
$\widetilde{\Leg}$ is a combination of the others. Since it is just a long calculation in
coordinates, we omit the proof of this result.

Notice that the component functions $u_J^\alpha$ with $|J| = 3$ of the section $\psi \in \Gamma(\rho_M^r)$
are not yet determined, since the coordinate expression of the field equation \eqref{eqn:UnifDynEqSect}
does not give any condition on these functions. In fact, these functions are determined by the equations
\eqref{eqn:UnifDynEqSectLocal} and \eqref{eqn:UnifDynEqSectRelationMomenta}. Indeed, since the section
$\psi \in \Gamma(\rho_M^r)$ must take values in the submanifold $\W_\Lag$ given by Proposition
\ref{prop:GraphLegMapSect}, then by replacing the local expression of the restricted Legendre map in
equations \eqref{eqn:UnifDynEqSectLocal} and \eqref{eqn:UnifDynEqSectRelationMomenta} we obtain the
Euler-Lagrange equations for field theories:
\begin{equation}\label{eqn:EulerLagrangeUnified}
\restric{\derpar{\hat{L}}{u^\alpha}}{\psi} - \restric{\frac{d}{dx^i} \, \derpar{\hat{L}}{u_i^\alpha}}{\psi}
+ \restric{\sum_{|I|=2} \frac{d^{|I|}}{dx^{I}} \, \derpar{\hat{L}}{u_I^\alpha}}{\psi} = 0 \, ,
\quad 1 \leqslant \alpha \leqslant n \, .
\end{equation}

Finally, observe that since the section $\psi \in \Gamma(\rho_M^r)$ must take values in the
submanifold $\W_\Lag \hookrightarrow \W_r$, it is natural to consider the
restriction of equation \eqref{eqn:UnifDynEqSect} to the submanifold $\W_\Lag$;
that is, to restrict the set of vector fields to those tangent to $\W_\Lag$.
Nevertheless, the new equation may not be equivalent to the former. The following result gives a
sufficient condition for these two equations to be equivalent.

\begin{proposition}\label{prop:UnifDynEqSectTangent}
If $\psi \in \Gamma(\rho_M^r)$ is holonomic in $\W_r$,
then the equation \eqref{eqn:UnifDynEqSect} is equivalent to
\begin{equation}\label{eqn:UnifDynEqSectTangent}
\psi^*i(Y)\Omega_r = 0 \, , \quad \mbox{for every } Y \in \vf(\W_r) \mbox{ tangent to } \W_\Lag \, .
\end{equation}
\end{proposition}
\begin{proof}
We prove this result in coordinates. First of all, let us compute the coordinate expression of a
vector field $X \in \vf(\W_r)$ tangent to $\W_\Lag$. Let $X$ be a generic vector field locally given by
$$
X = f^i \derpar{}{x^i} + F^\alpha\derpar{}{u^\alpha} + F^\alpha_i\derpar{}{u_i^\alpha} + F_I^\alpha\derpar{}{u_I^\alpha}
+ F_J^\alpha\derpar{}{u_J^\alpha} + G_\alpha^i\derpar{}{p_\alpha^i} + G_\alpha^I\derpar{}{p_\alpha^I} \, .
$$
Then, since $\W_\Lag$ is the submanifold of $\W_r$ defined locally by the $nm + nm(m+1)/2$
constraint functions $\xi^i_\alpha$, $\xi^I_\alpha$ with coordinate expression
$$
\xi^i_\alpha = p_\alpha^i - \derpar{\hat{L}}{u_i^\alpha}
+ \sum_{j=1}^{m} \frac{1}{n(ij)} \, \frac{d}{dx^j} \, \derpar{\hat{L}}{u_{1_i+1_j}^\alpha}
\quad ; \quad \xi^I_\alpha = p_\alpha^I - \derpar{\hat{L}}{u_I^\alpha} \, ,
$$
then the tangency condition of $X$ along $\W_\Lag$, which is
$\Lie(X)(\xi^i_\alpha) = \Lie(X)(\xi^I_\alpha) = 0$ (on $\W_\Lag$), gives the following relation
on the component functions of $X$
\begin{align*}
G_\alpha^i &=
f^k\left( \derpars{\hat{L}}{x^k}{u_i^\alpha} - \frac{1}{n(ij)} \, \frac{d}{dx^j} \, \derpar{p_\alpha^{1_i+1_j}}{x^k} \right)
+ F^\beta\left( \derpars{\hat{L}}{u^\beta}{u_i^\alpha} - \frac{1}{n(ij)} \, \frac{d}{dx^j} \, \derpar{p_\alpha^{1_i+1_j}}{u^\beta} \right) \\
&\ {} + F_k^\beta\left( \derpars{\hat{L}}{u_k^\beta}{u_i^\alpha} - \frac{1}{n(ij)} \, \frac{d}{dx^j} \, \derpar{p_\alpha^{1_i+1_j}}{u_k^\beta} \right)
+ F^\beta_I\left( \derpars{\hat{L}}{u^\beta_I}{u_i^\alpha} - \frac{1}{n(ij)} \, \frac{d}{dx^j} \, \derpar{p_\alpha^{1_i+1_j}}{u^\beta_I} \right) \\
&\ {} - \frac{1}{n(ij)} \left( F_j^\beta \derpar{p_\alpha^{1_i+1_j}}{u^\beta}
+ F_{1_k+1_j}^\beta\derpar{p_\alpha^{1_i+1_j}}{u_k^\beta}
+ F_{I+1_j}^\beta\derpar{p_\alpha^{1_i+1_j}}{u_I^\beta} \right) \, , \\
G_\alpha^I &= f^i\derpars{\hat{L}}{x^i}{u_I^\alpha} + F^\beta\derpars{\hat{L}}{u^\beta}{u_I^\alpha}
+ F_i^\beta\derpars{\hat{L}}{u_i^\beta}{u_I^\alpha} + F_J^\beta\derpars{\hat{L}}{u_J^\beta}{u_I^\alpha} \, .
\end{align*}
Hence, the tangency condition enables us to write the component functions $G_\alpha^i$, $G_\alpha^I$
as functions $\widetilde{G}_\alpha^i$, $\widetilde{G}_\alpha^I$ depending on the rest of the components
$f^i,F^\alpha,F^\alpha_i,F^\alpha_I,F^\alpha_J$.

Now, if $\psi(x^i) = (x^i,u^\alpha,u^\alpha_i,u^\alpha_I,u^\alpha_J,p_\alpha^i,p_\alpha^I)$,
then the equation \eqref{eqn:UnifDynEqSect} gives in coordinates
\begin{align*}
\psi^*i(X)\Omega_r
&= \left[ f^k ( \cdots )
+ F^\alpha\left( \derpar{p_\alpha^i}{x^i} - \derpar{\hat{L}}{u^\alpha} \right)
+ F_i^\alpha\left( \frac{1}{n(ij)} \,\derpar{p_\alpha^{1_i+1_j}}{x^j} + p_\alpha^i - \derpar{\hat{L}}{u_i^\alpha} \right)\right. \\
&\quad {} + F_{I}^\alpha \left( p_\alpha^{I} - \derpar{\hat{L}}{u_{I}^\alpha} \right)
+ G_\alpha^i\left( -\derpar{u^\alpha}{x^i} + u_i^\alpha \right) \\
&\quad {} + \left. G_\alpha^{I} \left( u_{I}^\alpha - \sum_{1_i+1_j=I} \frac{1}{n(ij)} \, \derpar{u_i^\alpha}{x^j} \right) \right] d^mx \, .
\end{align*}
where the terms $(\cdots)$ contain a long expression with several partial derivatives of the component
functions and the Lagrangian function, which is not relevant in this proof. On the other hand, if we
take a vector field $Y$ tangent to $\W_\Lag$, then we must replace the component functions $G_\alpha^i$
and $G_\alpha^I$ by $\widetilde{G}_\alpha^i$ and $\widetilde{G}_\alpha^I$ in the previous equation, thus obtaining
\begin{align*}
\psi^*i(Y)\Omega_r
&= \left[ f^k ( \cdots )
+ F^\alpha\left( \derpar{p_\alpha^i}{x^i} - \derpar{\hat{L}}{u^\alpha} \right)
+ F_i^\alpha\left( \frac{1}{n(ij)} \,\derpar{p_\alpha^{1_i+1_j}}{x^j} + p_\alpha^i - \derpar{\hat{L}}{u_i^\alpha} \right)\right. \\
&\quad {} + F_{I}^\alpha \left( p_\alpha^{I} - \derpar{\hat{L}}{u_{I}^\alpha} \right)
+ \widetilde{G}_\alpha^i\left( -\derpar{u^\alpha}{x^i} + u_i^\alpha \right) \\
&\quad {} + \left. \widetilde{G}_\alpha^{I} \left( u_{I}^\alpha - \sum_{1_i+1_j=I} \frac{1}{n(ij)} \, \derpar{u_i^\alpha}{x^j} \right) \right] d^mx \, .
\end{align*}
Finally, if $\psi$ is holonomic, then equations \eqref{eqn:UnifDynEqSectHolonomy} are satisfied,
and the last two terms of both $i(X)\Omega_r$ and $i(Y)\Omega_r$ vanish, thus obtaining
\begin{align*}
\psi^*i(X)\Omega_r
&= \left[ f^k ( \cdots )
+ F^\alpha\left( \derpar{p_\alpha^i}{x^i} - \derpar{\hat{L}}{u^\alpha} \right)
+ F_i^\alpha\left( \frac{1}{n(ij)} \,\derpar{p_\alpha^{1_i+1_j}}{x^j} + p_\alpha^i - \derpar{\hat{L}}{u_i^\alpha} \right)\right. \\
&\quad {} + \left. F_{I}^\alpha \left( p_\alpha^{I} - \derpar{\hat{L}}{u_{I}^\alpha} \right) \right] d^mx = \psi^*i(Y)\Omega_r \, .
\end{align*}
Hence, we have $i(X)\Omega_r = 0$ if, and only if, $i(Y)\Omega_r = 0$.
\end{proof}

\noindent\textbf{Remark:}
Observe that, contrary to first-order field theories \cite{art:Echeverria_Lopez_Marin_Munoz_Roman04},
the holonomy condition is not recovered from the coordinate expression of the field equations.
Moreover, in this case, unlike in higher-order time-depending mechanical systems
\cite{art:Prieto_Roman12_1}, not even a condition for the holonomy of type $2$ can be obtained.
This is due to the constraints $p_\alpha^{ij} - p_\alpha^{ji} = 0$ introduced in Section
\ref{sec:SymmetricMultimomenta} to define both the extended and restricted $2$-symmetric
multimomentum bundles. Hence, the full holonomy condition is necessarily required in this formalism.

It is important to point out that, although the holonomy condition cannot be obtained from the
field equation, a holonomic section $\psi \in \Gamma(\rho_M^r)$ satisfies equations
\eqref{eqn:UnifDynEqSectHolonomy}. Hence, a holonomic section can be a solution to the equation
\eqref{eqn:UnifDynEqSect}.

\noindent\textbf{Remark:}
The regularity of the Lagrangian density seems to play a secondary role in this formulation, because
the holonomy of the section solution to the equation \eqref{eqn:UnifDynEqSect} is necessarily required,
regardless of the regularity of the Lagrangian density given. Nevertheless, recall that the
Euler-Lagrange equations \eqref{eqn:EulerLagrangeUnified} may not be compatible if the Lagrangian
density is singular, and thus the regularity of $\Lag$ still determines if the section
$\psi \in \Gamma(\rho_M^r)$ solution to the equation \eqref{eqn:UnifDynEqSect} lies in
$\W_\Lag$ or in a submanifold of $\W_\Lag$. If $\Lag$ is
singular, in the most favourable cases, there exists a submanifold
$\W_f \hookrightarrow \W_\Lag$ where the section $\psi$ takes values.

\subsection{Field equations for multivector fields}
\label{sec:UnifFieldEquationsMultiVF}

The \emph{Lagrangian-Hamiltonian problem for multivector fields} associated with the
premultisymplectic manifold $(\W_r,\Omega_r)$ consists in finding a class of locally
decomposable holonomic multivector fields $\{\X\} \subset \vf^m(\W_r)$
satisfying the following field equation
\begin{equation}\label{eqn:UnifDynEqMultiVF}
i(\X)\Omega_r = 0 \ , \quad \mbox{for every } \X \in \{\X\} \, .
\end{equation}

According to \cite{art:deLeon_Marin_Marrero_Munoz_Roman05}, we have the following result.

\begin{proposition}\label{prop:CompSubmanifoldMultiVF}
A solution $\X \in \vf^{m}(\W_r)$ to equation \eqref{eqn:UnifDynEqMultiVF}
exists only on the points of the submanifold $\W_c \hookrightarrow \W_r$ defined by
\begin{align*}
\W_c &= \left\{ w \in \W_r \colon (i(Z)d\hat{H})(w) = 0 \, , \mbox{ for every }
Z \in \ker(\Omega) \right\} \\
&= \left\{ w \in \W_r \colon (i(Y)\Omega_r)(w) = 0 \, , \mbox { for every }
Y \in \vf^{V(\rho_2^r)}(\W_r)\right\} \, .
\end{align*}
\end{proposition}

The submanifold $\W_c \hookrightarrow \W_r$ is the so-called \emph{compatibility
submanifold} for the premultisymplectic system $(\W_r,\Omega_r)$. Observe that we denoted
this submanifold by $\W_c$, which is the notation used for the first contraint submanifold
defined in \eqref{eqn:CompSubmanifoldSect}. Indeed, both submanifolds are equal. In order to prove
this, recall that the first constraint submanifold is defined locally by the constraints
$p^I_\alpha - \partial\hat{L}/u_I^\alpha = 0$. Hence, it suffices to prove that the compatibility
submanifold given in Proposition \ref{prop:CompSubmanifoldMultiVF} is defined locally by the same contraints.

In fact, in natural coordinates, the coordinate expression for the local Hamiltonian function
$\hat{H}$ is given by \eqref{eqn:HamiltonianFunctionLocal}, and thus we have
\begin{align*}
d\hat{H} &= u_i^\alpha d p_\alpha^i + p_\alpha^i d u_i^\alpha + u_I^\alpha d p_\alpha^I + p_\alpha^I d u_I^\alpha
- \left( \derpar{\hat{L}}{u^\alpha}d u^\alpha + \derpar{\hat{L}}{u_i^\alpha}d u_i^\alpha + \derpar{\hat{L}}{u_I^\alpha}d u_I^\alpha \right) \\
&= -\derpar{\hat{L}}{u^\alpha}d u^\alpha + \left( p_\alpha^i - \derpar{\hat{L}}{u_i^\alpha} \right)d u_i^\alpha
+ \left( p_\alpha^I - \derpar{\hat{L}}{u_I^\alpha} \right) d u_I^\alpha + u_i^\alpha d p_\alpha^i + u_I^\alpha d p_\alpha^I \, .
\end{align*}
Now, bearing in mind that $\ker\Omega$ is the $(nm(m+1)/2 + nm(m+1)(m+2)/6)$-dimensional
$C^\infty(\W)$-module locally given by \eqref{eqn:PremultisymplecticKernelLocal},
the functions $i(Z)d\hat{H}$ for $Z \in \ker\Omega$ have the following coordinate expressions
$$
i\left( \derpar{}{u_I^\alpha} \right) d\hat{H} = p_\alpha^I - \derpar{\hat{L}}{u_I^\alpha} \ \mbox{ for }|I| = 2 \quad ; \quad
i\left( \derpar{}{u_J^\alpha} \right) d\hat{H} = 0 \ \mbox{ for } |J| = 3 \, .
$$
Therefore, the submanifold $\W_c \hookrightarrow \W_r$ is locally defined by the $nm(m+1)/2$
constraints $p_\alpha^I - \partial\hat{L} / \partial u_I^\alpha = 0$. In particular, it is equal
to the submanifold defined in \eqref{eqn:CompSubmanifoldSect}, and we have
$$
\dim\W_c = \dim\W_r - nm(m+1)/2 = m + n + 2mn + nm(m+1)/2 + nm(m+1)(m+2)/6 \ .
$$

Now we compute the coordinate expression of the equation \eqref{eqn:UnifDynEqMultiVF} in a local
chart of $\W_r$. From the results in \cite{art:Echeverria_Munoz_Roman98}, a representative $\X$
of a class of locally decomposable, integrable and $\rho_M^r$-transverse $m$-vector fields
$\{\X\} \subset \vf^m(\W_r)$ can be written in coordinates
\begin{equation}\label{eqn:UnifGenericMultiVFLocal}
\X = f \bigwedge_{j=1}^{m}
\left(  \derpar{}{x^j} + F_j^\alpha\derpar{}{u^\alpha} +
\sum_{|I| = 1}^{3} F_{I,j}^\alpha\derpar{}{u_{I}^\alpha}
+ G_{\alpha,j}^i\derpar{}{p_\alpha^i} + G_{\alpha,j}^{I}\derpar{}{p_\alpha^{I}} \right) \, ,
\end{equation}
where $f$ is a non-vanishing local function. Taking $f = 1$ as a representative of
the equivalence class, the equation \eqref{eqn:UnifDynEqMultiVF} gives the following system of equations
\begin{align}
& F_j^\alpha = u_j^\alpha \quad ; \quad \sum_{1_i+1_j=I} \frac{1}{n(ij)} \, F_{i,j}^\alpha = u_I^\alpha \, , \label{eqn:UnifDynEqMultiVFHolonomy} 
\end{align}\begin{align}
& \sum_{i=1}^{m} G_{\alpha,i}^{i} = \derpar{\hat L}{u^\alpha} \, , \label{eqn:UnifDynEqMultiVFLocal1} \\
& \sum_{j=1}^{m} \frac{1}{n(ij)} \, G_{\alpha,j}^{1_i+1_j} = \derpar{\hat L}{u_i^\alpha} - p_\alpha^i \, , \label{eqn:UnifDynEqMultiVFLocal2} \\
& p_\alpha^K = \derpar{\hat L}{u^\alpha_K} \, , \quad |K| = 2 \, . \label{eqn:HOMomentaMultiVF}
\end{align}
The $m$ additional equations alongside the $d x^i$ are a straightforward consequence of the others
and the tangency condition that follows, and thus we omit them. Therefore, the multivector field $\X$
is locally given by
$$
\X = \bigwedge_{j=1}^{m}
\left(  \derpar{}{x^j} + u_j^\alpha\derpar{}{u^\alpha} +
\sum_{|I| = 1}^{3} F_{I,j}^\alpha\derpar{}{u_{I}^\alpha}
+ G_{\alpha,j}^i\derpar{}{p_\alpha^i} + G_{\alpha,j}^{I}\derpar{}{p_\alpha^{I}} \right) \, ,
$$
where the functions $F_{i,j}^\alpha$, $G_{\alpha,j}^i$ and $G_{\alpha,j}^I$ must satisfy the
equations \eqref{eqn:UnifDynEqMultiVFHolonomy}, \eqref{eqn:UnifDynEqMultiVFLocal1} and
\eqref{eqn:UnifDynEqMultiVFLocal2}. Note that most of the component functions remain undetermined,
and that there can be several different functions satisfying the referred equations. However, recall
that the statement of the problem requires the class of multivector fields to be holonomic.
In coordinates, this implies that equations \eqref{eqn:MultiVFHolonomyLocal} are satisfied with
$k = 3$ and $r=1$, and thus the multivector field $\X$ has the following coordinate expression
$$
\X = \bigwedge_{j=1}^{m}
\left(  \derpar{}{x^j} + u_j^\alpha\derpar{}{u^\alpha} +
\sum_{|I| = 1}^{2} u_{I+1_j}^\alpha \derpar{}{u_{I}^\alpha}
+ F_{J,j}^\alpha\derpar{}{u_{J}^\alpha} + G_{\alpha,j}^i\derpar{}{p_\alpha^i} + G_{\alpha,j}^{I}\derpar{}{p_\alpha^{I}} \right) \, ,
$$
with $G_{\alpha,j}^i$ and $G_{\alpha,j}^{I}$ satisfying \eqref{eqn:UnifDynEqMultiVFLocal1}
and \eqref{eqn:UnifDynEqMultiVFLocal2}.

Observe that the equations \eqref{eqn:HOMomentaMultiVF} are a compatibility condition for the
multivector field $\X$, which state that the multivector field solution to the field equation
\eqref{eqn:UnifDynEqMultiVF} exists only at support on the submanifold $\W_c$. Hence, we recover
in coordinates the result stated in Proposition \ref{prop:CompSubmanifoldMultiVF}.

Let us analyze the tangency of the multivector field $\X$ along the submanifold
$\W_c \hookrightarrow \W_r$. From \cite{art:Echeverria_Munoz_Roman98} we know that the necessary
and sufficient condition for $\X = X_1 \wedge \ldots \wedge X_m \in \vf^{m}(\W_r)$ to be
tangent to $\W_c$ is that $X_j$ is tangent to $\W_{c}$ for every $j=1,\ldots,m$.

Therefore, since the submanifold $\W_c \hookrightarrow \W_r$ is locally defined by the constraint
functions $\xi^{K}_{\alpha} = p_\alpha^{K} - \partial\hat{L} / \partial u_{K}^\alpha$, we must
check if the condition $\Lie(X_j)(\xi_\alpha^{K}) \equiv X_j(\xi_{\alpha}^{K}) = 0$ holds on $\W_c$
for every $1 \leqslant j \leqslant m$, $1 \leqslant \alpha \leqslant n$, $|K| = 2$. Computing, we obtain
\begin{align*}
&\left( \derpar{}{x^j} + u_j^\alpha\derpar{}{u^\alpha} +
\sum_{|I| = 1}^{2} u_{I+1_j}^\alpha \derpar{}{u_{I}^\alpha}
+ G_{\alpha,j}^i\derpar{}{p_\alpha^i} + G_{\alpha,j}^{I}\derpar{}{p_\alpha^{I}}\right)\left( p_\alpha^K - \derpar{\hat{L}}{u_K^\alpha} \right) = 0 \\
&\qquad \Longleftrightarrow
G_{\alpha,j}^K - \derpars{\hat{L}}{x^j}{u_K^\alpha} - u_j^\beta\derpars{\hat{L}}{u^\beta}{u_K^\alpha}
- u_{1_i+1_j}^\beta \derpars{\hat{L}}{u_i^\beta}{u_K^\alpha} - u_{I+1_j}^\beta\derpars{\hat{L}}{u_I^\beta}{u_K^\alpha} = 0 \\
&\qquad \Longleftrightarrow
G_{\alpha,j}^K - \frac{d}{dx^j}\, \derpar{\hat{L}}{u_K^\alpha} = 0 \, .
\end{align*}
Hence, the tangency condition enables us to determinate all the functions $G_{\alpha,j}^K$, since we
obtain $nm^2(m+1)/2$ equations, one for each function. Now, taking into account equations
\eqref{eqn:UnifDynEqMultiVFLocal2} and the coefficients $G_{\alpha,j}^K$ that we have determined, we obtain
$$
\sum_{j=1}^{m} \frac{1}{n(ij)} \, G_{\alpha,j}^{1_i+1_j} - \derpar{\hat L}{u_i^\alpha} + p_\alpha^i = 0
\ \Longleftrightarrow \
p_\alpha^i - \derpar{\hat{L}}{u_i^\alpha} + \sum_{j=1}^{m} \frac{1}{n(ij)} \, \frac{d}{dx^j} \, \derpar{\hat{L}}{u_{1_i+1_j}^\alpha} = 0 \, .
$$
Hence, the tangency condition for the multivector field $\X$ along $\W_{c}$ gives rise to $mn$ new
constraints defining a submanifold of $\W_{c}$ that coincides with the submanifold $\W_\Lag$
introduced in Proposition \ref{prop:GraphLegMapSect}. Now we must study the tangency of $\X$ along
the new submanifold $\W_\Lag$. After a long but straightforward calculation, we obtain
\begin{align*}
G_{\alpha,k}^{i} &=
\frac{d}{dx^k} \, \derpar{\hat{L}}{u_i^\alpha}
- \frac{d}{dx^k} \sum_{j=1}^{m} \frac{1}{n(ij)} \, \frac{d}{dx^j} \, \derpar{\hat{L}}{u_{1_i+1_j}^\alpha} \\
&\quad{} - \sum_{j=1}^{m}\frac{1}{n(ij)}\left( F_{I+1_j,k}^\beta - \frac{d}{dx^k} \, u_{I+1_j}^\beta \right) \derpars{\hat{L}}{u_I^\beta}{u_{1_i+1_j}^\alpha} \, .
\end{align*}
Therefore, the tangency condition along the submanifold $\W_\Lag$ enables us to determinate all the functions
$G_{\alpha,k}^{i}$. Now, taking into account equations \eqref{eqn:UnifDynEqMultiVFLocal1}, we have
\begin{align*}
\sum_{i=1}^{m} G_{\alpha,i}^{i} - \derpar{\hat{L}}{u^\alpha} = 0 & \Longleftrightarrow
\derpar{\hat{L}}{u^\alpha}
- \frac{d}{dx^i} \, \derpar{\hat{L}}{u_i^\alpha}
+ \sum_{|I|=2} \frac{d^{|I|}}{dx^{I}} \, \derpar{\hat{L}}{u_{I}^\alpha} \\
&\qquad{} + \sum_{i=1}^{m}\sum_{j=1}^{m}\frac{1}{n(ij)}\left( F_{I+1_j,i}^\beta - \frac{d}{dx^i} \, u_{I+1_j}^\beta \right) \derpars{\hat{L}}{u_I^\beta}{u_{1_i+1_j}^\alpha} = 0 \, .
\end{align*}
These $n$ equations are the Euler-Lagrange equations for a locally decomposable
holonomic multivector field. Observe that if $\hat{\Lag}$ is a regular Lagrangian density,
then the Hessian of $\hat{L}$ with respect to the second-order velocities is regular, and we can
assure the existence of a local multivector field $\X$ solution to the equation
\eqref{eqn:UnifDynEqMultiVF}, defined at support on
$\W_\Lag \hookrightarrow \W_r$ and tangent to $\W_\Lag$.
A global solution is then obtained using partitions of the unity.

If the Lagrangian density is not regular, then the above equations may or may not be compatible,
and may give rise to new constraints. In the most favourable cases there exists a submanifold
$\W_f \hookrightarrow \W_\Lag$ (where we admit
$\W_f = \W_\Lag)$ where we have a well-defined holonomic
multivector field at support on $\W_f$, and tangent to $\W_f$, solution to the equation
\begin{equation}\label{eqn:UnifDynEqMultiVFSing}
\restric{i(\X)\Omega_r}{\W_f} = 0 \, .
\end{equation}

Therefore, we can state the following result.

\begin{theorem}\label{thm:EquivalenceTheoremUnified}
The following assertions on a holonomic section $\psi \in \Gamma(\rho_M^r)$ are equivalent:
\begin{enumerate}
\item $\psi$ is a solution to the equation \eqref{eqn:UnifDynEqSect}, that is,
$$
\psi^*i(X)\Omega_r = 0 \, , \quad \mbox{for every } X \in \vf(\W_r) \, .
$$
\item If the coordinate expression of $\psi$ in the induced natural coordinates of $\W_r$ is
$\psi(x^i) = (x^i,u^\alpha(x^i),u^\alpha_j(x^i),u^\alpha_I(x^i),u^\alpha_J(x^i),p_\alpha^{j}(x^i),p_\alpha^{I}(x^i))$,
then the component functions of $\psi$ satisfy equations \eqref{eqn:UnifDynEqSectLocal}
and \eqref{eqn:UnifDynEqSectRelationMomenta}, that is, the following system of $n + nm$
partial differential equations
\begin{equation}\label{eqn:EquivalenceTheoremUnifiedLocal}
\sum_{i=1}^{m}\derpar{p_\alpha^i}{x^i} = \derpar{\hat{L}}{u^\alpha} \quad ; \quad
\sum_{j=1}^{m} \frac{1}{n(ij)} \, \derpar{p_\alpha^{1_i+1_j}}{x^j} = \derpar{\hat{L}}{u_i^\alpha} - p_\alpha^i \, .
\end{equation}
\item $\psi$ is a solution to the equation
\begin{equation}\label{eqn:UnifDynEqIntSect}
i(\Lambda^m\psi^\prime)(\Omega_r \circ \psi) = 0 \, ,
\end{equation}
where $\Lambda^m\psi^\prime \colon M \to \Lambda^mT\W_r$ is the canonical lifting of $\psi$.
\item $\psi$ is an integral section of a multivector field contained in a class of locally
decomposable holonomic multivector fields $\{ \X \} \subset \vf^{m}(\W_r)$, tangent to $\W_\Lag$,
and satisfying the equation \eqref{eqn:UnifDynEqMultiVF}, that is,
$$
i(\X)\Omega_r = 0 \, .
$$
\end{enumerate}
\end{theorem}
\begin{proof} \

($1 \Longleftrightarrow 2$) \quad
From the results in Section \ref{sec:UnifFieldEquationsSect}, the field equation
\eqref{eqn:UnifDynEqSect} gives in coordinates the equations \eqref{eqn:UnifDynEqSectLocal},
\eqref{eqn:UnifDynEqSectRelationMomenta}, \eqref{eqn:HOMomentaSect} and \eqref{eqn:UnifDynEqSectHolonomy}.
As stated in the aforementioned Section, the equations \eqref{eqn:HOMomentaSect} are the local
constraints defining the first constraint submanifold $\W_c \hookrightarrow \W_r$. In addition,
since we assume that the section $\psi \in \Gamma(\rho_M^r)$ is holonomic, the equations
\eqref{eqn:UnifDynEqSectHolonomy} are satisfied. Therefore, the equation \eqref{eqn:UnifDynEqSect}
is locally equivalent to equations \eqref{eqn:UnifDynEqSectLocal} and
\eqref{eqn:UnifDynEqSectRelationMomenta}, that is, to equations \eqref{eqn:EquivalenceTheoremUnifiedLocal}.

($2 \Longleftrightarrow 3$) \quad
If $\psi \in \Gamma(\rho_M^r)$ is locally given by
\begin{equation*}
\psi(x^i) = (x^i,u^\alpha(x^i),u^\alpha_i(x^i),u^\alpha_I(x^i),u^\alpha_J(x^i),p_\alpha^j(x^i),p_\alpha^I(x^i)) \, ,
\end{equation*}
then its canonical lifting to $\Lambda^mT\W_r$ is locally given by
$\Lambda^m\psi^\prime = \psi^\prime_1 \wedge \ldots \wedge \psi_m^\prime$, with
\begin{equation*}
\psi^\prime_j = \left(0,\ldots,0,1,0,\ldots,0,\frac{d}{dx^j}u^\alpha,\frac{d}{dx^j}u^\alpha_i,
\frac{d}{dx^j}u^\alpha_I,\frac{d}{dx^j}u^\alpha_J,\frac{d}{dx^j}p_\alpha^i,\frac{d}{dx^j}p_\alpha^I\right) \, ,
\end{equation*}
where $d/dx^j$ is the $j$th coordinate total derivative, and the $1$ is at the $j$th position. Then,
the inner product $i(\Lambda^m\psi^\prime)(\Omega_r \circ \psi)$ gives, in coordinates,
\begin{align*}
i(\Lambda^m\psi^\prime)(\Omega_r \circ \psi) &=
\sum_{i=1}^{m}\left( \cdots \right)d x^i
+ \left( \derpar{\hat{L}}{u^\alpha} - \frac{dp_\alpha^i}{dx^i} \right) d u^\alpha \\
&\quad{} + \left( \derpar{\hat{L}}{u_i^\alpha} - p_\alpha^{i} - \sum_{j=1}^{m} \frac{1}{n(ij)} \frac{dp_\alpha^{1_i+1_j}}{dx^j} \right)d u_i^\alpha
+ \left( p_\alpha^I - \derpar{\hat{L}}{u_I^\alpha} \right)d u_I^\alpha \\
&\quad{} + \left( \frac{du^\alpha}{dx^i} - u_i^\alpha \right)d p_\alpha^i +
\left( \sum_{1_i+1_j=I} \frac{1}{n(ij)} \frac{du^\alpha_i}{dx^j} - u_I^\alpha \right) d p_\alpha^I \, ,
\end{align*}
where the terms $(\cdots)$ along the forms $d x^i$ involve of partial derivatives of the Lagrangian
function and of the rest of component functions. Now, requiring this last expression to vanish, we obtain
equations \eqref{eqn:UnifDynEqSectLocal}, \eqref{eqn:UnifDynEqSectRelationMomenta}, \eqref{eqn:HOMomentaSect}
and \eqref{eqn:UnifDynEqSectHolonomy}, along with $m$ additional equations which are a combination of those.
Same comments as in the proof of the previous item apply. In particular,
equations \eqref{eqn:HOMomentaSect} are the local constraints defining the first constraint submanifold
$\W_c \hookrightarrow \W_r$, and equations \eqref{eqn:UnifDynEqSectHolonomy} are automatically satisfied
because of the holonomy assumption. Hence, equation \eqref{eqn:UnifDynEqIntSect}
is locally equivalent to equations \eqref{eqn:UnifDynEqSectLocal} and
\eqref{eqn:UnifDynEqSectRelationMomenta}, that is, to equations \eqref{eqn:EquivalenceTheoremUnifiedLocal}.

($2 \Longleftrightarrow 4$) \quad
From the results in this Section, if $\X \in \vf^{m}(\W_r)$ is a generic
locally decomposable multivector field locally given by \eqref{eqn:UnifGenericMultiVFLocal}, then,
taking $f = 1$ as a representative of the equivalence class, the field equation \eqref{eqn:UnifDynEqMultiVF}
is locally equivalent to the equations \eqref{eqn:UnifDynEqMultiVFHolonomy},
\eqref{eqn:UnifDynEqMultiVFLocal1}, \eqref{eqn:UnifDynEqMultiVFLocal2} and \eqref{eqn:HOMomentaMultiVF}.
As already stated, equations \eqref{eqn:HOMomentaMultiVF} give, in coordinates, the compatibility
submanifold $\W_c$ obtained using the constraint algorithm in \cite{art:deLeon_Marin_Marrero_Munoz_Roman05}.
On the other hand, since the multivector field $\X$ is assumed to be holonomic, then equations
\eqref{eqn:UnifDynEqMultiVFHolonomy} are satisfied. Hence, the field equation
\eqref{eqn:UnifDynEqMultiVF} is locally equivalent to equations \eqref{eqn:UnifDynEqMultiVFLocal1}
and \eqref{eqn:UnifDynEqMultiVFLocal2}.

Let $\gamma \in \Gamma(\rho_M^r)$ be an integral section of $\X$ given in the natural coordinates of $\W$ by
$\gamma(x^i) = (x^i,u^\alpha,u^\alpha_i,u^\alpha_I,u^\alpha_J,p_\alpha^{i},p_\alpha^{I})$. Then,
the condition of integral section is locally equivalent to the following system of equations
\begin{gather*}
\derpar{u^\alpha}{x^i} = F_i^\alpha \circ \gamma \quad ; \quad
\derpar{u^\alpha_i}{x^j} = F_{i,j}^\alpha \circ \gamma \quad ; \quad
\derpar{u^\alpha_I}{x^j} = F_{I,j}^\alpha \circ \gamma \quad ; \quad
\derpar{u^\alpha_J}{x^j} = F_{J,j}^\alpha \circ \gamma \, , \\
\derpar{p_\alpha^i}{x^j} = G_{\alpha,j}^i \circ \gamma \quad ; \quad
\derpar{p_\alpha^I}{x^j} = G_{\alpha,j}^I \circ \gamma \, .
\end{gather*}
Replacing these equations in  \eqref{eqn:UnifDynEqMultiVFLocal1} and \eqref{eqn:UnifDynEqMultiVFLocal2},
we obtain the following system of partial differential equations for the component functions of $\gamma$
\begin{gather*}
\derpar{u^\alpha}{x^i} = u_i^\alpha \quad ; \quad
\derpar{u^\alpha_i}{x^j} = u_{1_i+1_j}^\alpha \quad ; \quad
\derpar{u^\alpha_I}{x^j} = u_{I+1_j}^\alpha \quad ; \quad
\derpar{u^\alpha_J}{x^j} = F_{J,j}^\alpha \, , \\
\sum_{i=1}^{m}\derpar{p_\alpha^i}{x^i} = \derpar{\hat{L}}{u^\alpha} \quad ; \quad
\sum_{j=1}^{m} \frac{1}{n(ij)} \, \derpar{p_\alpha^{1_i+1_j}}{x^j} = \derpar{\hat{L}}{u_i^\alpha} - p_\alpha^i \, .
\end{gather*}
Since the multivector field $\X$ is holonomic and tangent to $\W_\Lag$, the first equations are
identically satisfied. Thus, the condition of $\gamma$ to be an integral section of a locally
decomposable holonomic multivector field $\X \in \vf^{m}(\W_r)$, tangent to $\W_\Lag$, and
satisfying the equation \eqref{eqn:UnifDynEqMultiVF} is locally equivalent to equations
\eqref{eqn:EquivalenceTheoremUnifiedLocal}.
\end{proof}

\section{Lagrangian formalism}
\label{sec:LagrangianForm}

Now we recover the Lagrangian field equations and geometric structures from the unified formalism.
The results remain the same for both regular and singular Lagrangian densities. Thus, no distinction
will be made in this matter.

\subsection{General setting}

In order to establish the field equations in the Lagrangian formalism, we must define the
Poincar\'{e}-Cartan $m$ and $(m+1)$-forms in $J^3\pi$. Since a unique Legendre map is recovered
in the unified framework, we can give the following definition:

\begin{definition}\label{def:PoincareCartanForms}
Let $\Theta_1^s \in \Omega^{m}(J^2\pi^\dagger)$ and $\Omega_1^s \in \Omega^{m+1}(J^2\pi^\dagger)$ be the
symmetrized Liouville forms in $J^2\pi^\dagger$. The \emph{Poincar\'{e}-Cartan forms} in
$J^3\pi$ are the forms defined as $\Theta_\Lag = \widetilde{\Leg}^*\Theta_1^s \in \Omega^{m}(J^3\pi)$ and
$\Omega_\Lag = \widetilde{\Leg}^*\Omega_1^s = -d\Theta_\Lag \in \Omega^{m+1}(J^3\pi)$.
\end{definition}

These forms coincide with the usual Poincar\'{e}-Cartan forms for second-order classical field
theories that can be found in the literature (see, for instance, \cite{art:Aldaya_Azcarraga78_2,proc:Garcia_Munoz83,art:Kouranbaeva_Shkoller00,art:Munoz85}).
They can also be recovered directly from the unified formalism. In fact:

\begin{lemma}\label{lemma:LagRelatedForms}
Let $\Theta = \rho_2^*\Theta_1^s$ and $\Theta_r = \hat{h}^*\Theta$ be the canonical $m$-forms
defined in $\W$ and $\W_r$, respectively. Then, the Poincar\'{e}-Cartan $m$-form satisfies
$\Theta = \rho_1^*\Theta_\Lag$ and $\Theta_r = (\rho_1^r)^*\Theta_\Lag$.
\end{lemma}
\begin{proof}
A straightforward computation leads to this result. For the first statement we have
$$
\rho_1^*\Theta_\Lag = \rho_1^*(\widetilde{\Leg}^*\Theta_1^s)
= (\widetilde{\Leg} \circ \rho_1)^*\Theta_1^s = \rho_2^*\Theta_1^s = \Theta\, ,
$$
and from this the second statement follows:
\begin{equation*}
(\rho_1^r)^*\Theta_\Lag = (\rho_1 \circ \hat{h})^*\Theta_\Lag
= \hat{h}^*(\rho_1^*\Theta_\Lag) = \hat{h}^*\Theta = \Theta_r \, . \qedhere
\end{equation*}
\end{proof}

Observe that, as the pull-back of a form by a function and the exterior derivative
commute, this result also holds for the Poincar\'{e}-Cartan $(m+1)$-form $\Omega_\Lag$.

Using the natural coordinates $(x^i,u^\alpha,u_i^\alpha,u_I^\alpha,u_J^\alpha)$ in $J^3\pi$, and
bearing in mind the local expression \eqref{eqn:LiouvilleSymmetricFormsLocal} of $\Theta_1^s$,
and \eqref{eqn:ExtendedLegendreMapLocal} of the extended Legendre map, the local expression of
the Poincar\'{e}-Cartan $m$-form is
\begin{align*}
\Theta_\Lag &= \left( \derpar{L}{u_i^\alpha} - \sum_{j=1}^{m}\frac{1}{n(ij)} \, \frac{d}{dx^j} \, \derpar{L}{u_{1_i+1_j}^\alpha} \right)(d u^\alpha \wedge d^{m-1}x_i - u_i^\alpha d^mx) \\
&\qquad {} + \frac{1}{n(ij)} \, \derpar{L}{u_{1_i+1_j}^\alpha} \, ( d u_i^\alpha \wedge d^{m-1}x_j - u_{1_i+1_j}^\alpha d^mx ) + L d^mx \, .
\end{align*}

An important fact regarding the Poincar\'{e}-Cartan $(m+1)$-form $\Omega_\Lag$ is that
it is $1$-degenerate when $m > 1$, regardless of the regularity of the Lagrangian density.
Indeed, since the restricted Legendre map $\Leg \colon J^3\pi \to J^2\pi^\ddagger$ is a
submersion with $\dim J^3\pi > \dim J^2\pi^\ddagger$, and
$\textnormal{rank}(\Leg) = \textnormal{rank}(\widetilde{\Leg})$,
there exists a non-zero vector field $X \in \vf(J^3\pi)$ which is $\widetilde{\Leg}$-related
to $\mathbf{0} \in \vf(J^2\pi^\dagger)$, that is,
$T\widetilde{\Leg} \circ X = \mathbf{0} \circ \widetilde{\Leg}$.
Then, we have
$$
i(X)\Omega_\Lag = i(X)\widetilde{\Leg}^*\Omega_1^s = \widetilde{\Leg}^*i(\mathbf{0}) \Omega_1^s = 0 \, .
$$

\begin{proposition}\label{prop:LagDiffWL}
The map $\rho_1^\Lag = \rho_1^r \circ j_\Lag \colon \W_\Lag \to J^3\pi$ is a diffeomorphism.
\end{proposition}
\begin{proof}
Since $\rho_1^\Lag$ is a surjective submersion, the equality $\dim J^3\pi = \dim\W_\Lag$ implies
that it is also an injective immersion, and therefore a diffeomorphism.
\end{proof}

\subsection{Field equations for sections}

\begin{proposition}\label{prop:UnifToLagSect}
Let $\psi \in \Gamma(\rho_M^r)$ be a holonomic section solution to the equation
\eqref{eqn:UnifDynEqSect}. Then the section $\psi_\Lag = \rho_1^r \circ \psi \in \Gamma(\bar{\pi}^3)$
is holonomic, and is a solution to the equation
\begin{equation}\label{eqn:LagDynEqSect}
\psi_\Lag^*i(X)\Omega_\Lag = 0 \, , \quad \mbox{for every } X \in \vf(J^3\pi) \, .
\end{equation}
\end{proposition}
\begin{proof}
By definition, a section $\psi \in \Gamma(\rho_M^r)$ is holonomic if the section
$\rho_1^r \circ \psi \in \Gamma(\bar{\pi}^3)$ is holonomic. Hence, $\psi_\Lag = \rho_1^r \circ \psi$
is clearly a holonomic section.

Now, since $\rho_1^r \colon \W_r \to J^3\pi$ is a submersion, for every vector field
$X \in \vf(J^3\pi)$ there exist some vector fields $Y \in \vf(\W_r)$ such that $X$ and $Y$ are
$\rho_1^r$-related. Observe that this vector field $Y$ is not unique because the vector field
$Y + Y_o$, with $Y_o \in \ker T\rho_1^r$ is also $\rho_1^r$-related with $X$. Thus, using this
particular choice of $\rho_1^r$-related vector fields, we have
\begin{align*}
\psi_\Lag^*i(X)\Omega_\Lag &= (\rho_1^r \circ \psi)^*i(X)\Omega_\Lag
= \psi^*((\rho_1^r)^*i(X)\Omega_\Lag) \\
&= \psi^*i(Y)(\rho_1^r)^*\Omega_\Lag
= \psi^*i(Y)\Omega_r \, .
\end{align*}
Since the equality $\psi^*i(Y)\Omega_r = 0$ holds for every $Y \in \vf(\W_r)$, it holds, in
particular, for every $Y \in \vf(\W_r)$ which is $\rho_1^r$-related with $X \in \vf(J^3\pi)$.
Hence we obtain
\begin{equation*}
\psi_\Lag^*i(X)\Omega_\Lag = \psi^*i(Y)\Omega_r = 0 \, . \qedhere
\end{equation*}
\end{proof}

The following diagram illustrates the situation of the above Proposition:
$$
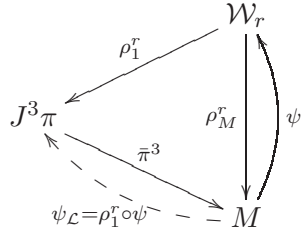
\xymatrix{
\ & \ & \W_r \ar[dll]_{\rho^r_1} \ar[dd]_{\rho^r_M} \\
J^3\pi \ar[drr]^{\bar{\pi}^3} & \ & \ \\
\ & \ & M \ar@/^1pc/@{-->}[ull]^{\psi_\Lag = \rho_1^r \circ \psi} \ar@/_1pc/[uu]_{\psi}
}
$$

Observe that Proposition \ref{prop:UnifToLagSect} states that every section solution to the field
equations in the unified formalism projects to a section solution to the field equations in the
Lagrangian formalism, but it does not establish an equivalence between the solutions. This
equivalence does exist, due to the fact that the map $\rho_1^\Lag \colon \W_\Lag \to J^3\pi$ is a
diffeomorphism. In order to establish this equivalence, we first need the following technical result.

\begin{lemma}\label{lemma:LagTechLemmaSect}
The Poincar\'{e}-Cartan forms defined in $J^3\pi$
satisfy the identities $(\rho_1^\Lag)^*\Theta_\Lag = j_\Lag^*\Theta_r$
and $(\rho_1^\Lag)^*\Omega_\Lag = j_\Lag^*\Omega_r$
\end{lemma}
\begin{proof}
Since the exterior derivative and the pull-back commute, it suffices to prove the statement for
the $m$-forms. We have
\begin{align*}
(\rho_1^\Lag)^*\Theta_\Lag &= (\rho_1^r\circ j_\Lag)^*\Theta_\Lag
= (\rho_1 \circ \hat{h} \circ j_\Lag)^*\Theta_\Lag
= (\rho_1 \circ \hat{h} \circ j_\Lag)^*(\widetilde{\Leg}^*\Theta_1^s) \\
&= (\widetilde{\Leg} \circ \rho_1 \circ \hat{h} \circ j_\Lag)^*\Theta_1^s
= (\rho_2 \circ \hat{h} \circ j_\Lag)^*\Theta_1^s
= (\hat{h} \circ j_\Lag)^*\Theta
= j_\Lag^*\Theta_r \, . \qedhere
\end{align*}
\end{proof}

Now we can state the remaining part of the equivalence between the solutions of the Lagrangian and
unified formalisms.

\begin{proposition}\label{prop:LagToUnifSect}
Let $\psi_\Lag \in \Gamma(\bar{\pi}^3)$ be a holonomic section solution to the field equation
\eqref{eqn:LagDynEqSect}. Then the section
$\psi = j_\Lag \circ (\rho_1^\Lag)^{-1} \circ \psi_\Lag \in \Gamma(\rho_M^r)$
is holonomic and it is a solution to the equation \eqref{eqn:UnifDynEqSect}.
\end{proposition}
\begin{proof}
By definition, a section $\psi \in \Gamma(\rho_M^r)$ is holonomic if the section
$\rho_1^r \circ \psi \in \Gamma(\bar{\pi}^3)$ is holonomic. Computing, we have
$$
\rho_1^r \circ \psi = \rho_1^r \circ j_\Lag \circ (\rho_1^\Lag)^{-1} \circ \psi_\Lag = \psi_\Lag \, ,
$$
since
$\rho_1^r \circ j_\Lag = \rho_1^\Lag \Leftrightarrow \rho_1^r \circ j_\Lag \circ (\rho_1^\Lag)^{-1} =
\textnormal{Id}_{J^3\pi}$.
Hence, since $\psi_\Lag$ is holonomic, the section $\psi = j_\Lag \circ (\rho_1^\Lag)^{-1} \circ \psi_\Lag$
is holonomic in $\W_r$.

Now, since $j_\Lag \colon \W_\Lag \to \W_r$ is an embedding, for every vector field $X \in \vf(\W_r)$
tangent to $\W_\Lag$, there exists a unique vector field $Y \in \vf(\W_\Lag)$ which is $j_\Lag$-related
with $X$. Hence, let us assume that $X \in \vf(\W_r)$ is tangent to $\W_\Lag$. Then we have
$$
\psi^*i(X)\Omega_r = (j_\Lag \circ (\rho_1^\Lag)^{-1} \circ \psi_\Lag)^*i(X)\Omega_r
= ((\rho_1^\Lag)^{-1} \circ \psi_\Lag)^*i(Y)j_\Lag^*\Omega_r \, .
$$
Applying Lemma \ref{lemma:LagTechLemmaSect} we obtain
\begin{align*}
((\rho_1^\Lag)^{-1} \circ \psi_\Lag)^*i(Y)j_\Lag^*\Omega_r
&= ((\rho_1^\Lag)^{-1} \circ \psi_\Lag)^*i(Y)(\rho_1^\Lag)^*\Omega_\Lag \\
&= (\rho_1^\Lag \circ (\rho_1^\Lag)^{-1} \circ \psi_\Lag)^*i(Z)\Omega_\Lag = \psi_\Lag^*i(Z)\Omega_\Lag \, ,
\end{align*}
where $Z \in \vf(J^3\pi)$ is the unique vector field related with $Y$ by the diffeomorphism $\rho_1^\Lag$.
Hence, as $\psi_\Lag^*i(Z)\Omega_\Lag = 0$, for every $Z \in \vf(J^3\pi)$ by hypothesis, we have
proved that the section
$\psi = j_\Lag \circ (\rho_1^\Lag)^{-1} \circ \psi_\Lag \in \Gamma(\rho_M^r)$ satisfies the equation
$$
\psi^*i(X)\Omega_r = 0 \, , \quad \mbox{for every } X \in \vf(\W_r) \mbox{ tangent to }\W_\Lag \, .
$$
However, from Proposition \ref{prop:UnifDynEqSectTangent} we know that if $\psi \in \Gamma(\rho_M^r)$
is a holonomic section, then the last equation is equivalent to the equation
\eqref{eqn:UnifDynEqSect}, that is,
\begin{equation*}
\psi^*i(X)\Omega_r = 0 \, , \quad \mbox{for every } X \in \vf(\W_r) \, . \qedhere
\end{equation*}
\end{proof}

Let us compute the local equation for the section
$\psi_\Lag = \rho_1^r \circ \psi \in \Gamma(\bar{\pi}^3)$.
Assume that the section $\psi \in \Gamma(\rho_M^r)$ is given locally by
$\psi(x^i) = (x^i,u^\alpha,u^\alpha_i,u^\alpha_{I},u^\alpha_{J},p_\alpha^i,p_\alpha^{I})$.
Since $\psi$ is a holonomic section solution to equation \eqref{eqn:UnifDynEqSect},
it must satisfy the local equations \eqref{eqn:UnifDynEqSectLocal},
\eqref{eqn:UnifDynEqSectRelationMomenta} and \eqref{eqn:UnifDynEqSectHolonomy}.
The equations \eqref{eqn:UnifDynEqSectHolonomy} are automatically satisfied as a consequence of the
assumption of $\psi$ being holonomic. Now, taking into account that $\psi$ takes values in the
submanifold $\W_\Lag \cong \textnormal{graph}(\Leg)$, the equations \eqref{eqn:UnifDynEqSectLocal}
and \eqref{eqn:UnifDynEqSectRelationMomenta} can be $\rho_1^r$-projected to $J^3\pi$, thus giving
the following system of $n$ partial differential equations for the component functions of the
section $\psi_\Lag = \rho_1^r \circ \psi$
$$
\restric{\derpar{L}{u^\alpha}}{\psi_\Lag} - \restric{\frac{d}{dx^i} \, \derpar{L}{u_i^\alpha}}{\psi_\Lag} + \restric{\sum_{|I|=2} \frac{d^{|I|}}{dx^{I}} \, \derpar{L}{u_I^\alpha}}{\psi_\Lag} = 0 \, , \quad 1 \leqslant \alpha \leqslant n \, ,
$$
where the section $\psi_\Lag$ is locally given by
$\psi_\Lag(x^i) = (x^i,u^\alpha,u^\alpha_i,u^\alpha_{I},u^\alpha_{J})$. Finally, since $\psi_\Lag$
is holonomic in $J^3\pi$, there exists a section $\phi \in \Gamma(\pi)$ with coordinate expression
$\phi(x^i) = (x^i,u^\alpha(x^i))$ satisfying $j^3\phi = \psi_\Lag$. Then, the above equations can
be rewritten as follows
\begin{equation} \label{eqn:EulerLagrangeLagrangian}
\restric{\derpar{L}{u^\alpha}}{j^3\phi} - \restric{\frac{d}{dx^i} \, \derpar{L}{u_i^\alpha}}{j^3\phi} + \restric{\sum_{|I|=2} \frac{d^{|I|}}{dx^{I}} \, \derpar{L}{u_I^\alpha}}{j^3\phi} = 0 \, \quad 1 \leqslant \alpha \leqslant n \, ,
\end{equation}
Therefore, we obtain the Euler-Lagrange equations for a second-order field theory.

\subsection{Field equations for multivector fields}

\begin{lemma}\label{lemma:LagRelatedMultiVF}
Let $\X \in \vf^{m}(\W_r)$ be a multivector field tangent to $\W_\Lag \hookrightarrow \W_r$. Then
there exists a unique multivector field $\X_\Lag \in \vf^m(J^3\pi)$ such that
$\X_\Lag \circ \rho_1^r \circ j_\Lag = \Lambda^mT\rho_1^r \circ \X \circ j_\Lag$.

 Conversely, if $\X_\Lag \in \vf^m(J^3\pi)$, then there exists a unique multivector
field $\X \in \vf(\W_r)$ tangent to $\W_\Lag$ such that
$\X_\Lag \circ \rho_1^r \circ j_\Lag = \Lambda^mT\rho_1^r \circ \X \circ j_\Lag$.
\end{lemma}
\begin{proof}
Since the multivector field $\X$ is tangent to $\W_\Lag$, there exists a unique multivector field
$\X_o \in \vf^{m}(\W_\Lag)$ which is $j_\Lag$-related to $\X$, that is,
$\Lambda^mT j_\Lag \circ \X_o = \X \circ j_\Lag$. Furthermore, since
$\rho_1^\Lag \colon \W_\Lag \to J^3\pi$ is a diffeomorphism, there is a unique multivector field
$\X_\Lag \in \vf^{m}(J^3\pi)$ which is $\rho_1^\Lag$-related to $\X_o$; that is,
$\X_\Lag \circ \rho_1^\Lag = \Lambda^mT j_1^\Lag \circ \X_o$. Then, computing we have
\begin{align*}
\X_\Lag \circ \rho_1^r \circ j_\Lag &= \X_\Lag \circ \rho_1^\Lag
= \Lambda^mT\rho_1^\Lag \circ \X_o \\
&= \Lambda^mT\rho_1^r \circ \Lambda^mT j_\Lag \circ \X_o
= \Lambda^mT\rho_1^r \circ \X \circ j_\Lag \, .
\end{align*}
The converse is proved reversing this reasoning.
\end{proof}

The above result states that there is a $1$-to-$1$ correspondence between the set of multivector
fields $\X \in \vf^m(\W_r)$ tangent to $\W_\Lag$ and the set of multivector fields
$\X_\Lag \in \vf^{m}(J^3\pi)$, which makes the following diagram commutative
$$
\xymatrix{
\ & \ & \ & \Lambda^mT\W_r \ar[dlll]_{\Lambda^mT\rho_1^r} \\
\Lambda^mT J^3\pi & \ & \ & \Lambda^mT\W_\Lag \ar[lll]^{\Lambda^mT\rho_1^\Lag} \ar@{^{(}->}[u]^{\Lambda^mT j_\Lag} \\
 \ & \ & \ & \ \\
\ & \ & \ & \W_r \ar[dlll]_{\rho_1^r} \ar@/_2.25pc/[uuu]_{\X} \\
J^3\pi \ar[uuu]^{\X_\Lag} & \ & \ & \W_\Lag \ar[lll]^{\rho_1^\Lag} \ar@{_{(}->}[u]_{j_\Lag} \ar@/^1.5pc/[uuu]^{\X_o}|(.32)\hole
}
$$

As a consequence, we obtain the following result:

\begin{theorem}\label{thm:UnifToLagMultiVF}
Let $\X \in \vf^m(\W_r)$ be a locally decomposable holonomic multivector field solution to the
equation \eqref{eqn:UnifDynEqMultiVF} (at least on the points of a submanifold
$\W_f \hookrightarrow \W_\Lag$) and tangent to $\W_\Lag$ (resp. tangent to $\W_f$). Then there
exists a unique locally decomposable holonomic multivector field $\X_\Lag \in \vf^m(J^3\pi)$
solution to the equation
\begin{equation}\label{eqn:LagDynEqMultiVF}
i(\X_\Lag)\Omega_\Lag = 0 \, ,
\end{equation}
(at least on the points of $S_f = \rho_1^\Lag(\W_f)$, and tangent to $S_f$).

Conversely, if $\X_\Lag \in \vf^{m}(J^3\pi)$ is a locally decomposable holonomic multivector field
solution to the equation \eqref{eqn:LagDynEqMultiVF} (at least on the points of a submanifold
$S_f \hookrightarrow J^3\pi$, and tangent to $S_f$), then there exists a unique locally decomposable
holonomic multivector field $\X \in \vf^{m}(\W_r)$ which is a solution to the equation
\eqref{eqn:UnifDynEqMultiVF} (at least on the points of $(\rho_1^\Lag)^{-1}(S_f) \hookrightarrow \W_\Lag$),
and tangent to $\W_\Lag$ (resp. tangent to $\W_f$).
\end{theorem}
\begin{proof}
Applying Lemmas \ref{lemma:LagRelatedForms} and \ref{lemma:LagRelatedMultiVF}, we have
\begin{align*}
\restric{i(\X)\Omega_r}{\W_\Lag}
&= \restric{i(\X)(\rho_1^r)^*\Omega_\Lag}{\W_\Lag}
= \restric{(\rho_1^r)^*i(\X_\Lag)\Omega_\Lag}{\W_\Lag} \\
&= \restric{i(\X_\Lag)\Omega_\Lag}{\rho_1^r(\W_\Lag)}
= \restric{i(\X_\Lag)\Omega_\Lag}{J^3\pi} \, .
\end{align*}
Hence, $\X_\Lag$ is a solution to the equation $i(\X_\Lag)\Omega_\Lag = 0$ if, and only if,
$\X$ is a solution to the equation $i(\X)\Omega_r = 0$.

Now we must prove that $\X_\Lag$ is holonomic if, and only if, $\X$ is holonomic.
Observe that, following the same reasoning as above, we have
\begin{align*}
\restric{i(\X)(\rho_M^r)^*\eta}{\W_\Lag} &= \restric{i(\X)(\bar{\pi}^3 \circ \rho_1^r)^*\eta}{\W_\Lag}
= \restric{(\rho_1^r)^*i(\X_\Lag)(\bar{\pi}^3)^*\eta}{\W_\Lag} \\
&= \restric{i(\X_\Lag)(\bar{\pi}^3)^*\eta}{\rho_1^r(\W_\Lag)}
= \restric{i(\X_\Lag)(\bar{\pi}^3)^*\eta}{J^3\pi} \, .
\end{align*}
Hence, $\X_\Lag$ is $\bar{\pi}^3$-transverse if, and only if, $\X$ is $\rho_M^r$-transverse.

Now, let us assume that $\X \in \vf^m(\W_r)$ is holonomic, and let $\psi \in \Gamma(\rho_M^r)$ be
an integral section of $\X$. Then, the section $\psi_\Lag = \rho_1^r \circ \psi \in \Gamma(\bar{\pi}^3)$
is holonomic by definition, and we have
$$
\X_\Lag \circ \psi_\Lag = \X_\Lag \circ \rho_1^r \circ \psi
= \Lambda^mT\rho_1^r \circ \X \circ \psi
= \Lambda^mT\rho_1^r \circ \psi^\prime
= \psi^\prime_\Lag \, ,
$$
where $\psi^\prime \colon M \to \Lambda^mT\W_r$ is the canonical lifting of $\psi$ to
$\Lambda^mT\W_r$. That is, $\psi_\Lag$ is an integral section of $\X_\Lag$. Hence, if $\X$ is
holonomic, then $\X_\Lag$ is holonomic.

For the converse, let us assume that $\X_\Lag \in \vf^m(J^3\pi)$ is holonomic, and let
$\psi_\Lag \in \Gamma(\bar{\pi}^3)$ be an integral section of $\X_\Lag$. Then, the section
$\psi = j_\Lag \circ (\rho_1^\Lag)^{-1} \circ \psi_\Lag \in \Gamma(\rho_M^3)$ satisfies
$$
\rho_1^r \circ \psi = \rho_1^r \circ j_\Lag \circ (\rho_1^\Lag)^{-1} \circ \psi_\Lag = \psi_\Lag \, ,
$$
since
$\rho_1^r \circ j_\Lag = \rho_1^\Lag
\Leftrightarrow \rho_1^r \circ j_\Lag \circ (\rho_1^\Lag)^{-1} = \textnormal{Id}_{J^3\pi}$.
Therefore, the section $\psi$ is holonomic. Finally,
since the multivector field $\X$ is tangent to $\W_\Lag$, there exists a unique multivector
field $\X_o \in \vf^{m}(\W_\Lag)$ satisfying $\Lambda^mT j_\Lag \circ \X_o = \X \circ j_\Lag$.
In addition, since the map $\rho_1^\Lag$ is a diffeomorphism, $\X_\Lag$ and $\X_o$ are
$(\rho_1^\Lag)^{-1}$-related; that is,
$\X_o \circ (\rho_1^\Lag)^{-1} = (\Lambda^mT\rho_1^\Lag)^{-1} \circ \X_\Lag$. Then we have
\begin{align*}
\X \circ \psi &= \X \circ j_\Lag \circ (\rho_1^\Lag)^{-1} \circ \psi_\Lag
= \Lambda^mT j_\Lag \circ \X_o \circ (\rho_1^\Lag)^{-1} \circ \psi_\Lag \\
&= \Lambda^mT j_\Lag \circ (\Lambda^mT\rho_1^\Lag)^{-1} \circ \X_\Lag \circ \psi_\Lag
= \Lambda^mT j_\Lag \circ (\Lambda^mT\rho_1^\Lag)^{-1} \circ \psi^\prime_\Lag \\
&= (j_\Lag \circ (\rho_1^\Lag)^{-1} \circ \psi_\Lag)^\prime
= \psi^\prime \, .
\end{align*}
Hence, $\psi$ is an integral section of $\X$. Therefore, $\X$ is holonomic if, and only if, $\X_\Lag$
is holonomic.
\end{proof}

Let $\X_\Lag \in \vf^m(J^3\pi)$ be a locally decomposable multivector field. From the results in
\cite{art:Echeverria_Munoz_Roman98} we know that $\X_\Lag$ admits the following local expression
\begin{equation}\label{eqn:LagGenericMultiVFLocal}
\X = f \bigwedge_{j=1}^{m}
\left(  \derpar{}{x^j} + F_j^\alpha\derpar{}{u^\alpha} + F_{i,j}^\alpha\derpar{}{u_i^\alpha} + F_{I,j}^\alpha\derpar{}{u_{I}^\alpha}
+ F_{J,j}^\alpha\derpar{}{u_{J}^\alpha} \right) \, .
\end{equation}
Taking $f = 1$ as a representative of the equivalence class, since $\X_\Lag$ is required to be
holonomic, it must satisfy the equations \eqref{eqn:MultiVFHolonomyLocal} with $k = 3$ and $r = 1$, that is,
$$
F_j^\alpha = u_j^\alpha \quad ; \quad
F_{i,j}^\alpha = u_{1_i+1_j} \quad ; \quad
F_{I,j}^\alpha = u_{I + 1_j} \, .
$$
In addition, $\X_\Lag$ is a solution to the equation \eqref{eqn:LagDynEqMultiVF}. Bearing in mind the
local equations for the multivector field $\X$, we obtain that the local equations for the component
functions of $\X_\Lag$ are
\begin{align*}
\derpar{\hat{L}}{u^\alpha}
- \frac{d}{dx^i} \, \derpar{\hat{L}}{u_i^\alpha}
&+ \sum_{|I|=2} \frac{d^{|I|}}{dx^{I}} \, \derpar{\hat{L}}{u_{I}^\alpha} \\
&\quad{} + \sum_{i=1}^{m}\sum_{j=1}^{m}\frac{1}{n(ij)}\left( F_{I+1_j,i}^\beta - \frac{d}{dx^i} \, u_{I+1_j}^\beta \right) \derpars{\hat{L}}{u_I^\beta}{u_{1_i+1_j}^\alpha} = 0 \, .
\end{align*}

\begin{theorem}
The following assertions on a section $\phi \in \Gamma(\pi)$ are equivalent:
\begin{enumerate}
\item $j^{3}\phi$ is a solution to equation \eqref{eqn:LagDynEqSect}, that is,
$$
(j^3\phi)^*i(X)\Omega_\Lag = 0 \, , \quad \mbox{for every } X \in \vf(J^3\pi) \, .
$$
\item In natural coordinates, if $\phi$ is given by $\phi(x^i) = (x^i,u^\alpha)$, then
its $3$rd prolongation $j^{3}\phi(x^i) = (x^i,u^\alpha,u^\alpha_i,u^\alpha_I,u^\alpha_J)$
is a solution to the Euler-Lagrange equations given by \eqref{eqn:EulerLagrangeLagrangian},
that is,
$$
\restric{\derpar{L}{u^\alpha}}{j^3\phi} - \restric{\frac{d}{dx^i} \, \derpar{L}{u^\alpha_i}}{j^3\phi}
+ \restric{\sum_{|I|=2}\frac{d^{|I|}}{dx^I} \, \derpar{L}{u^\alpha_I}}{j^3\phi} = 0 \, .
$$
\item $\psi_\Lag = j^3\phi$ is a solution to the equation
\begin{equation*}
i(\Lambda^m\psi_\Lag^\prime)(\Omega_\Lag \circ \psi_\Lag) = 0 \, ,
\end{equation*}
where $\Lambda^m\psi_\Lag^\prime \colon M \to \Lambda^mT(J^3\pi)$ is the canonical lifting of
$\psi_\Lag$.
\item $j^3\phi$ is an integral section of a multivector field belonging to a class of locally
decomposable holonomic multivector fields $\{ \X_\Lag \} \subset \vf^{m}(J^3\pi)$ satisfying
equation \eqref{eqn:LagDynEqMultiVF}, that is,
$$
i(\X_\Lag)\Omega_\Lag = 0 \, .
$$
\end{enumerate}
\end{theorem}

\section{Hamiltonian formalism}
\label{sec:HamiltonianForm}

\subsection{General setting}

In order to describe the Hamiltonian formalism for second-order field theories using the results
obtained in Section \ref{sec:LagHamFormalism}, we must distinguish between the regular and non-regular
cases.

Let $\widetilde{\Leg} \colon J^3\pi \to J^2\pi^\dagger$ be the extended Legendre map obtained in
\eqref{eqn:ExtendedLegendreMapLocal} and $\Leg \colon J^3\pi \to J^2\pi^\ddagger$ the restricted
Legendre map obtained in \eqref{eqn:RestrictedLegendreMapLocal}. Let us denote
$\widetilde{\P} = \textnormal{Im}(\widetilde{\Leg}) = \widetilde{\Leg}(J^3\pi) \stackrel{\tilde{\jmath}}{\hookrightarrow} J^2\pi^\dagger$
and $\P = \textnormal{Im}(\Leg) = \Leg(J^3\pi) \stackrel{\jmath}{\hookrightarrow} J^2\pi^\ddagger$
the image of the extended and restricted Legendre maps, respectively, which we assume to be submanifolds.
We denote $\bar{\pi}_\P \colon \P \to M$ the natural projection, and $\Leg_o$ the map defined by
$\Leg = \jmath \circ \Leg_o$.

\noindent\textbf{Remark:}
In the hyperregular case, we have $\P = J^2\pi^\ddagger$ and $\Leg_o = \Leg$.

With the previous notations, we can give the following definition:

\begin{definition}
A Lagrangian density $\Lag \in \Omega^{m}(J^2\pi)$ is said to be \emph{almost-regular} if
\begin{enumerate}
\item $\P$ is a closed submanifold of $J^2\pi^\ddagger$.
\item $\Leg$ is a submersion onto its image.
\item For every $j^3_x\phi \in J^3\pi$, the fibers $\Leg^{-1}(\Leg(j^3_x\phi))$
are connected submanifolds
of $J^3\pi$.
\end{enumerate}
\end{definition}

If the Lagrangian density is almost-regular, the Legendre map is a submersion onto its image, and
therefore it admits local sections defined on the submanifold $\P \hookrightarrow J^2\pi^\ddagger$.
We denote by $\Gamma_\P(\Leg)$ the set of local sections of $\Leg$ defined on the submanifold $\P$.
Observe that if $\Lag$ is regular, then $\Gamma_\P(\Leg)$ is exactly the set of
local sections of $\Leg$.

As a consequence of Proposition \ref{prop:RankLegendreMaps}, we have that $\widetilde{\P}$ is diffeomorphic
to $\P$. This diffeomorphism is $\widetilde{\mu} = \mu \circ \tilde{\jmath} \colon \widetilde{\P} \to \P$.
This enables us to state:

\begin{lemma}
If the Lagrangian density $\Lag \in \Omega^{m}(J^2\pi)$ is, at least, almost-regular, then the Hamiltonian section
$\hat{h} \in \Gamma(\mu_\W)$ induces a Hamiltonian section $h \in \Gamma(\widetilde{\mu})$ defined by
$$
h([\omega]) = (\rho_2 \circ \hat{h}) ([(\rho_2^r)^{-1}(\jmath([\omega]))]) \, , \quad \mbox{for every } [\omega] \in \P \, .
$$
\end{lemma}
\begin{proof}
It is clear that, given $[\omega] \in J^{2}\pi^\ddagger$, the section $\hat{h}$ maps every point
$(j^3_x\phi,[\omega]) \in (\rho^r_2)^{-1}([\omega])$ into $\rho_2^{-1}[\rho_2(\hat{h}(j^3_x\phi,[\omega]))]$.
So we have the diagram
$$
\xymatrix{
\widetilde{\P} \ar[rr]^-{\tilde{\jmath}} \ar[d]^-{\tilde{\mu}} & \ & J^{2}\pi^\dagger \ar[d]^-{\mu} & \ & \W \ar[d]_-{\mu_\W} \ar[ll]_-{\rho_2} \\
\P \ar[rr]^-{\jmath} \ar[urr]^-{h}& \ & J^{2}\pi^\ddagger & \ & \W_r \ar@/_0.7pc/[u]_{\hat{h}} \ar[ll]_-{\rho_2^r}
}
$$
Thus, the crucial point is the $\rho_2$-projectability of the local function $\hat{H}$. However,
since a local base for $\ker T\rho_2$ is given by
$$
\ker T\rho_2 = \left\langle \derpar{}{u_I^\alpha},\derpar{}{u_J^\alpha} \right\rangle \, ,
$$
with $|I| = 2$ and $|J| = 3$, then we have that $\hat{H}$ is $\rho_2$-projectable if, and only if,
$$
p_{\alpha}^{I} = \derpar{L}{u_I^\alpha} \, .
$$
This condition is fulfilled if $[\omega] \in \P = \textnormal{Im}(\Leg)$, which
implies that $\rho_2[\hat{h}((\rho_2^r)^{-1}([\omega]))] \in \widetilde{\P}$.
\end{proof}

As in the unified setting, this Hamiltonian $\mu$-section is specified by a local Hamiltonian function
$H \in C^\infty(\P)$, that is,
$$
h(x^i,u^\alpha,u_i^\alpha,p_\alpha^i,p_\alpha^I) = (x^i,u^\alpha,u_i^\alpha,-H,p_\alpha^i,p_\alpha^I) \, .
$$
Using the Hamiltonian $\mu$-section we define the Hamilton-Cartan forms
$\Theta_h = h^*\Theta_1^s \in \Omega^{m}(\P)$ and $\Omega_h = h^*\Omega_1^s \in \Omega^{m+1}(\P)$.
Observe that $\Leg_o^*\Theta_h = \Theta_\Lag$ and $\Leg_o^*\Omega_h = \Omega_\Lag$.

\noindent\textbf{Remark:}
The Hamiltonian $\mu$-section can be defined in some equivalent ways without passing through the
unified formalism. First, we can define it as $h = \tilde{\jmath} \circ \widetilde{\mu}^{-1}$.
From this, bearing in mind the definition of $\widetilde{\P}$ and $\P$ as the image sets of the
extended and restricted Legendre maps, respectively, we can also define the Hamiltonian $\mu$-section
as $h = \widetilde{\Leg} \circ \sigma$, where $\sigma \in \Gamma_\P(\Leg)$.

\subsection{Hyperregular and regular Lagrangian densities}

For the sake of simplicity, we assume throughout this Section that the Lagrangian density
$\Lag \in \Omega^{m}(J^2\pi)$ is hyperregular, and that $\Upsilon \colon J^2\pi^\ddagger \to J^3\pi$
is a global section of $\Leg$. All the results stated also hold for regular Lagrangians,
restricting to the corresponding open sets where the Legendre map admits local sections.

First, observe that if the Lagrangian density is hyperregular, then the local Hamiltonian function
associated to the Hamiltonian $\mu$-section $h$ has the following coordinate expression
\begin{equation}\label{eqn:HamiltonianFunctionLocalRegularLag}
H(x^i,u^\alpha,u_i^\alpha,p_\alpha^i,p_\alpha^I) = p_\alpha^iu_i^\alpha + p_\alpha^If_I^\alpha
- (\pi_2^3 \circ \Upsilon)^*L \, ,
\end{equation}
where $f_I^\alpha(x^i,u^\alpha,u^\alpha_i,p_\alpha^i,p_\alpha^I) = \Upsilon^*u_I^\alpha$.
Therefore, the Hamilton-Cartan forms have the following coordinate expression
\begin{align*}
\Theta_h &= - H d^mx + p_\alpha^id u^\alpha \wedge d^{m-1}x_i + \frac{1}{n(ij)} \, p_\alpha^{1_i+1_j}d u_i^\alpha \wedge d^{m-1}x_j \, , \\
\Omega_h &= d H \wedge d^mx - d p_\alpha^i \wedge d u^\alpha \wedge d^{m-1}x_i
- \frac{1}{n(ij)} \, d p_\alpha^{1_i+1_j} \wedge d u_i^\alpha \wedge d^{m-1}x_j \, .
\end{align*}
In addition, since $\textnormal{Im}(\Leg) = J^2\pi^\ddagger$, then the Hamiltonian sections $h$ and
$\hat{h}$ satisfy $h \circ \rho_2^r = \rho_2 \circ \hat{h}$, that is, the following diagram commutes
$$
\xymatrix{
\W \ar[drr]^{\rho_2} & \ & \ \\
\W_r \ar[u]^{\hat{h}} \ar[drr]_{\rho_2^r} & \ & J^2\pi^\dagger \\
\ & \ & J^2\pi^\ddagger \ar[u]_{h}
}
$$

\begin{proposition}
If the Lagrangian density is hyperregular, then the Hamilton-Cartan $(m+1)$-form
$\Omega_h = h^*\Omega_1^s \in \Omega^{m+1}(J^2\pi^\ddagger)$ is a multisymplectic form
in $J^2\pi^\ddagger$.
\end{proposition}
\begin{proof}
A direct computation in coordinates leads to this result. Let $\Upsilon \in \Gamma(\Leg)$ be a
global section of the restricted Legendre map, and assume that the local Hamiltonian function $H$
is given locally by \eqref{eqn:HamiltonianFunctionLocalRegularLag}. Then we have the following
coordinate expression for $d H$
\begin{align*}
d H &=
- \derpar{L}{x^i}\,d x^i - \derpar{L}{u^\alpha}\,d u^\alpha
+ \left( p_\alpha^i - \derpar{L}{u_i^\alpha} \right)d u_i^\alpha
+ \left( p_\alpha^I - \derpar{L}{u_I^\alpha} \right)d f_I^\alpha \\
&\quad{} + u_i^\alpha d p_\alpha^i + f_I^\alpha d p_\alpha^I \, ,
\end{align*}
where
$$
d f_I^\alpha = \derpar{f_I^\alpha}{x^j}\,d x^j + \derpar{f_I^\alpha}{u^\beta}\,d u^\beta
+ \derpar{f_I^\alpha}{u^\beta_j}\,d u^\beta_j + \derpar{f_I^\alpha}{p_\beta^j}\,d p^j_\beta
+ \derpar{f_I^\alpha}{p^K_\beta}\,d p^K_\beta \, .
$$
Observe that since $H$ takes values in $J^2\pi^\ddagger = \textnormal{Im}(\Leg)$, we have
$p_\alpha^I - \partial L / \partial u_I^\alpha = 0$. Thus, the expression of $d H$ reads
$$
d H = - \derpar{L}{x^i}\,d x^i - \derpar{L}{u^\alpha}\,d u^\alpha
+ \left( p_\alpha^i - \derpar{L}{u_i^\alpha} \right)d u_i^\alpha
+ u_i^\alpha d p_\alpha^i + f_I^\alpha d p_\alpha^I \, ,
$$
and therefore the Hamilton-Cartan $(m+1)$-form is locally given by
\begin{align*}
\Omega_h &= \left[ - \derpar{L}{u^\alpha}\,d u^\alpha
+ \left( p_\alpha^i - \derpar{L}{u_i^\alpha} \right)d u_i^\alpha
+ u_i^\alpha d p_\alpha^i + f_I^\alpha d p_\alpha^I \right] \wedge d^mx \\
&\qquad {}- d p_\alpha^i \wedge d u^\alpha \wedge d^{m-1}x_i - \frac{1}{n(ij)} \, d p_\alpha^{1_i+1_j} \wedge d u_i^\alpha \wedge d^{m-1}x_j \, .
\end{align*}
Now, since the $C^\infty(J^2\pi^\ddagger)$-module of vector fields $\vf(J^2\pi^\ddagger)$ is locally given by
$$
\vf(J^2\pi^\ddagger) = \left\langle \derpar{}{x^i} \, , \,
\derpar{}{u^\alpha} \, , \, \derpar{}{u^\alpha_i} \, , \,
\derpar{}{p_\alpha^i} \, , \, \derpar{}{p_\alpha^I} \right\rangle \, ,
$$
we have
\begin{align*}
i\left(\derpar{}{x^k}\right)\Omega_h &=
- d H \wedge d^{m-1}x_k - d p_\alpha^i \wedge d u^\alpha \wedge d^{m-2}x_{ik} \\
&\quad{} - \frac{1}{n(ij)}d p_\alpha^{1_i+1_j} \wedge d u_i^\alpha \wedge d^{m-2}x_{jk} \, , \\
i\left(\derpar{}{u^\alpha}\right)\Omega_h &=
- \derpar{L}{u^\alpha} d^{m}x + d p_\alpha^i  \wedge d^{m-1}x_{i} \, , \\
i\left(\derpar{}{u^\alpha_i}\right)\Omega_h &=
\left( p_\alpha^i - \derpar{L}{u_i^\alpha} \right) d^{m}x + \frac{1}{n(ij)}d p_\alpha^{1_i+1_j} \wedge d^{m-1}x_{j} \, , \\
i\left(\derpar{}{p_\alpha^i}\right)\Omega_h &=
u_i^\alpha d^{m}x - d u^\alpha \wedge d^{m-1}x_{i}\, , \\
i\left(\derpar{}{p_\alpha^I}\right)\Omega_h &=
f_I^\alpha d^{m}x - \sum_{1_i+1_j=I} \frac{1}{n(ij)} d u_i^\alpha \wedge d^{m-1}x_{j} \, .
\end{align*}
From this it is clear that $i(X)\Omega_h = 0$ if, and only if, $X = 0$, that is, $\Omega_h$
is multisymplectic.
\end{proof}

Now we recover the field equations from the unified setting using the natural projection
$\rho_2^r \colon \W_r \to J^2\pi^\ddagger$. First, the sections solution in the Hamiltonian
formalism are recovered using the following result:

\begin{proposition}\label{prop:UnifToHamSectReg}
Let $\psi \in \Gamma(\rho_M^r)$ be a holonomic section solution to the equation
\eqref{eqn:UnifDynEqSect}. Then the section
$\psi_h = \rho_2^r \circ \psi \in \Gamma(\bar{\pi}_{J^1\pi}^\ddagger)$ is a solution to the equation
\begin{equation}\label{eqn:HamDynEqSectReg}
\psi_h^*i(X)\Omega_h = 0 \, , \quad \mbox{for every } X \in \vf(J^2\pi^\ddagger) \, .
\end{equation}
\end{proposition}
\begin{proof}
Since $\rho_2^r \colon \W_r \to J^3\pi$ is a submersion, for every vector field
$X \in \vf(J^2\pi^\ddagger)$ there exist some vector fields $Y \in \vf(\W_r)$ such that $X$ and $Y$
are $\rho_2^r$-related. Observe that this vector field $Y$ is not unique, the vector field $Y + Y_o$,
with $Y_o \in \ker T\rho_2^r$ is also $\rho_2^r$-related with $X$. Thus, using this particular
choice of $\rho_2^r$-related vector fields, we have
\begin{align*}
\psi_h^*i(X)\Omega_h &= (\rho_2^r \circ \psi)^*i(X)\Omega_h
= \psi^*((\rho_2^r)^*i(X)\Omega_h) = \psi^*i(Y)(\rho_2^r)^*\Omega_h \\
&= \psi^*i(Y)(h \circ \rho_2^r)^*\Omega_1^s = \psi^*i(Y)(\rho_2 \circ \hat{h})^*\Omega_1^s
= \psi^*i(Y)\Omega_r \, .
\end{align*}
Since the equality $\psi^*i(Y)\Omega_r = 0$ holds for every $Y \in \vf(\W_r)$, in particular
it holds for every $Y \in \vf(\W_r)$ which is $\rho_2^r$-related with $X \in \vf(J^2\pi^\ddagger)$.
Hence we obtain
\begin{equation*}
\psi_h^*i(X)\Omega_h = \psi^*i(Y)\Omega_r = 0 \, . \qedhere
\end{equation*}
\end{proof}

The diagram illustrating the situation of the above Proposition is the following:
$$
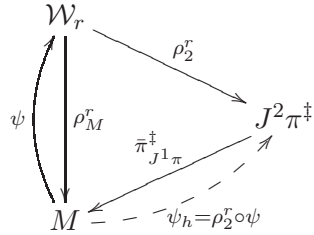
\xymatrix{
\W_r \ar[drr]^{\rho^r_2} \ar[dd]^{\rho^r_M} & \ & \\
\ & \ & J^2\pi^\ddagger \ar[dll]_{\bar{\pi}_{J^1\pi}^\ddagger} \\
M \ar@/_1pc/@{-->}[urr]_{\psi_h = \rho_2^r \circ \psi} \ar@/^1pc/[uu]^{\psi} & \ & \
}
$$

Let us compute the local equations for the section
$\psi_h = \rho_2^r \circ \psi \in \Gamma(\bar{\pi}_{J^1\pi}^\ddagger)$. If the section
$\psi \in \Gamma(\rho_M^r)$ is locally given by
$\psi(x^i) = (x^i,u^\alpha,u^\alpha_i,u^\alpha_I,u^\alpha_J,p_\alpha^i,p_\alpha^I)$,
then the section $\psi_h = \rho_2^r \circ \psi$ is given in coordinates by
$\psi_h(x^i) = (x^i,u^\alpha,u^\alpha_i,p_\alpha^i,p_\alpha^I)$. Now, bearing in mind that the
section $\psi$ solution to the equation \eqref{eqn:UnifDynEqSect} must satisfy the local equations
\eqref{eqn:UnifDynEqSectLocal}, \eqref{eqn:UnifDynEqSectRelationMomenta} and
\eqref{eqn:UnifDynEqSectHolonomy}, and that the section $\psi$ takes values in the submanifold
$\W_\Lag \cong \textnormal{graph}(\Leg)$ and the local expression
\eqref{eqn:HamiltonianFunctionLocalRegularLag} of the Hamiltonian function $H$ in the hyperregular
case, we obtain the following system of partial differential equations for the section $\psi_h$
\begin{equation}\label{eqn:HamiltonEquationsReg}
\begin{array}{c}
\displaystyle\derpar{u^\alpha}{x^i} = \derpar{H}{p_\alpha^i} \quad ; \quad
\sum_{1_i+1_j=I}\frac{1}{n(ij)}\,\derpar{u_i^\alpha}{x^j} = \derpar{H}{p_\alpha^I} \, , \\[15pt]
\displaystyle\sum_{i=1}^m\derpar{p_\alpha^i}{x^i} = -\derpar{H}{u^\alpha} \quad ; \quad
\sum_{j=1}^m\derpar{p_\alpha^{1_i+1_j}}{x^j} = -\derpar{H}{u_i^\alpha} \, .
\end{array}
\end{equation}

In order to recover the field equations for multivector fields, we first need the following technical
result, which is similar to Lemma \ref{lemma:LagRelatedMultiVF}.

\begin{lemma}\label{lemma:HamRelatedMultiVFReg}
Let $\X \in \vf^{m}(\W_r)$ be a multivector field tangent to $\W_\Lag \hookrightarrow \W_r$, and
let $\X_o \in \vf^{m}(\W_\Lag)$ be the unique multivector field which is $j_\Lag$-related to $\X$.
If $\X_o$ is $\rho_2^\Lag$-projectable, then there exists a unique multivector field
$\X_h \in \vf^m(J^2\pi^\ddagger)$ such that
$\X_h \circ \rho_2^r \circ j_\Lag = \Lambda^mT\rho_2^r \circ \X \circ j_\Lag$.

\noindent Conversely, if $\X_h \in \vf^m(J^2\pi^\ddagger)$, then there exist multivector fields
$\X \in \vf(\W_r)$ tangent to $\W_\Lag$ such that
$\X_h \circ \rho_2^r \circ j_\Lag = \Lambda^mT\rho_2^r \circ \X \circ j_\Lag$.
\end{lemma}
\begin{proof}
The proof of this result is analogous to the proof of Lemma \ref{lemma:LagRelatedMultiVF}, bearing
in mind that $\rho_2^\Lag = \rho_2^r \circ j_\Lag \colon \W_\Lag \to J^2\pi^\ddagger$ is a
submersion onto $J^2\pi^\ddagger$. In particular, since the multivector fields $\X \in \vf^m(\W_r)$
and $\X_o \in \vf^{m}(\W_\Lag)$ are $j_\Lag$-related, the relation $\Lambda^mT j_\Lag \circ \X_o
= \X \circ j_\Lag$ is satisfied. On the other hand, as $\X_o$ is $\rho_2^\Lag$-projectable and
$\rho_2^\Lag \colon \W_\Lag \to J^2\pi^\ddagger$ is a submersion, there is a unique multivector field
$\X_h \in \vf^{m}(J^2\pi^\ddagger)$ which is $\rho_2^\Lag$-related to $\X_o$; that
is, $\X_h \circ \rho_2^\Lag = \Lambda^mT \rho_1^\Lag \circ \X_o$. Then we have
\begin{align*}
\X_h \circ \rho_2^r \circ j_\Lag &= \X_h \circ \rho_2^\Lag
= \Lambda^mT\rho_2^\Lag \circ \X_o \\
&= \Lambda^mT\rho_2^r \circ \Lambda^mT j_\Lag \circ \X_o
= \Lambda^mT\rho_2^r \circ \X \circ j_\Lag \, .
\end{align*}
The converse is proved reversing this reasoning, but now the multivector field
$\X_o \in \vf^m(\W_\Lag)$ which is $\rho_2^\Lag$-related
with the given $\X_h \in \vf^m(J^2\pi^\ddagger)$ is not unique, since
$\rho_2^\Lag$ is a submersion with $\ker T\rho_2^\Lag \neq \{ 0 \}$.
\end{proof}

As in the Lagrangian formalism, the previous result gives a correspondence between the set of
multivector fields $\X \in \vf^m(\W_r)$ tangent to $\W_\Lag$
and the set of multivector fields $\X_h \in \vf^{m}(J^2\pi^\ddagger)$ such that the
following diagram is commutative
$$
\xymatrix{
\Lambda^mT\W_r \ar[drrr]^{\Lambda^mT\rho_2^r} & \ & \ & \ \\
\Lambda^mT\W_\Lag \ar[rrr]_{\Lambda^mT\rho_2^\Lag} \ar@{_{(}->}[u]_{\Lambda^mT j_\Lag} & \ & \ & \Lambda^mT J^2\pi^\ddagger \\
\ & \ & \ & \ \\
\W_r \ar[drrr]^{\rho_2^r} \ar@/^2.25pc/[uuu]^{\X} & \ & \ & \ \\
\W_\Lag \ar[rrr]_{\rho_2^\Lag} \ar@{^{(}->}[u]^{j_\Lag} \ar@/_1.5pc/[uuu]_{\X_o}|(.32)\hole & \ & \ & J^2\pi^\ddagger \ar[uuu]^{\X_h}
}
$$
Nevertheless, observe that in the Hamiltonian formalism, the map
$\rho_2^\Lag = \rho_2^r \circ j_\Lag \colon \W_\Lag \to J^2\pi^\ddagger$
is a submersion (instead of a diffeomorphism, as in the Lagrangian setting), and thus the correspondence
is not $1$-to-$1$. In particular, for every multivector field $\X \in \vf^m(\W_r)$
tangent to $\W_\Lag$ we can define a unique multivector field
$\X_h \in \vf^m(J^2\pi^\ddagger)$ such that the previous diagram commutes.
But since $\rho_2^\Lag$ is a submersion, for every $\X_h \in \vf^m(J^2\pi^\ddagger)$ there are
several multivector fields $\X \in \vf^m(\W_r)$, tangent to $\W_\Lag$,
satisfying the same property.

\begin{theorem}\label{thm:UnifToHamMultiVFReg}
Let $\X \in \vf^m(\W_r)$ be a locally decomposable, $\rho_M^r$-transverse and integrable
multivector field solution to the equation \eqref{eqn:UnifDynEqMultiVF}, tangent to $\W_\Lag$ and
such that the unique multivector field in $\X_o \in \vf^{m}(\W_\Lag)$ which is
$j_\Lag$-related to $\X$ is $\rho_2^\Lag$-projectable.
Then there exists a locally decomposable, $(\bar{\pi}_{J^1\pi}^\ddagger)$-transverse and integrable
multivector field $\X_h \in \vf^m(J^2\pi^\ddagger)$ solution to the equation
\begin{equation}\label{eqn:HamDynEqMultiVFReg}
i(\X_h)\Omega_h = 0 \, ,
\end{equation}

Conversely, if $\X_h \in \vf^{m}(J^2\pi^\ddagger)$ is a locally decomposable,
$(\bar{\pi}_{J^1\pi}^\ddagger)$-transverse and integrable multivector field solution to the equation
\eqref{eqn:HamDynEqMultiVFReg}, then there exist locally decomposable, integrable and
$\rho_M^r$-transverse multivector fields $\X \in \vf^{m}(\W_r)$ tangent to $\W_\Lag$ solution to
the equation \eqref{eqn:UnifDynEqMultiVF}.
\end{theorem}
\begin{proof}
The proof of this result is analogous to the proof of Theorem \ref{thm:UnifToLagMultiVF}.
\end{proof}

Let $\X_h \in \vf^m(J^2\pi^\ddagger)$ be a locally decomposable multivector field given in the natural
coordinates of $J^2\pi^\ddagger$ by
\begin{equation}\label{eqn:HamGenericMultiVFLocalReg}
\X_h = f \bigwedge_{j=1}^{m}
\left(  \derpar{}{x^j} + F_j^\alpha\derpar{}{u^\alpha} + F_{i,j}^\alpha\derpar{}{u_i^\alpha}
+ G_{\alpha,j}^i\derpar{}{p_\alpha^i} + G_{\alpha,j}^{I}\derpar{}{p_\alpha^{I}} \right) \, ,
\end{equation}
Taking $f = 1$ as a representative of the equivalence class, since $\X_h$ is a solution
to the equation \eqref{eqn:HamDynEqMultiVFReg}, we obtain that the local equations for the component
functions of $\X_h$ are
$$
\begin{array}{c}
\displaystyle F_j^\alpha = \derpar{H}{p_\alpha^j} \quad ; \quad
\sum_{1_i+1_j=I}\frac{1}{n(ij)}\, F_{i,j}^\alpha = \derpar{H}{p_\alpha^I} \, , \\[15pt]
\displaystyle\sum_{i=1}^m G^{i}_{\alpha,i} = -\derpar{H}{u^\alpha} \quad ; \quad
\sum_{j=1}^m G^{1_i+1_j}_{\alpha,j} = -\derpar{H}{u_i^\alpha} \, .
\end{array}
$$

\begin{theorem}
The following assertions on a section $\psi_h \in \Gamma(\bar{\pi}_{J^1\pi}^\ddagger)$
are equivalent:
\begin{enumerate}
\item $\psi_h$ is a solution to equation \eqref{eqn:HamDynEqSectReg}, that is,
$$
\psi_h^*i(X)\Omega_h = 0 \, , \quad \mbox{for every } X \in \vf(J^2\pi^\ddagger) \, .
$$
\item In natural coordinates, if $\psi_h$ is given by
$\psi_h(x^i) = (x^i,u^\alpha,u_i^\alpha,p_\alpha^i,p_\alpha^I)$, then
its component functions are a solution to the equations \eqref{eqn:HamiltonEquationsReg}, that is,
$$
\begin{array}{c}
\displaystyle\derpar{u^\alpha}{x^i} = \derpar{H}{p_\alpha^i} \quad ; \quad
\sum_{1_i+1_j=I}\frac{1}{n(ij)}\,\derpar{u_i^\alpha}{x^j} = \derpar{H}{p_\alpha^I} \, , \\[15pt]
\displaystyle\sum_{i=1}^m\derpar{p_\alpha^i}{x^i} = -\derpar{H}{u^\alpha} \quad ; \quad
\sum_{j=1}^m\derpar{p_\alpha^{1_i+1_j}}{x^j} = -\derpar{H}{u_i^\alpha} \, .
\end{array}
$$
\item $\psi_h$ is a solution to the equation
\begin{equation*}
i(\Lambda^m\psi_h^\prime)(\Omega_h \circ \psi) = 0 \, ,
\end{equation*}
where $\Lambda^m\psi_h^\prime \colon M \to \Lambda^mT(J^2\pi^\ddagger)$ is the canonical
lifting of $\psi_h$.
\item $\psi_h$ is an integral section of a multivector field contained in a class of locally
decomposable, integrable and $(\bar{\pi}_{J^1\pi}^\ddagger)$-transverse multivector fields
$\{ \X_h \} \subset \vf^{m}(J^2\pi^\ddagger)$ satisfying equation \eqref{eqn:HamDynEqMultiVFReg}, that is,
$$
i(\X_h)\Omega_h = 0 \, .
$$
\end{enumerate}
\end{theorem}

\subsection{Singular (almost-regular) Lagrangian densities}

For singular (almost-regular) Lagrangian densities, only in the most favourable cases does there
exists a submanifold $\W_f \hookrightarrow \W_\Lag$ where the field equations can be solved. In this
situation, the solutions in the Hamiltonian formalism cannot be obtained directly from the projection
of the solutions in the unified setting, but rather by passing through the Lagrangian formalism and
using the Legendre map. Recall that, in this case, the phase space of the system is
$\P = \textnormal{Im}(\Leg) \hookrightarrow J^2\pi^\ddagger$.

\begin{proposition}
Let $\Lag \in \Omega^{m}(J^2\pi)$ be an almost-regular Lagrangian density. Let $\psi \in \Gamma(\rho_M^r)$
be a solution to the equation \eqref{eqn:UnifDynEqSect}. Then, the section
$\psi_h = \Leg_o \circ \rho_1^r \circ \psi = \Leg_o \circ \psi_\Lag \in \Gamma(\bar{\pi}_\P)$ is a solution
to the equation
\begin{equation}\label{eqn:HamDynEqSectSing}
\psi_h^*i(X)\Omega_h = 0 \, , \quad \mbox{for every } X \in \vf(\P) \, .
\end{equation}
\end{proposition}
\begin{proof}
Since the Lagrangian density $\Lag$ is assumed to be almost-regular, then the map $\Leg_o$ is a
submersion onto its image, $\P$. Thus, for every vector field $X \in \vf(\P)$ there exist some vector
fields $Y \in \vf(J^3\pi)$ such that $X$ and $Y$ are $\Leg_o$-related. Using this particular choice of
$\Leg_o$-related vector fields, we have
\begin{align*}
\psi_h^*i(X)\Omega_h &= (\Leg_o \circ \psi_\Lag)^*i(X)\Omega_h
=\psi_\Lag^*(\Leg_o^*i(X)\Omega_h) \\
&= \psi_\Lag^*i(Y)\Leg_o^*\Omega_h
= \psi_\Lag^*i(Y)\Omega_\Lag \, .
\end{align*}
Then, using Proposition \ref{prop:UnifToLagSect}, we have proved
$\psi_h^*i(X)\Omega_h = \psi_\Lag^*i(Y)\Omega_\Lag = 0$,
since the last equality holds for every $Y \in \vf(J^3\pi)$ and, in particular,
for every vector field $\Leg_o$-related to a vector field in $\P$.
\end{proof}

The diagram for this situation is the following
$$
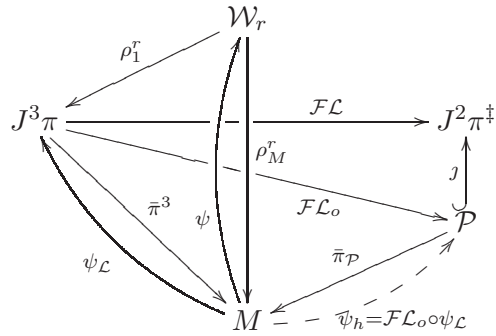
\xymatrix{
\ & \ & \W_r \ar[dll]_{\rho^r_1} \ar[ddd]^<(0.4){\rho^r_M} & \ & \\
J^3\pi \ar[ddrr]^{\bar{\pi}^3} \ar[rrrr]^<(0.65){\Leg}|(.425){\hole}|(.49){\hole} \ar[rrrrd]_<(0.65){\Leg_o}|(.42){\hole}|(.495){\hole} & \ & \ & \ & J^2\pi^\ddagger \\
\ & \ & \ & \ & \P \ar@{^{(}->}[u]^{\jmath} \ar[dll]_{\bar{\pi}_\P} \\
\ & \ & M \ar@/^1pc/[uull]^{\psi_\Lag} \ar@/_1pc/@{-->}[urr]_{\psi_h = \Leg_o \circ \psi_\Lag} \ar@/^1pc/[uuu]^<(0.3){\psi} & \ & \
}
$$

Now, assume that there exists a submanifold $\W_f \hookrightarrow \W_\Lag$ and a multivector field
$\X \in \vf^m(\W_r)$, defined at support on $\W_f$ and tangent to $\W_f$, which is a solution to the
equation \eqref{eqn:UnifDynEqMultiVFSing}. Now consider the submanifolds
$S_f = \rho_1^\Lag(\W_f) \hookrightarrow J^{3}\pi$ and
$P_f = \Leg(S_f) \hookrightarrow \P \hookrightarrow J^2\pi^\ddagger$. Using Theorem
\ref{thm:UnifToLagMultiVF}, from the holonomic multivector field $\X \in \vf^m(\W_r)$ we obtain the
corresponding holonomic multivector fields $\X_\Lag \in \vf^m(J^3\pi)$ solution to the equation
\eqref{eqn:LagDynEqMultiVF} at support on $S_f$. From this, one can prove that there are multivector
fields in $S_f$ (perhaps only on the points of another submanifold), which are $\Leg$-projectable to
$P_f$. So we have the diagram
$$
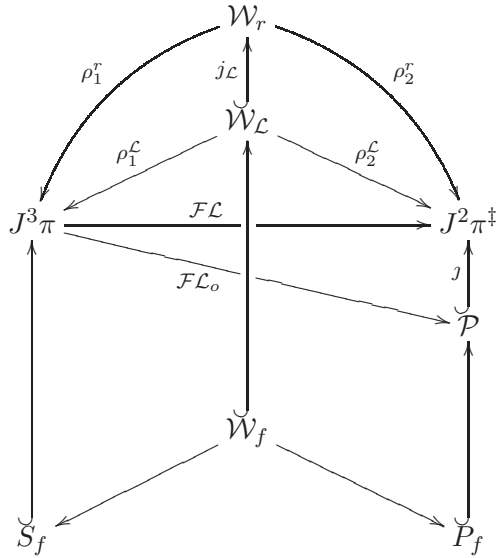
\xymatrix{
\ & \ & \W_r \ar@/_1.3pc/[ddll]_{\rho_1^r} \ar@/^1.3pc/[ddrr]^{\rho_2^r} & \ & \ \\
\ & \ & \W_\Lag \ar[dll]_{\rho_1^\Lag} \ar[drr]^{\rho_2^\Lag} \ar@{^{(}->}[u]^{j_\Lag} & \ & \ \\
J^3\pi \ar[rrrr]^<(0.35){\Leg}|(.495){\hole} \ar[drrrr]_<(0.35){\Leg_o}|(.495){\hole} & \ & \ & \ & J^2\pi^\ddagger \\
\ & \ & \ & \ & \P \ar@{^{(}->}[u]^-{\jmath} \\
\ & \ & \W_f \ar@{^{(}->}[uuu] \ar[dll] \ar[drr] & \ & \ \\
S_f \ar@{^{(}->}[uuu] & \ & \ & \ & P_f \ar@{^{(}->}[uu] \\
}
$$

Moreover, we can state the following result, which is the analogous theorem to Theorem
\ref{thm:UnifToHamMultiVFReg} in the case of almost-regular Lagrangian densities.

\begin{theorem}
Let $\X \in \vf^m(\W_r)$ be a locally decomposable, $\rho_M^r$-transverse
and integrable multivector field, defined at support on $\W_f$ and tangent to $\W_f$,
which is a solution to the equation \eqref{eqn:UnifDynEqMultiVFSing}. Then there exists a locally
decomposable, integrable and $(\bar{\pi}_{\P}^\ddagger)$-transverse multivector field
$\X_h \in \vf^m(\P)$, defined at support on $P_f$ and tangent to
$P_f$, which is a solution to the equation
\begin{equation}\label{eqn:HamDynEqMultiVFSing}
\restric{i(X_h)\Omega_h}{P_f} = 0 \, .
\end{equation}

Conversely, if $\X_h \in \vf^{m}(\P)$ is a locally decomposable,
$(\bar{\pi}_{\P}^\ddagger)$-transverse and integrable multivector field defined at support
on $P_f$ and tangent to $P_f$ which is a solution to the equation \eqref{eqn:HamDynEqMultiVFSing},
then there exist locally decomposable, $\rho_M^r$-transverse and integrable multivector fields
$\X \in \vf^{m}(\W_r)$, defined at support on $\W_f$ and
tangent to $\W_f$, which are solutions to the equation \eqref{eqn:UnifDynEqMultiVFSing}.
\end{theorem}

\section{Examples}
\label{sec:Examples}

\subsection{A first-order Lagragian density as a second-order one}

Let us first study the case of first-order classical field theories considered as second-order ones.
Hence, let $\pi \colon E \to M$ be the configuration bundle describing a classical field theory,
with $M$ being a $m$-dimensional orientable manifold and $E$ a $(m+n)$-dimensional
manifold. Let $\eta \in \Omega^{m}(M)$ be a fixed volume form for $M$, and $\Lag \in \Omega^m(J^1\pi)$
be a first-order Lagrangian density for this theory, that is, a $\bar{\pi}^1$-semibasic $m$-form
on $J^1\pi$. Since $\Lag$ is $\bar{\pi}^1$-semibasic, we can write $\Lag = L \cdot (\bar{\pi}^1)^*\eta$,
where $L \in C^\infty(J^1\pi)$ is the first-order Lagrangian function associated to $\Lag$ and
$\eta$.

Now, let $\Lag_o = (\pi_1^2)^*\Lag \in \Omega^{m}(J^2\pi)$ be the pull-back of $\Lag$ by the canonical
submersion $\pi_1^2 \colon J^2\pi \to J^1\pi$. Since $\Lag$ is $\bar{\pi}^1$-semibasic, we have that
$\Lag_o$ is $\bar{\pi}^2$-semibasic, and thus there exists a function $L_o = (\pi_1^2)^*L$ such that
$\Lag_o = L_o \cdot (\bar{\pi}^2)^*\eta$. Observe that we have
$$
\derpar{L_o}{u_I^\alpha} = 0 \, , \ \mbox{for every } |I| = 2 \, , \ 1 \leqslant \alpha \leqslant n \, ,
$$
and, therefore, this second-order Lagrangian density is always singular.

\paragraph{\textbf{Lagrangian-Hamiltonian formalism.}}

In this setting, the local expression of the local Hamiltonian function $\hat{H} \in C^\infty(\W_r)$
is exactly \eqref{eqn:HamiltonianFunctionLocal}, replacing $L$ by $L_o$. On the other hand, the
coordinate expressions of the forms $\Theta_r$ and $\Omega_r$ remain as in
\eqref{eqn:HamiltonCartanFormsLocal}.

Let $\psi \in \Gamma(\rho_M^r)$ be a section. Then, computing in coordinates the field equation
\eqref{eqn:UnifDynEqSect} in this particular case, we obtain the following system of equations
\begin{align*}
& \sum_{i=1}^{m}\derpar{p_\alpha^i}{x^i} - \derpar{\hat{L}_o}{u^\alpha} = 0 \, ,\\
& \sum_{j=1}^{m} \frac{1}{n(ij)} \, \derpar{p_\alpha^{1_i+1_j}}{x^j} + p_\alpha^i - \derpar{\hat{L}_o}{u_i^\alpha} = 0 \, , \\
& p_\alpha^{I} = 0 \, , \\
& u_i^\alpha - \derpar{u^\alpha}{x^i} = 0 \quad ; \quad u_{I}^\alpha - \sum_{1_i+1_j=I} \frac{1}{n(ij)} \, \derpar{u_i^\alpha}{x^j} = 0 \, .
\end{align*}
That is, the second-order multimomenta $p_\alpha^I$ vanish, and therefore these equations reduce to
\begin{align*}
& \sum_{i=1}^{m}\derpar{p_\alpha^i}{x^i} - \derpar{\hat{L}_o}{u^\alpha} = 0 \, , \\
& p_\alpha^i - \derpar{\hat{L}_o}{u_i^\alpha} = 0 \, , \\
& p_\alpha^{I} = 0 \, , \\
& u_i^\alpha - \derpar{u^\alpha}{x^i} = 0 \quad ; \quad u_{I}^\alpha - \sum_{1_i+1_j=I} \frac{1}{n(ij)} \, \derpar{u_i^\alpha}{x^j} = 0 \, .
\end{align*}
From these local equations, we obtain the coordinate expression of the Legendre map
$\Leg \colon J^3\pi \to J^2\pi^\ddagger$, which is
$$
\Leg^*p_\alpha^i = \derpar{\hat{L}_o}{u_i^\alpha} \quad ; \quad
\Leg^*p_\alpha^I = 0 \, ,
$$
that is, the coordinate expression of the Legendre map corresponding to a first-order classical
field theory.

On the other hand, by combining the first two groups of equations, we obtain the
Euler-Lagrange equations for classical field theories
$$
\restric{\derpar{\hat{L}_o}{u^\alpha}}{\psi}
- \restric{\frac{d}{dx^i} \, \derpar{\hat{L}_o}{u_i^\alpha}}{\psi} = 0 \, .
$$

Now, let $\X \in \vf^m(\W_r)$ be a locally decomposable multivector field given locally by
\eqref{eqn:UnifGenericMultiVFLocal}. Then the equation \eqref{eqn:UnifDynEqMultiVF} gives
locally the following system of equations
\begin{align*}
& F_j^\alpha = u_j^\alpha \quad ; \quad \sum_{1_i+1_j=I} \frac{1}{n(ij)} \, F_{i,j}^\alpha = u_I^\alpha \, , \\
& \sum_{i=1}^{m} G_{\alpha,i}^{i} = \derpar{\hat L_o}{u^\alpha} \, , \\
& \sum_{j=1}^{m} \frac{1}{n(ij)} \, G_{\alpha,j}^{1_i+1_j} = \derpar{\hat L_o}{u_i^\alpha} - p_\alpha^i \, , \\
& p_\alpha^K = 0 \, , \quad |K| = 2 \, .
\end{align*}
Furthermore, if we assume $\X$ to be holonomic, then we have the additional equations
$$
F_{i,j}^\alpha = u_{1_i+1_j}^\alpha \quad ; \quad
F_{I,j}^\alpha = u_{I + 1_j}^\alpha \, .
$$
From the field equations, we deduce that the first constraint submanifold $\W_c \hookrightarrow \W_r$
is given in coordinates by the local constraints $p^I_\alpha = 0$. The tangency condition for the
multivector field $\X$ along $\W_c$ enables us to determine all the coefficients $G_{\alpha,j}^I$,
with $1 \leqslant j \leqslant m$, $1 \leqslant \alpha \leqslant n$ and $|I| = 2$, in the following
way
$$
G_{\alpha,j}^I = 0 \, .
$$
Then, using the previous local field equations, we obtain the following additional constraints
$$
p_\alpha^i - \derpar{\hat{L}_o}{u_i^\alpha} = 0 \, ,
$$
which define a new submanifold $\W_\Lag \hookrightarrow \W_r$. Analyzing the tangency of $\X$ along
this new submanifold, we obtain the following equations
\begin{align*}
G_{\alpha,k}^{i} = \frac{d}{dx^k} \, \derpar{\hat{L}_o}{u_i^\alpha} \, .
\end{align*}
Using again the field equations, we obtain the Euler-Lagrange equations for a multivector
field, which are
\begin{align*}
\derpar{\hat{L}_o}{u^\alpha} - \frac{d}{dx^i} \, \derpar{\hat{L}_o}{u_i^\alpha}
+ \left( F_{i,j}^\beta - \frac{d}{dx^j} u_{i}^\beta \right) \derpars{\hat{L}_o}{u_i^\beta}{u_j^\alpha} = 0\, .
\end{align*}

That is, we obtain the coordinate expression of the field equations for first-order field
theories in the unified formalism, which were obtained previously in
\cite{art:Echeverria_Lopez_Marin_Munoz_Roman04}.

\paragraph{\textbf{Lagrangian formalism.}}

Now we recover the Lagrangian structures and equations from the unified setting. In order to obtain
the Poincar\'{e}-Cartan $m$-form $\Theta_\Lag = \widetilde{\Leg}^*\Theta_1^s \in \Omega^{m}(J^3\pi)$,
we need the extended Legendre map $\widetilde{\Leg} \colon J^3\pi \to J^2\pi^\dagger$. From the
results in Section \ref{sec:UnifFieldEquationsSect}, the extended Legendre map is locally given by
\eqref{eqn:ExtendedLegendreMapLocal}, which in our case reduces to
$$
\widetilde{\Leg}^*p_\alpha^i = \derpar{L_o}{u_i^\alpha} \quad ; \quad
\widetilde{\Leg}^*p_\alpha^I = 0 \quad ; \quad
\widetilde{\Leg}^*p = L_o - u_i^\alpha \derpar{L_o}{u_i^\alpha} \, .
$$
Therefore, the Poincar\'{e}-Cartan $m$-form is given locally by
$$
\Theta_\Lag = \derpar{L_o}{u_i^\alpha} (d u^\alpha \wedge d^{m-1}x_i - u_i^\alpha d^mx)
+ L_o d^mx \, ,
$$
which is exactly the Poincar\'{e}-Cartan $m$-form for a first-order classical field theory.

Now, if $\Omega_\Lag = -d\Theta_\Lag$, we recover the Lagrangian solutions for the field equations
from the unified formalism. In particular, if $\psi \in \Gamma(\rho^r_M)$ is a holonomic section
solution to the field equation \eqref{eqn:UnifDynEqSect}, then the section
$\psi_\Lag = \rho_1^r \circ \psi \in \Gamma(\bar{\pi}^3)$ is holonomic and is a solution to the field
equation \eqref{eqn:LagDynEqSect}. In coordinates, the component functions of the section
$\psi_\Lag = j^3\phi$, for some $\phi(x^i) = (x^i,u(x^i)) \in \Gamma(\pi)$, are a solution to the
Euler-Lagrange equation
$$
\restric{\derpar{L_o}{u^\alpha}}{j^3\phi}
- \restric{\frac{d}{dx^i} \, \derpar{L_o}{u_i^\alpha}}{j^3\phi} = 0 \, .
$$
Finally, if $\X \in \vf^{m}(\W_r)$ is a locally decomposable holonomic multivector field
solution to the field equation \eqref{eqn:UnifDynEqMultiVF}, then there exists a unique locally decomposable
holonomic multivector field $\X_\Lag \in \vf^{m}(J^3\pi)$ solution to the equation
\eqref{eqn:LagDynEqMultiVF}. In coordinates, the component functions of this multivector
field must satisfy the equation
$$
\derpar{L_o}{u^\alpha} - \frac{d}{dx^i} \, \derpar{L_o}{u_i^\alpha}
+ \left( F_{i,j}^\beta - \frac{d}{dx^j} u_{i}^\beta \right) \derpars{L_o}{u_i^\beta}{u_j^\alpha} = 0\, .
$$

\paragraph{\textbf{Hamiltonian formalism.}}

Observe that, in this situation, the second-order Lagrangian density $\Lag_o = (\pi^2_1)^*\Lag$
can not be regular. Nevertheless, it is straightforward to compute the coordinate expression of
a local Hamiltonian function $H$ that specifies the Hamiltonian $\mu$-section $h$ of a first-order
classical field theory as
$$
H(x^i,u^\alpha,u_i^\alpha,p_\alpha^i,p_\alpha^I) = p_\alpha^iu_i^\alpha - (\pi_1^3 \circ \sigma)^*L_o \, ,
$$
where $\sigma$ is any (local) section of the Legendre map associated to $L_o$. It is now straightforward
to obtain the Hamilton-De Donder-Weyl equations for this first-order classical field theory
\cite{art:Echeverria_DeLeon_Munoz_Roman07}.

\subsection{Loaded and clamped plate}

Let us consider a plate with clamped edges. We wish to determine the bending (or deflection)
perpendicular to the plane of the plate under the action of an external force given by a uniform
load. This system has been studied using a previous version of the unified formalism in
\cite{art:Campos_DeLeon_Martin_Vankerschaver09}, and can be modeled as a second-order field theory,
taking $M = \mathbb{R}^2$ as the base manifold (the plate) and the ``vertical'' bending as a fiber bundle
$E = \mathbb{R}^2 \times \mathbb{R} \stackrel{\pi}{\longrightarrow} \mathbb{R}^2$ (that is, the fibers are $1$-dimensional).

We consider in $M = \mathbb{R}^2$ the canonical coordinates $(x,y)$ of the Euclidean plane, and in $E = \mathbb{R}^3$
we take the global coordinates $(x,y,u)$ adapted to the bundle structure. Recall that $\mathbb{R}^2$ admits
a canonical volume form $\eta = d x \wedge d y \in \Omega^{2}(\mathbb{R}^2)$.

In the induced coordinates $(x,y,u,u_1,u_2,u_{(2,0)},u_{(1,1)},u_{(0,2)})$ of $J^2\pi$,
the Lagrangian density $\Lag \in \Omega^{2}(J^2\pi)$ for this field theory is given by
$$
\Lag = \frac{1}{2}( u_{(2,0)}^2 + 2u_{(1,1)}^2 + u_{(0,2)}^2 - 2qu ) \, dx \wedge dy \, ,
$$
where $q \in \mathbb{R}$ is a constant modeling the uniform load on the plate.

\paragraph{\textbf{Lagrangian-Hamiltonian formalism.}}

Following the results in Section \ref{sec:GeomSetting}, let us consider the fiber bundles
$$
\W = J^3\pi \times_{J^1\pi} J^2\pi^\dagger \quad ; \quad
\W_r = J^3\pi \times_{J^1\pi} J^2\pi^\ddagger \, ,
$$
with the natural coordinates introduced in the aforementioned Section.

Observe that, in this example, we have $\dim J^3\pi = 12$ and $\dim J^2\pi^\ddagger = 10$,
and therefore $\dim\W = 18$ and $\dim\W_r = 17$.

The Hamiltonian $\mu_\W$-section $\hat{h} \in \Gamma(\mu_\W)$ is specified by the local Hamiltonian function
\begin{align*}
\hat{H} &= p^1u_1 + p^2u_2 + p^{(2,0)}u_{(2,0)} + p^{(1,1)}u_{(1,1)} + p^{(0,2)}u_{(0,2)} \\
&\quad{} - \frac{1}{2}u_{(2,0)}^2 - u_{(1,1)}^2 - \frac{1}{2}u_{(0,2)}^2 + qu \, ,
\end{align*}
and the forms $\Theta_r \in \Omega^{m}(\W_r)$ and $\Omega_r \in \Omega^{m+1}(\W_r)$ are given by
\begin{align*}
\Theta_r &= -\hat{H}d x \wedge d y + p^1 d u \wedge d y - p^2 d u \wedge d x + p^{(2,0)}d u_1 \wedge d y
- \frac{1}{2} \, p^{(1,1)} d u_1 \wedge d x \\
& \quad{} + \frac{1}{2} \, p^{(1,1)} d u_2 \wedge d y - p^{(0,2)}d u_2 \wedge d x \, , \\
\Omega_r &= d \hat{H} \wedge d x \wedge d y - d p^1 d u \wedge d y + d p^2 d u \wedge d x - d p^{(2,0)}d u_1 \wedge d y \\
& \quad{} + \frac{1}{2} \, d p^{(1,1)} d u_1 \wedge d x
- \frac{1}{2} \, d p^{(1,1)} d u_2 \wedge d y + d p^{(0,2)}d u_2 \wedge d x \, .
\end{align*}

Let $\psi \in \Gamma(\rho_M^r)$ be a section. Then the field equation \eqref{eqn:UnifDynEqSect}
gives in coordinates the following system of equations
\begin{align*}
&\derpar{p^1}{x} + \derpar{p^2}{y} + q = 0 \, , \\
&\derpar{p^{(2,0)}}{x} + \frac{1}{2}\derpar{p^{(1,1)}}{y} + p^1 = 0 \quad ; \quad
\frac{1}{2} \derpar{p^{(1,1)}}{x} + \derpar{p^{(0,2)}}{y} + p^2 = 0 \, , \\[6pt]
&p^{(2,0)} - u_{(2,0)} = 0 \quad ; \quad p^{(1,1)} - 2u_{(1,1)} = 0 \quad ; \quad p^{(0,2)} - u_{(0,2)} = 0 \, , \\
&u_1 - \derpar{u}{x} = 0 \quad ; \quad u_2 - \derpar{u}{y} = 0 \, , \\
&u_{(2,0)} - \derpar{u_1}{x} = 0 \quad ; \quad
u_{(1,1)} - \frac{1}{2}\left( \derpar{u_1}{y} + \derpar{u_2}{x} \right) = 0 \quad ; \quad
u_{(0,2)} - \derpar{u_2}{y} = 0 \, .
\end{align*}
Combining the second and third group of equations, we obtain the constraints defining the submanifold
$\W_\Lag$, and hence the Legendre map associated to this Lagrangian density, which is the fiber bundle
map $\Leg \colon J^3\pi \to J^2\pi^\ddagger$ given locally by
\begin{equation*}
\begin{array}{c}
\displaystyle \Leg^*p^1 = -u_{(3,0)} - u_{(1,2)} \quad ; \quad
\Leg^*p^2 = -u_{(2,1)} - u_{(0,3)} \, , \\[10pt]
\displaystyle \Leg^*p^{(2,0)} = u_{(2,0)} \quad ; \quad
\Leg^*p^{(1,1)} = 2u_{(1,1)} \quad ; \quad
\Leg^*p^{(0,2)} = u_{(0,2)} \, .
\end{array}
\end{equation*}
Observe that the tangent map of $\Leg$ at every point $j^3\phi \in J^3\pi$ is given in coordinates
by the $10 \times 12$ real matrix
$$
T_{j^3\phi}\Leg =
\left(
\begin{array}{cccccccccccc}
1 & 0 & 0 & 0 & 0 & 0 & 0 & 0 & 0 & 0 & 0 & 0 \\
0 & 1 & 0 & 0 & 0 & 0 & 0 & 0 & 0 & 0 & 0 & 0 \\
0 & 0 & 1 & 0 & 0 & 0 & 0 & 0 & 0 & 0 & 0 & 0 \\
0 & 0 & 0 & 1 & 0 & 0 & 0 & 0 & 0 & 0 & 0 & 0 \\
0 & 0 & 0 & 0 & 1 & 0 & 0 & 0 & 0 & 0 & 0 & 0 \\
0 & 0 & 0 & 0 & 0 & 0 & 0 & 0 & -1 & 0 & -1 & 0 \\
0 & 0 & 0 & 0 & 0 & 0 & 0 & 0 & 0 & -1 & 0 & -1 \\
0 & 0 & 0 & 0 & 0 & 1 & 0 & 0 & 0 & 0 & 0 & 0 \\
0 & 0 & 0 & 0 & 0 & 0 & 2 & 0 & 0 & 0 & 0 & 0 \\
0 & 0 & 0 & 0 & 0 & 0 & 0 & 1 & 0 & 0 & 0 & 0
\end{array}
\right) \, .
$$
From this it is clear that $\textnormal{rank}(\Leg(j^3\phi)) = 10 = \dim J^2\pi^\ddagger$. Hence,
the restricted Legendre map is a submersion onto $J^2\pi^\ddagger$, and therefore the Lagrangian density
$\Lag \in \Omega^{2}(J^2\pi)$ is regular.

Finally, combining the first three groups of equations, we obtain the Euler-Lagrange equation
$$
u_{(4,0)} + 2u_{(2,2)} + u_{(0,4)} = q \Longleftrightarrow
\frac{\partial^4 u}{\partial x^4} + \frac{\partial^4 u}{\partial x^2 \partial y^2} + \frac{\partial^4 u}{\partial y^4} = q \, .
$$
This is the classical equation for the bending of a clamped plate under a uniform load $q$.

Now, let $\X \in \vf^2(\W_r)$ be a locally decomposable bivector field given locally by
\eqref{eqn:UnifGenericMultiVFLocal}. Then the equation \eqref{eqn:UnifDynEqMultiVF}
gives in coordinates the following system of equations
\begin{align*}
&F_1 = u_1 \quad ; \quad F_2 = u_2 \, , \\
&F_{1,1} = u_{(2,0)} \quad ; \quad
\frac{1}{2} \left( F_{1,2} + F_{2,1} \right) = u_{(1,1)} \quad ; \quad
F_{2,2} = u_{(0,2)} \, , \\
&G_1^1 + G_2^2 = - q \, , \\
&G_1^{(2,0)} + \frac{1}{2} \, G_2^{(1,1)} = -p^1 \quad ; \quad
\frac{1}{2} \, G_1^{(1,1)} + G_2^{(0,2)} = -p^2 \, , \\
&p^{(2,0)} - u_{(2,0)} = 0 \quad ; \quad
p^{(1,1)} - 2u_{(1,1)} = 0 \quad ; \quad
p^{(0,2)} - u_{(0,2)} = 0 \, .
\end{align*}
Moreover, if we assume that $\X$ is holonomic, then we have the following additional equations
$$
\begin{array}{c}
F_{1,2} = u_{(1,1)} \quad ; \quad F_{2,1} = u_{(1,1)} \quad ; \quad
F_{(2,0),1} = u_{(3,0)} \quad ; \quad F_{(2,0),2} = u_{(2,1)} \, , \\[5pt]
F_{(1,1),1} = u_{(2,1)} \quad ; \quad F_{(1,1),2} = u_{(1,2)} \quad ; \quad
F_{(0,2),1} = u_{(1,2)} \quad ; \quad F_{(0,2),2} = u_{(0,3)} \, .
\end{array}
$$
From the field equations, we deduce that the first constraint submanifold $\W_c \hookrightarrow \W_r$
is given in coordinates by the local constraints
$$
p^{(2,0)} - u_{(2,0)} = 0 \quad ; \quad p^{(1,1)} - 2u_{(1,1)} = 0 \quad ; \quad p^{(0,2)} - u_{(0,2)} = 0 \, .
$$
The tangency condition for the multivector field $\X$ along $\W_c$ enables us to determine all the
coefficients $G_i^I$, with $i = 1,2$ and $|I| = 2$, in the following way
\begin{align*}
G_1^{(2,0)} = u_{(3,0)} \quad ; \quad G_1^{(1,1)} = 2u_{(2,1)} \quad ; \quad G_1^{(0,2)} = u_{(1,2)} \, , \\
G_2^{(2,0)} = u_{(2,1)} \quad ; \quad G_2^{(1,1)} = 2u_{(1,2)} \quad ; \quad G_2^{(0,2)} = u_{(0,3)} \, .
\end{align*}
Then, using the previous field equations, we obtain the following additional constraints
$$
p^1 + u_{(3,0)} + u_{(1,2)} = 0 \quad ; \quad p^2 + u_{(2,1)} + u_{(0,3)} = 0 \, ,
$$
which define a new submanifold $\W_\Lag \hookrightarrow \W_r$. Analyzing the tangency of the multivector
field $\X$ along this new submanifold $\W_\Lag$, we obtain the following equations
\begin{align*}
G_1^1 + F_{(3,0),1} + F_{(1,2),1} = 0 \quad ; \quad G_1^2 + F_{(2,1),1} + F_{(0,3),1} = 0 \, , \\
G_2^1 + F_{(3,0),2} + F_{(1,2),2} = 0 \quad ; \quad G_1^2 + F_{(2,1),2} + F_{(0,3),2} = 0 \, .
\end{align*}
Using again the field equations, we obtain the Euler-Lagrange equation for
a multivector field, which is
$$
F_{(3,0),1} + F_{(1,2),1} + F_{(2,1),2} + F_{(0,3),2} = q \, .
$$
Observe that if $\psi \in \Gamma(\rho_M^r)$ is an integral section of $\X$, then its component functions
must satisfy the Euler-Lagrange equation previously obtained for sections.

\paragraph{\textbf{Lagrangian formalism.}}

Now we recover the Lagrangian structures and equations from the unified setting. In order to obtain
the Poincar\'{e}-Cartan $2$-form $\Theta_\Lag = \widetilde{\Leg}^*\Theta_1^s \in \Omega^{2}(J^3\pi)$,
we need the extended Legendre map $\widetilde{\Leg} \colon J^3\pi \to J^2\pi^\dagger$. From the results
in Section \ref{sec:UnifFieldEquationsSect}, the extended Legendre map is given locally by
\begin{gather*}
\widetilde{\Leg}^*p^1 = -u_{(3,0)} - u_{(1,2)} \quad ; \quad
\widetilde{\Leg}^*p^2 = -u_{(2,1)} - u_{(0,3)} \, , \\[10pt]
\widetilde{\Leg}^*p^{(2,0)} = u_{(2,0)} \quad ; \quad
\widetilde{\Leg}^*p^{(1,1)} = 2u_{(1,1)} \quad ; \quad
\widetilde{\Leg}^*p^{(0,2)} = u_{(0,2)} \, , \\
\widetilde{\Leg}^*p = u_{(3,0)}u_1 + u_{(1,2)}u_1 + u_{(2,1)}u_2 + u_{(0,3)}u_2
- \frac{1}{2}u_{(2,0)}^2 - u_{(1,1)}^2 - \frac{1}{2}u_{(0,2)}^2 - qu \, .
\end{gather*}
Therefore, the Poincar\'{e}-Cartan $2$-form is given locally by
\begin{align*}
\Theta_\Lag &=
\left( \frac{1}{2}u_{(2,0)}^2 + u_{(1,1)}^2 + \frac{1}{2}u_{(0,2)}^2 + qu
- u_{(3,0)}u_1 - u_{(1,2)}u_1 - u_{(2,1)}u_2 \right. \\
& \left. \phantom{\frac{1}{2}} - u_{(0,3)}u_2 \right) d x \wedge d y
- (u_{(3,0)} + u_{(1,2)}) d u \wedge d y + (u_{(2,1)} + u_{(0,3)}) d u \wedge d x \\
&\quad{} + u_{(2,0)}d u_1 \wedge d y - u_{(1,1)}d u_1 \wedge d x â?
+ u_{(1,1)}d u_2 \wedge d y - u_{(0,2)}d u_2 \wedge d x \, .
\end{align*}

Now, if $\Omega_\Lag = -d\Theta_\Lag$, we recover the Lagrangian solutions for the field equations
from the unified formalism. In particular, if $\psi \in \Gamma(\rho^r_M)$ is a holonomic section
solution to the field equation \eqref{eqn:UnifDynEqSect}, then the section
$\psi_\Lag = \rho_1^r \circ \psi \in \Gamma(\bar{\pi}^3)$ is holonomic and is a solution to the field
equation \eqref{eqn:LagDynEqSect}. In coordinates, the component functions of the section
$\psi_\Lag = j^3\phi$, for some $\phi(x,y) = (x,y,u(x,y)) \in \Gamma(\pi)$, are a solution to the
Euler-Lagrange equation
$$
u_{(4,0)} + 2u_{(2,2)} + u_{(0,4)} = q \, .
$$
Finally, if $\X \in \vf^{2}(\W_r)$ is a locally decomposable holonomic multivector field
solution to the field equation \eqref{eqn:UnifDynEqMultiVF}, then there exists a unique locally decomposable
holonomic multivector field $\X_\Lag \in \vf^{2}(J^3\pi)$ solution to the equation
\eqref{eqn:LagDynEqMultiVF}. In coordinates, the component functions of this multivector
field must satisfy the equation
$$
F_{(3,0),1} + F_{(1,2),1} + F_{(2,1),2} + F_{(0,3),2} = q \, .
$$

\paragraph{\textbf{Hamiltonian formalism.}}

Since the Lagrangian density is regular, the Hamiltonian formalism takes place in an open set of
$J^2\pi^\ddagger$. In fact, $\Lag \in \Omega^{2}(J^2\pi)$ is a hyperregular Lagrangian density, since
the restricted Legendre map admits global sections. For instance, the map
$$
\Upsilon =
\left(x,y,u,u_1,u_2,p^{(2,0)},\frac{1}{2} \, p^{(1,1)},p^{(0,2)},-\frac{1}{2} \, p^1,-\frac{1}{2} \, p^2,-\frac{1}{2} \, p^1,-\frac{1}{2} \, p^2\right) \, ,
$$
is a section of $\Leg$ defined everywhere in $J^2\pi^\ddagger$.

In the natural coordinates of $J^2\pi^\ddagger$, the local Hamiltonian function $H$ that specifies
the Hamiltonian $\mu$-section $h$ is given by
$$
H = p^1u_1 + p^2u_2 + \frac{1}{2}\left(p^{(2,0)}\right)^2 + \frac{1}{4}\left(p^{(1,1)}\right)^2
+ \frac{1}{2}\left(p^{(0,2)}\right)^2 + qu \, .
$$
Hence, the Hamilton-Cartan $2$-form $\Theta_h \in \Omega^{2}(J^2\pi^\ddagger)$ is given locally by
\begin{align*}
\Theta_h &= \left( - p^1u_1 - p^2u_2 - \frac{1}{2}\left(p^{(2,0)}\right)^2 - \frac{1}{4}\left(p^{(1,1)}\right)^2
- \frac{1}{2}\left(p^{(0,2)}\right)^2 - qu \right) d x \wedge d y \\
&\quad{} + p^1 d u \wedge d y - p^2 d u \wedge d x
+ p^{(2,0)}d u_1 \wedge d y - \frac{1}{2} \, p^{(1,1)} d u_1 \wedge d x \\
&\quad{} + \frac{1}{2} \, p^{(1,1)} d u_2 \wedge d y - p^{(0,2)}d u_2 \wedge d x \, .
\end{align*}

Now we recover the Hamiltonian field equations and solutions from the unified setting. First, let
$\psi \in \Gamma(\rho_M^r)$ be a (holonomic) section solution to the field equation \eqref{eqn:UnifDynEqSect}.
Then, the section $\psi_h = \rho_2^r \circ \psi \in \Gamma(\bar{\pi}_{J^1\pi}^\ddagger)$ is a
solution to the equation \eqref{eqn:HamDynEqSectReg}. In coordinates, the component functions
of $\psi_h$ must satisfy the following system of partial differential equations
\begin{gather*}
\derpar{u}{x} = u_1  \ \ ;  \ \ \derpar{u}{y} = u_2  \ \; \ \
\derpar{u_1}{x} = p^{(2,0)}  \ \ ;  \ \
\derpar{u_2}{x} + \derpar{u_1}{y} = p^{(1,1)}  \ \ ;  \ \
\derpar{u_2}{y} = p^{(0,2)} \, , \\
\derpar{p^1}{x} + \derpar{p^2}{y} = q  \ \ ;  \ \
\derpar{p^{(2,0)}}{x} + \frac{1}{2} \, \derpar{p^{(1,1)}}{y} = - p^1  \ \ ;  \ \
\frac{1}{2} \, \derpar{p^{(1,1)}}{x} + \derpar{p^{(0,2)}}{y} = - p^2 \, .
\end{gather*}
Finally, if $\X \in \vf^2(\W_r)$ is a locally decomposable multivector field solution to the equation
\eqref{eqn:UnifDynEqMultiVF}, then there exists a locally decomposable multivector field
$\X_h \in \vf^{2}(J^2\pi^\ddagger)$ solution to the equation \eqref{eqn:HamDynEqMultiVFReg}. If $\X_h$
is locally given by \eqref{eqn:HamGenericMultiVFLocalReg}, then its component functions must satisfy
the following equations
\begin{gather*}
F_1 = u_1  \ \ ;  \ \ F_2 = u_2  \ \ ;  \ \
F_{1,1} = p^{(2,0)}  \ \ ;  \ \
F_{2,1} + F_{1,2} = p^{(1,1)}  \ \ ;  \ \
F_{2,2} = p^{(0,2)} \, , \\
G^{1}_{1} + G^{2}_2 = q  \ \ ;  \ \
G^{(2,0)}_1 + \frac{1}{2} \, G^{(1,1)}_2 = - p^1  \ \ ;  \ \
\frac{1}{2} \, G^{(1,1)}_1 + G^{(0,2)}_2 = - p^2 \, .
\end{gather*}

\subsection{Korteweg-de Vries equation}

Next we derive the Korteweg-de Vries equation, usually denoted as the KdV equation
for short, using the geometric formalism introduced in this paper. The KdV equation is a
mathematical model of waves on shallow water surfaces, and has become the prototypical example
of a non-linear partial differential equation whose solutions can be specified exactly. Many papers
are devoted to analyzing this model and, in particular, some previous multisymplectic descriptions
of it are available for instance \cite{art:Ascher_McLachlan2005,proc:Gotay88,art:Zhao_Qin2000}.
A further analysis using a different version of the unified formalism is given in \cite{art:Vitagliano10}.

The usual form of the KdV equation is
$$
\derpar{y}{t} - 6y\derpar{y}{x} + \frac{\partial^3y}{\partial x^3} = 0 \, ,
$$
that is, a non-linear, dispersive partial differential equation for a real function $y$ depending
on two real variables, the space $x$ and the time $t$. It is known that the KdV equation can be
derived from a least action principle as the Euler-Lagrange equation of the Lagrangian density
$$
\Lag = \frac{1}{2}\derpar{u}{x}\derpar{u}{t} - \left(\derpar{u}{x}\right)^3 - \frac{1}{2} \left(\frac{\partial^2u}{\partial x^2}\right)^2 \, ,
$$
where $y = \partial u / \partial x$. It is therefore clear that we can use our formulation to derive
the KdV equation as the field equations of a second-order field theory with a
$2$-dimensional base manifold and a $1$-dimensional fiber over this base.

Hence, let $M = \mathbb{R}^2$ with global coordinates $(x,t)$, and $E = \mathbb{R}^2 \times \mathbb{R}$ with
natural coordinates adapted to the bundle structure, $(x,t,u)$. In these coordinates, the canonical
volume form in $\mathbb{R}^2$ is given by $\eta = d x \wedge d t \in \Omega^2(\mathbb{R}^2)$.

In the induced coordinates $(x,t,u,u_1,u_2,u_{(2,0)},u_{(1,1)},u_{(0,2)})$ of $J^2\pi$,
the Lagrangian density $\Lag \in \Omega^{2}(J^2\pi)$ given above may be written as
$$
\Lag = \frac{1}{2} \left( u_1u_2 - 2u_1^3 - u_{(2,0)}^2\right) d x \wedge d t \, .
$$

\paragraph{\textbf{Lagrangian-Hamiltonian formalism.}}

Following Section \ref{sec:GeomSetting}, consider the fiber bundles
$$
\W = J^3\pi \times_{J^1\pi} J^2\pi^\dagger \quad ; \quad
\W_r = J^3\pi \times_{J^1\pi} J^2\pi^\ddagger \, ,
$$
with the natural coordinates introduced in the aforementioned Section.
Observe that, as in the previous example, we have $\dim J^3\pi = 12$ and $\dim J^2\pi^\ddagger = 10$,
and therefore $\dim\W = 18$ and $\dim\W_r = 17$.

The Hamiltonian $\mu_\W$-section $\hat{h} \in \Gamma(\mu_\W)$ is specified
by the local Hamiltonian function
$$
\hat{H} = p^1u_1 + p^2u_2 + p^{(2,0)}u_{(2,0)} + p^{(1,1)}u_{(1,1)} + p^{(0,2)}u_{(0,2)}
- \frac{1}{2}u_1u_2 + u_1^3 + \frac{1}{2}u_{(2,0)}^2 \, ,
$$
and the Hamilton-Cartan forms have the same expressions as in the previous example, replacing the
local Hamiltonian function.

Let $\psi \in \Gamma(\rho_M^r)$ be a section. Then the field equation \eqref{eqn:UnifDynEqSect}
gives in coordinates the following system of equations
\begin{align*}
&\derpar{p^1}{x} + \derpar{p^2}{t} = 0 \, , \\
&\derpar{p^{(2,0)}}{x} + \frac{1}{2} \, \derpar{p^{(1,1)}}{t} + p^1 - \frac{1}{2}u_2 + 3u_1^2 = 0 \quad ; \quad
\frac{1}{2} \, \derpar{p^{(1,1)}}{x} + \derpar{p^{(0,2)}}{t} + p^2 - \frac{1}{2}u_1 = 0 \, , \\[6pt]
&p^{(2,0)} + u_{(2,0)} = 0 \quad ; \quad p^{(1,1)} = 0 \quad ; \quad p^{(0,2)} = 0 \, , \\[6pt]
&u_1 - \derpar{u}{x} = 0 \quad ; \quad u_2 - \derpar{u}{t} = 0 \, , \\
&u_{(2,0)} - \derpar{u_1}{x} = 0 \quad ; \quad
u_{(1,1)} - \frac{1}{2}\left( \derpar{u_1}{t} + \derpar{u_2}{x} \right) = 0 \quad ; \quad
u_{(0,2)} - \derpar{u_2}{t} = 0 \, .
\end{align*}
From these local equations, we obtain the coordinate expression of the Legendre map
$\Leg \colon J^3\pi \to J^2\pi^\ddagger$, which is
$$
\begin{array}{c}
\displaystyle \Leg^*p^1 = \frac{1}{2}u_2 - 3u_1^2 + u_{(3,0)} \quad ; \quad
\Leg^*p^2 = \frac{1}{2} u_1 \, , \\[10pt]
\displaystyle \Leg^*p^{(2,0)} = - u_{(2,0)} \quad ; \quad
\Leg^*p^{(1,1)} = 0 \quad ; \quad
\Leg^*p^{(0,2)} = 0 \, .
\end{array}
$$
The tangent map of $\Leg$ at every point $j^3\phi \in J^3\pi$ is given in coordinates by
$$
T_{j^3\phi}\Leg =
\left(
\begin{array}{cccccccccccc}
1 & 0 & 0 & 0 & 0 & 0 & 0 & 0 & 0 & 0 & 0 & 0 \\
0 & 1 & 0 & 0 & 0 & 0 & 0 & 0 & 0 & 0 & 0 & 0 \\
0 & 0 & 1 & 0 & 0 & 0 & 0 & 0 & 0 & 0 & 0 & 0 \\
0 & 0 & 0 & 1 & 0 & 0 & 0 & 0 & 0 & 0 & 0 & 0 \\
0 & 0 & 0 & 0 & 1 & 0 & 0 & 0 & 0 & 0 & 0 & 0 \\
0 & 0 & 0 & -6u_1 & 1/2 & 0 & 0 & 0 & 1 & 0 & 0 & 0 \\
0 & 0 & 0 & 1/2 & 0 & 0 & 0 & 0 & 0 & 0 & 0 & 0 \\
0 & 0 & 0 & 0 & 0 & -1 & 0 & 0 & 0 & 0 & 0 & 0 \\
0 & 0 & 0 & 0 & 0 & 0 & 0 & 0 & 0 & 0 & 0 & 0 \\
0 & 0 & 0 & 0 & 0 & 0 & 0 & 0 & 0 & 0 & 0 & 0
\end{array}
\right) \, .
$$
From this it is clear that $\textnormal{rank}(\Leg(j^3\phi)) = 7 < 10 = \dim J^2\pi^\ddagger$.
Therefore, the Lagrangian density $\Lag \in \Omega^{2}(J^2\pi)$ is singular.

Finally, by combining the first three groups of equations, we obtain the second-order
Euler-Lagrange equation for this field theory
$$
u_{(1,1)} - 6u_1u_{(2,0)} + u_{(4,0)} = 0 \longleftrightarrow
\frac{\partial^2 u}{\partial t \, \partial x} - 6\,\derpar{u}{x} \, \frac{\partial^2 u}{\partial x^2} + \frac{\partial^4 u}{\partial x^4} = 0 \, ,
$$
which, taking $y = \partial u / \partial x$, is the usual Korteweg-de Vries equation.

Now, let $\X \in \vf^2(\W_r)$ be a locally decomposable $2$-vector field with coordinate expression
\eqref{eqn:UnifGenericMultiVFLocal}. Then the field equation \eqref{eqn:UnifDynEqMultiVF}
gives in coordinates the following system of equations
\begin{align*}
&F_1 = u_1 \quad ; \quad F_2 = u_2 \, , \\
&F_{1,1} = u_{(2,0)} \quad ; \quad
\frac{1}{2} \left( F_{1,2} + F_{2,1} \right) = u_{(1,1)} \quad ; \quad
F_{2,2} = u_{(0,2)} \, , \\
&G_1^1 + G_2^2 = 0 \, , \\
&G_1^{(2,0)} + \frac{1}{2} \, G_2^{(1,1)} = \frac{1}{2} u_2 - 3u_1^2 - p^1 \quad ; \quad
\frac{1}{2}G_1^{(1,1)} + G_2^{(0,2)} = \frac{1}{2} u_1 - p^2 \, , \\
&p^{(2,0)} + u_{(2,0)} = 0 \quad ; \quad
p^{(1,1)} = 0 \quad ; \quad
p^{(0,2)} = 0 \, .
\end{align*}
Moreover, if we assume that $\X$ is holonomic, then we have the following additional equations
$$
\begin{array}{c}
F_{1,2} = u_{(1,1)} \quad ; \quad F_{2,1} = u_{(1,1)} \quad ; \quad
F_{(2,0),1} = u_{(3,0)} \quad ; \quad F_{(2,0),2} = u_{(2,1)} \, , \\[5pt]
F_{(1,1),1} = u_{(2,1)} \quad ; \quad F_{(1,1),2} = u_{(1,2)} \quad ; \quad
F_{(0,2),1} = u_{(1,2)} \quad ; \quad F_{(0,2),2} = u_{(0,3)} \, .
\end{array}
$$
From the coordinate expression of the field equation, we obtain the local constraints defining
the first constraint submanifold $\W_c \hookrightarrow \W_r$, which are
$$
p^{(2,0)} + u_{(2,0)} = 0 \quad ; \quad p^{(1,1)} = 0 \quad ; \quad p^{(0,2)} = 0 \, .
$$
The tangency condition for the $2$-vector field $\X$ along $\W_c$ gives the following local equations
\begin{align*}
G_1^{(2,0)} + u_{(3,0)} = 0 \quad ; \quad G_1^{(1,1)} = 0 \quad ; \quad G_1^{(0,2)} = 0 \, , \\
G_2^{(2,0)} + u_{(2,1)} = 0 \quad ; \quad G_2^{(1,1)} = 0 \quad ; \quad G_2^{(0,2)} = 0 \, .
\end{align*}
Then, using the local equations obtained above, we have the following additional constraints
$$
p^1 - \frac{1}{2}u_2 + 3u_1^2 - u_{(3,0)} = 0 \quad ; \quad
p^2 - \frac{1}{2} u_1 = 0 \, ,
$$
which define a new submanifold $\W_\Lag \hookrightarrow \W_r$. Analyzing the tangency
of the multivector field along this new submanifold $\W_\Lag$, we obtain the following equations
\begin{align*}
G_1^1 - \frac{1}{2}u_{(1,1)} + 6u_1u_{(2,0)} - F_{(3,0),1} = 0 \quad ; \quad
G_1^2 - \frac{1}{2}u_{(2,0)} = 0 \, , \\[6pt]
G_2^1 - \frac{1}{2}u_{(0,2)} + 6u_1u_{(1,1)} - F_{(3,0),2} = 0 \quad ; \quad
G_2^2 - \frac{1}{2}u_{(1,1)} = 0 \, .
\end{align*}
Using again the field equations, we obtain the Euler-Lagrange equation for a multivector field
$$
u_{(1,1)} - 6u_1u_{(2,0)} + F_{(3,0),1} = 0 \, ,
$$
from where we can determinate $F_{(3,0),1}$ as
$$
F_{(3,0),1} = 6u_1u_{(2,0)} - u_{(1,1)}  \, .
$$

\noindent\textbf{Remark:}
Observe that, in this case, the Lagrangian density is singular, but there are no additional
constraints. This implies that the final constraint submanifold is the whole submanifold $\W_\Lag$
in the unified formalism.

\paragraph{\textbf{Lagrangian formalism.}}

Now we recover the Lagrangian formalism from the unified setting. First, we need the coordinate
expression of the extended Legendre map $\widetilde{\Leg} \colon J^3\pi \to J^2\pi^\dagger$. From
the results in Section \ref{sec:UnifFieldEquationsSect}, the local expression of $\widetilde{\Leg}$ is
\begin{gather*}
\Leg^*p^1 = \frac{1}{2}u_2 - 3u_1^2 + u_{(3,0)} \quad ; \quad
\Leg^*p^2 = \frac{1}{2} u_1 \, , \\
\displaystyle \Leg^*p^{(2,0)} = - u_{(2,0)} \quad ; \quad
\Leg^*p^{(1,1)} = 0 \quad ; \quad
\Leg^*p^{(0,2)} = 0 \, , \\
\displaystyle \widetilde{\Leg}^*p =
- \frac{1}{2}u_1u_2 + 2u_1^3 - u_{(3,0)}u_1 + \frac{1}{2}u_{(2,0)}^2 \, .
\end{gather*}
Therefore, the Poincar\'{e}-Cartan $2$-form $\Theta_\Lag = \widetilde{\Leg}^*\Theta_1^s \in \Omega^{2}(J^3\pi)$
is given locally by
\begin{align*}
\Theta_\Lag &=
\left( \frac{1}{2}u_1u_2 - 2u_1^3 + u_{(3,0)}u_1 - \frac{1}{2}u_{(2,0)}^2 \right) d x \wedge d y \\
&\quad{} + \left( \frac{1}{2}u_2 - 3u_1^2 + u_{(3,0)} \right) d u \wedge d y
- \frac{1}{2}u_1 d u \wedge d x
- u_{(2,0)}d u_1 \wedge d y \, .
\end{align*}

Let $\psi \in \Gamma(\rho^r_M)$ be a holonomic section solution to the field equation
\eqref{eqn:UnifDynEqSect}. Then, the section $\psi_\Lag = \rho_1^r \circ \psi \in \Gamma(\bar{\pi}^3)$
is holonomic and is a solution to the Lagrangian field equation \eqref{eqn:LagDynEqSect}. In
coordinates, the component functions of the section $\psi_\Lag = j^3\phi$ for some
$\phi(x,t) = (x,t,u(x,t)) \in \Gamma(\pi)$, are a solution to the Euler-Lagrange equation
$$
u_{(1,1)} - 6u_1u_{(2,0)} + u_{(4,0)} = 0 \, .
$$
On the other hand, if $\X \in \vf^{2}(\W_r)$ is a locally decomposable holonomic multivector field
solution to the field equation \eqref{eqn:UnifDynEqMultiVF}, then there exists a unique locally decomposable
holonomic multivector field $\X_\Lag \in \vf^{2}(J^3\pi)$ solution to the equation
\eqref{eqn:LagDynEqMultiVF}. In coordinates, the component functions of this multivector
field must satisfy the equation
$$
F_{(3,0),1} = 6u_1u_{(2,0)} - u_{(1,1)}  \, .
$$

\paragraph{\textbf{Hamiltonian formalism.}}

Since the Lagrangian density is singular, the Hamiltonian formalism takes place in the submanifold
$\P = \textnormal{Im}(\Leg) \hookrightarrow J^2\pi^\ddagger$. Bearing in mind the coordinate
expression of the Legendre map, the submanifold $\P$ is locally defined by the constraints
$$
p^2 - \frac{1}{2}u_1 = 0 \quad ; \quad
p^{(1,1)} = 0 \quad ; \quad
p^{(0,2)} = 0 \, .
$$
Observe that $\dim\P = \textnormal{rank}(\Leg) = 7$.

The natural coordinates $(x,t,u,u_1,u_2,p^1,p^2,p^{(2,0)},p^{(1,1)},p^{(0,2)})$ in $J^2\pi^\ddagger$
induce coordinates $(x,t,u,u_1,u_2,p^1,p^{(2,0)})$ in $\P$, with the natural embedding
$\jmath \colon \P \hookrightarrow J^2\pi^\ddagger$ given locally by
$$
\jmath^*p^2 = \frac{1}{2}u_1 \quad ; \quad
\jmath^*p^{(1,1)} = 0 \quad ; \quad \jmath^*p^{(0,2)} = 0 \, .
$$

In these coordinates, the local Hamiltonian function that specifies the Hamiltonian section $h$ is given by
$$
H = p^1u_1 + u_1^3 - \frac{1}{2}\left(p^{(2,0)}\right)^2 \, .
$$

Therefore, the Hamilton-Cartan $2$-form $\Theta_h = h^*\Theta_1^s \in \Omega^{2}(\P)$ is given locally by
\begin{align*}
\Theta_h &= \left( \frac{1}{2}\left(p^{(2,0)}\right)^2 - p^1u_1 - u_1^3 \right) d x \wedge d t
+ p^1 d u \wedge d t \\
&\quad{} - \frac{1}{2}u_1 d u \wedge d x + p^{(2,0)}d u_1 \wedge d t  \, .
\end{align*}

Now we recover the Hamiltonian field equations. If $\psi \in \Gamma(\rho_M^r)$ is a (holonomic)
section solution to the field equation \eqref{eqn:UnifDynEqSect}, then the section
$\psi_h = \Leg \circ \rho_1^r \circ \psi \in \Gamma(\bar{\pi}_\P)$ is a solution to the equation
\eqref{eqn:HamDynEqSectSing}. In coordinates, the component functions of $\psi_h$ must satisfy the
following system of partial differential equations
$$
\derpar{u}{x} = u_1 \quad ; \quad
\frac{1}{2}\derpar{u}{t} = p^1 + 3u_1^2 \quad ; \quad
\derpar{p^1}{x} + \frac{1}{2}\derpar{u_1}{t} = 0 \quad ; \quad
\derpar{u_1}{x} = -p^{(2,0)} \, .
$$
Finally, if $\X \in \vf^2(\W_r)$ is a locally decomposable $2$-vector field solution to the equation
\eqref{eqn:UnifDynEqMultiVF}, then there exists a locally decomposable $2$-vector field
$\X_h \in \vf^{2}(\P)$ solution to the equation \eqref{eqn:HamDynEqMultiVFSing}. If $\X_h$ is locally given by
\begin{align*}
\X_h &= \left(  \derpar{}{x} + F_1\derpar{}{u} + F_{1,1}\derpar{}{u_1} + F_{2,1}\derpar{}{u_2}
+ G_{1}^1 \derpar{}{p^1} + G_{1}^{(2,0)} \derpar{}{p^{(2,0)}} \right) \\
&\quad{} \wedge \left(  \derpar{}{t} + F_2\derpar{}{u} + F_{1,2}\derpar{}{u_1} + F_{2,2}\derpar{}{u_2}
+ G_{2}^1 \derpar{}{p^1} + G_{2}^{(2,0)} \derpar{}{p^{(2,0)}} \right) \, ,
\end{align*}
then its component functions must satisfy the following equations
$$
F_1 = u_1 \quad ; \quad \frac{1}{2} F_2 = p^1 + 3u_1^2 \quad ; \quad
G^1_1 + \frac{1}{2}F_{1,2} = 0 \quad ; \quad
F_{1,1} = -p^{(2,0)} \, .
$$

\section{Conclusions and further research}
\label{sec:Conclusions}

We develop a new multisymplectic framework for describing higher-order field theories, and, in
particular, second-order ones which are the most relevant in physics (to the best of our knowledge,
the most interesting higher-order models and theories in physics are of second-order). This model is
based on the extension of the so-called \emph\emph{Skinner-Rusk unified formalism} from mechanical
systems to higher-order field theories, and thereby complements previous papers such as
\cite{art:Campos_DeLeon_Martin_Vankerschaver09,art:Vitagliano10},
in which analogous but different formulations are given.

The key points of the formalism are as follows:
\begin{itemize}
\item
The Skinner-Rusk formalism is a special case of what (in the modern terminology) is called a
\emph{Dirac structure}. It unifies in a single frame the Lagrangian and Hamiltonian formalisms,
and hence gives a unified version of the Euler-Lagrange and the Hamilton equations.

In our case, the $4th$-order Euler-Lagrange equations and the Hamilton-De Donder-Weil equations
for field theories described by $2nd$-order Lagrangian densities are stated in a combined form using
both sections and multivector fields  in a suitable fiber bundle over the configuration bundle of
the theory, $E \stackrel{\pi}{\longrightarrow} M$. This bundle is the restricted $2$-symmetric
jet-multimomentum bundle $\W_r=J^3\pi \times_{J^1\pi} J^{2}\pi^\ddagger$, which is a quotient bundle
of the extended $2$-symmetric jet-multimomentum bundle $\W=J^3\pi \times_{J^1\pi} J^{2}\pi^\dagger$,
where $J^{2}\pi^\dagger$ is the $2$-symmetric multimomentum bundle introduced in
\cite{art:Saunders_Crampin90}, and $J^{2}\pi^\ddagger = J^2\pi^\dagger / \Lambda^m_1(J^1\pi)$.
The use of this bundle is the crucial point for univocally defining a Legendre map, and therefore
the Poincar\'{e}-Cartan forms.

As usual, the physical information of the theory is given by a Lagrangian density, although the
geometry is provided by the canonical multisymplectic form $\Omega_1$ with which the $2$-symmetric
multimomentum bundle is endowed. This enables us to construct the form $\Omega_r$ which induces
the geometry of $\W_r$. Thus, in the unified formalism the geometry and the physical information are
separated.

\item
As is characteristic in the unified formalism, independently of the regularity of the Lagrangian
density, $\Omega_r$ is a premultisymplectic form in $\W_r$. Hence, the compatibility condition for
the field equations and the subsequent tangency or consistent condition for their solutions allows us
to determine univocally the Legendre map, thanks to the symmetry relation introduced in the highest-order
multimomenta coordinates. This relation equals the number of highest-order multimomenta with the number
of highest-order ``velocities'' in the Lagrangian density, and therefore enables us to establish a
$1$-to-$1$ correspondence between these two sets of coordinates, giving rise to the highest-order
equations defining the Legendre map. If the Lagrangian is regular (in the sense given in Definition
\ref{reglag}), then the constraint algorithm stops at the first level; otherwise it continues in the usual way.

Furthermore, as stated above, from the form $\Omega_r$ in the unified formalism we also recover the
Poincar\'e-Cartan form of
the Lagrangian formalism in an unambiguous way. Hence, the Lagrangian formalism for second-order field
theories is stated straightforwardly for the regular and singular (almost-regular) cases. In the same
way, we can obtain the associated Hamiltonian formalism in both cases using the unambiguously defined
Legendre map, and eventually a Hamiltonian section associated to the Lagrangian function.

\item
Despite what occurs in higher-order mechanics, the condition for the solutions to the field equations
to be holonomic is not guaranteed (even in the regular case), and neither can it be obtained from
the constraint algorithm. In higher-order field theory, this condition constitutes an additional
requirement of the theory.

\item Comparing our formulation with previous works found in the literature, we have that:

The unified formalism developed in \cite{art:Campos_DeLeon_Martin_Vankerschaver09} is different from
ours, since it uses  $J^2\pi \times_{J^1\pi} \Lambda^{m}_2(J^1\pi)$ as the extended jet-multimomentum
bundle, and, as pointed out in the introduction, some parameters appearing in the solutions of the
higher-order field equations (which are written in terms of sections and Ehresmann connections),
and in the definition of the Legendre map remain undetermined and must be fixed ``ad-hoc''. This does
not occur in our formalism; in fact, the constraint algorithm plays a crucial role in the determination
of all these arbitrary parameters. In addition,  in \cite{art:Campos_DeLeon_Martin_Vankerschaver09}
the theory is stated only in the unified setting, and the Lagrangian and Hamiltonian formalisms
are not explicitly recovered.

In \cite{art:Grabowska_Vitagliano14} the authors use a different approach to higher-order field
theories by means of a generalized version of Tulczyjew's triple, where the field equations are
obtained as Lagrangian submanifolds of the suitable extended phase spaces, and no explicit use is
made of Poincar\'{e}-Cartan forms.

Our formalism is also different from the unified formalism developed in \cite{art:Vitagliano10},
where infinite-order jet bundles are used, which are infinite-dimensional manifolds.

Another construction of a unique Poincar\'{e}-Cartan form for second-order classical field theories is made
in \cite{art:Kouranbaeva_Shkoller00} using purely variational methods, whereas that in this work this form
is derived using a Legendre transformation obtained by means of the constraint algorithm.

Finally, in \cite{art:Aldaya_Azcarraga78_2,art:Munoz84,art:Munoz85}
the authors make a more standard formulation of higher-order field theories generalizing both the Lagrangian
and Hamiltonian formalisms separately.

\item
In addition to analyzing the example of the loaded and clamped plate, we use this unified framework
to give a multisymplectic description of the KdV equation, which is also different from the standard
ones existing in the literature.
\end{itemize}

As further research, we intend to study the variational principles of second-order field theories
from this perspective.

In the main, we wish to apply this formalism to provide a multisymplectic description of the
Hilbert-Einstein theory of gravitation and other classical theories in theoretical physics. We believe
that this formalism will be useful for studying new reduction procedures of the corresponding field
equations, or for developing new numerical techniques of integration of these equations
using multisymplectic integrators.

This formulation fails when we try to generalize it to a classical field theory of order greater or
equal than $3$. The main obstruction is also the fundamental tool that we have used to obtain a unique
Legendre map from the constraint algorithm in the unified setting: the space of $2$-symmetric multimomenta.
In particular, the relation among the multimomentum coordinates that we have introduced in Section
\ref{sec:SymmetricMultimomenta}, $p^{ij}_\alpha = p_{\alpha}^{ji}$ for every
$1 \leqslant i,j \leqslant m$ and every $1 \leqslant \alpha \leqslant n$, can indeed be generalized
to higher-order field theories \cite{phd:Campos}. That is, we can generalize both the extended and
restricted $2$-symmetric multimomentum bundles to higher-order field theories. The main issue, however,
is that only the ``symmetric'' relation among the multimomentum coordinates holds for the highest-order
multimomenta. That is, this relation of symmetry on the multimomenta is not invariant under change
of coordinates for lower orders, and hence we do not obtain a submanifold of $\Lambda^m_2(J^{k-1}\pi)$.
A work to overcome this obstruction and to obtain a coordinate-free definition of a suitable Hamiltonian
phase space for classical field theories of order greater or equal than $3$ is nowadays in progress.

\vspace{10pt}
\appendix
\section{Multivector fields}
\label{sec:MultiVF}

(See \cite{art:Echeverria_Munoz_Roman98} for details).

Let $\mathcal{M}$ be a $n$-dimensional differentiable manifold. Sections of
$\Lambda^m(T \mathcal{M})$ are called $m$-\emph{multivector fields} in $\mathcal{M}$
(they are the contravariant skew-symmetric tensors of order $m$ in $\mathcal{M}$).
We denote the set of $m$-multivector fields in $\mathcal{M}$ by $\vf^m (\mathcal{M})$.

If $\mathcal{Y}\in\vf^m(\mathcal{M})$, for every $p\in \mathcal{M}$, there exists an open neighbourhood
 $U_p\subset \mathcal{M}$ and $Y_1,\ldots ,Y_r\in\vf (U_p)$ such that
$$\mathcal{Y}\vert_{U_p}=\sum_{1\leq i_1<\ldots <i_m\leq r} f^{i_1\ldots i_m}Y_{i_1}\wedge\ldots\wedge Y_{i_m} \, ,
$$
with $f^{i_1\ldots i_m} \in C^\infty (U_p)$ and $m \leqslant r\leqslant{\rm dim}\, \mathcal{M}$.
Then, $\mathcal{Y} \in \vf^m(\mathcal{M})$ is said to be \emph{locally decomposable} if,
for every $p\in \mathcal{M}$, there exists an open neighbourhood  $U_p\subset \mathcal{M}$
and $Y_1,\ldots ,Y_m\in\vf (U_p)$ such that $\mathcal{Y}\vert_{U_p}=Y_1\wedge\ldots\wedge Y_m$.

A non-vanishing $m$-multivector field $\mathcal{Y}\in\vf^m(\mathcal{M})$ and a $m$-dimensional
distribution $D\subset T\mathcal{M}$ are \emph{locally associated} if there exists a connected
open set $U\subseteq \mathcal{M}$ such that $\mathcal{Y}\vert_U$ is a section of $\Lambda^mD\vert_U$.
If $\mathcal{Y},\mathcal{Y}'\in\vf^m(\mathcal{M})$ are non-vanishing multivector fields locally
associated with the same distribution $D$, on the same connected open set $U$, then there exists a
non-vanishing function $f \in C^\infty (U)$ such that $\mathcal{Y}'\vert_U=f\mathcal{Y}\vert_U$. This
fact defines an equivalence relation in the set of non-vanishing $m$-multivector fields in
$\mathcal{M}$, whose equivalence classes will be denoted by $\{ \mathcal{Y}\}_U$. Then there is a
one-to-one correspondence between the set of $m$-dimensional orientable distributions $D$ in
$T \mathcal{M}$ and the set of the equivalence classes $\{ \mathcal{Y}\}_\mathcal{M}$ of
non-vanishing, locally decomposable $m$-multivector fields in $\mathcal{M}$.

If $\mathcal{Y}\in\vf^m(\mathcal{M})$ is non-vanishing and locally decomposable, and
$U\subseteq \mathcal{M}$ is a connected open set, the distribution associated with the class
$\{ \mathcal{Y}\}_U$ is denoted by $\mathcal{D}_U(\mathcal{Y})$. If $U=\mathcal{M}$ we write
$\mathcal{D}(\mathcal{Y})$.

A non-vanishing, locally decomposable multivector field $\mathcal{Y}\in\vf^m(\mathcal{M})$ is said
to be \emph{integrable} (resp. \emph{involutive}) if  its associated distribution
$\mathcal{D}_U(\mathcal{Y})$ is integrable (resp. involutive). Of course, if
$\mathcal{Y}\in\vf^m(\mathcal{M})$ is integrable (resp. involutive), then so is every other in
its equivalence class $\{ \mathcal{Y}\}$, and all of them have the same integral manifolds.
Moreover, \emph{Frobenius theorem} allows us to state that a non-vanishing and locally decomposable
multivector field is integrable  if, and only if, it is involutive. Nevertheless, in many applications
we have locally decomposable multivector fields $\mathcal{Y}\in\vf^m(\mathcal{M})$ which are not
integrable in $\mathcal{M}$,  but integrable in a submanifold of $\mathcal{M}$. A (local) algorithm
for finding this submanifold has been developed \cite{art:Echeverria_Munoz_Roman98}.

The particular situation in which we are interested is the study of multivector fields in fiber
bundles. If $\pi\colon \mathcal{M}\to M$ is a fiber bundle, we will be interested in the case where
the integral manifolds of integrable multivector fields in $\mathcal{M}$ are sections of $\pi$.
Thus, $\mathcal{Y}\in\vf^m(\mathcal{M})$ is said to be \emph{$\pi$-transverse} if, at every point
$y\in \mathcal{M}$, $(i (\mathcal{Y})(\pi^*\beta))_y\not= 0$, for every $\beta\in\Omega^m(M)$ with
$\beta (\pi(y))\not= 0$. Then, if $\mathcal{Y}\in\vf^m(\mathcal{M})$ is integrable, it is
$\pi$-transverse if, and only if, its integral manifolds are local sections of
$\pi\colon \mathcal{M}\to M$. In this case, if $\phi\colon U\subset M\to \mathcal{M}$ is a local
section with $\phi (x)=y$ and $\phi (U)$ is the integral manifold of $\mathcal{Y}$ through $y$,
then $T_y({\rm Im}\,\phi) = \mathcal{D}_y(\mathcal{Y})$.


\section*{Acknowledgments}

We acknowledge the financial support of \emph{Ministerio de Ciencia e Innovaci\'on} (Spain),
project MTM2011-22585, and \emph{Ministerio de Econom\'{\i}a y Competitividad} (Spain), project MTM2014-54855-P.
P.D. Prieto-Mart\'{\i}nez wants to thank the UPC for a Ph.D grant. We thank Mr. Jeff Palmer for his
assistance in preparing the English version of the manuscript. We thank the referees for their useful
comments and suggestions.


\end{document}